\documentclass[11pt]{article}

\usepackage{amsmath} 
\usepackage{amssymb}
\usepackage{amsthm}
\usepackage[ansinew]{inputenc}
\usepackage[T1]{fontenc}
\usepackage{fullpage}
\usepackage{tikz}
\newtheorem{theorem}{Theorem}[section]
\newtheorem{definition}[theorem]{Definition}
\newtheorem{lemma}[theorem]{Lemma}
\newtheorem{corollary}[theorem]{Corollary}
\newtheorem{proposition}[theorem]{Proposition}
\newtheorem{claim}[theorem]{Claim}

\newtheorem{example}{Example}[section]

\newcommand{\vone}{{\mathbf{1}}}
\newcommand{\vzero}{{\mathbf{0}}}

\def\nat{{\mathbb N}}
 \def\real{{\mathbb R}}
 \def\rat{{\mathbb Q}}

\begin{document}

\title{Greatest Fixed Points
of Probabilistic Min/Max\\ Polynomial Equations,
and Reachability for\\ Branching Markov Decision Processes} 
\author{Kousha Etessami\\U. of Edinburgh\\{\tt kousha@inf.ed.ac.uk}
\and 
Alistair Stewart\\U. of Southern California\\{\tt stewart.al@gmail.com} 
\and
Mihalis Yannakakis\\Columbia U.\\{\tt mihalis@cs.columbia.edu}}

\date{}

\maketitle

\begin{abstract}
We give polynomial time algorithms for quantitative
(and qualitative) {\em reachability} analysis for 
{\em Branching Markov Decision Processes} (BMDPs).
Specifically, given a BMDP, and given an initial population, where the
objective of the controller is to maximize (or minimize) the
probability of eventually reaching a population that contains an object of a desired
(or undesired) type,
we give algorithms for approximating the supremum
(infimum) reachability probability,
within desired precision $\epsilon > 0$, in time
polynomial in the encoding size of the BMDP and in $\log(1/\epsilon)$.  
We furthermore give P-time
algorithms for computing $\epsilon$-optimal strategies for both
maximization and minimization of reachability probabilities.  
We also give P-time algorithms for all associated {\em qualitative} analysis
problems, namely: deciding whether the optimal (supremum or infimum)
reachability probabilities are $0$ or $1$.
Prior to this paper, approximation of optimal reachability probabilities for BMDPs was not even known to be decidable.  

Our algorithms exploit the following basic fact: we show that for any
BMDP, its maximum (minimum) {\em non}-reachability probabilities are
given by the {\em greatest fixed point} (GFP) solution $g^* \in
[0,1]^n$ of a corresponding monotone max (min) Probabilistic
Polynomial System of equations (max/minPPS), $x=P(x)$, which are
the Bellman optimality equations for a BMDP with non-reachability
objectives.  
We show how to compute the GFP of max/minPPSs to desired precision
in P-time.

We also study more general {\em branching simple stochastic games} (BSSGs) 
with (non-)reachability objectives.   We show that: (1) the value of these
games is captured by the GFP, $g^*$, of a corresponding max-minPPS, $x=P(x)$;
(2) the {\em quantitative} problem of approximating the value is in TFNP; and 
(3)
the {\em qualitative} problems associated with the value are all solvable in P-time.

\end{abstract}

\section{Introduction}

{\em Multi-type branching processes} (BPs) are infinite-state purely stochastic processes that model the stochastic evolution of a population of entities of distinct types.  The BP specifies for every type a probability distribution for the offspring of 
entities of this type. Starting from an initial population, the process evolves from each generation to the next according to the probabilistic offspring rules.\footnote{Branching processes are used
both with discrete and with continuous time (where reproduction rules
for each type have associated rates instead of probabilities). 
However, the probabilities of extinction and reachability are not
time-dependent, and thus continuous-time processes can be studied via their corresponding discrete-time BPs, obtained by simply normalizing 
the rates on rules for each type.} Branching processes are a fundamental stochastic model with applications in many areas:
physics, biology, population genetics, medicine etc. 
{\em Branching Markov Decision Processes} (BMDPs) provide a natural extension of BPs
where the evolution is not purely stochastic but can be
potentially influenced or controlled to some extent: 
a controller can take actions which affect the probability distribution
for the set of offspring of the entities of each type. The goal is to design a policy for choosing the actions in order to optimize a desired objective. 

In recent years there has been great progress in
resolving algorithmic problems for BMDPs with the
objective of maximizing or minimizing the {\em extinction} probability,
i.e., the probability that the population eventually becomes extinct.
Polynomial time algorithms were developed for both maximizing and
minimizing BMDPs for {\em qualitative} analysis, i.e. to determine whether the
optimal extinction probability is 0, 1 or in-between \cite{rmdp},
and for {\em quantitative} analysis, to compute the 
optimal extinction probabilities to
any desired precision \cite{esy-icalp12}.
However, key problems related to
optimizing BMDP {\em reachability} probabilities (the probability
that the population eventually includes an entity having a target type)
have remained open.

Reachability objectives are very natural.
Some types may be undesirable, in which case we 
want to avoid them to the extent possible. Or conversely, we may want
to guide the process to reach certain desirable types.
For example, branching processes have been used recently
to  model cancer tumor  progression and multiple drug
resistance of tumors due to multiple mutations 
(\cite{Bozic13,RBCN13,KomBol13}).  
It could be fruitful to model the introduction
of multiple drugs (each of which controls/influences cells with a 
different type of mutation) via a ``controller'' that controls the offspring of different types, thus extending the current models (and associated
software tools)  which are based on BPs only, to controlled
models based on BMDPs.  A natural question
one could ask then is to compute the minimum
probability of reaching a {\em bad} (malignant) cell type,
and compute a drug introduction strategy that achieves 
(approximately) minimum probability.
Doing this efficiently (in P-time) 
would avoid the combinatorial explosion of trying all 
possible combinations of drug therapies.

In this paper we provide the first polynomial time algorithms 
for quantitative (and also qualitative) {\em reachability} analysis for BMDPs.
Specifically, we provide algorithms 
for $\epsilon$-approximating the supremum probability, 
as well as the infimum probability, of
reaching a given type (or a set of types) starting from an initial type (or 
an initial population of types), up to any desired additive error 
$\epsilon > 0$.  We also give algorithms for  
computing $\epsilon$-optimal strategies which achieve such 
$\epsilon$-optimal values.  
The running time of these algorithms (in the standard
Turing model of computation)  is 
polynomial in both the encoding size of the BMDP and 
in $\log(\frac{1}{\epsilon})$.
We also give P-time algorithms for the qualitative problems:
we determine whether the supremum or infimum probability
is 1 (or 0), and if so we actually compute an optimal
strategy that achieves $1$ ($0$, respectively).

In prior work \cite{rmdp}, we studied the
problem of optimizing extinction (a.k.a. termination) 
probabilities for  BMDPs,
and showed that the optimal extinction probabilities are
captured by the {\em least fixed point} (LFP) solution $q^* \in [0,1]^n$ of 
a corresponding system of multivariate monotone probabilistic
max (min) polynomial equations called maxPPSs (respectively minPPSs), 
which form the {\em Bellman optimality equations} for termination of
a BMDP.
A maxPPS is a system of equations $x=P(x)$ over a vector $x$ of variables, where the right-hand-side of each equation is of
the form $\max_j\{p_{j}(x)\}$, where each $p_j(x)$ is a polynomial with
non-negative coefficients (including the constant term) that sum to at most 1
(such a polynomial is called {\em probabilistic}). 
A minPPS is defined similarly.
In \cite{esy-icalp12}, we introduced an algorithm,
called {\em Generalized Newton's Method} (GNM),
for the solution of maxPPSs and minPPSs,
and showed that it computes the LFP of maxPPS and minPPS
(and hence also the optimal termination
probabilities for BMDPs) to desired precision in P-time. 
GNM is an iterative algorithm (like Newton's)
which in each iteration solves a suitable linear program 
(different ones for the max case and the min case).
In \cite{esy-icalp12} we also showed that for more general 
two player zero-sum {\em branching simple stochastic games} (BSSGs), 
with the player objectives of maximizing and minimizing the {\em extinction}
probability, we can approximate the value of the BSSG extinction game in TFNP.

In this paper we first model the reachability problem for a BMDP
by an appropriate system of equations: We show that the optimal 
{\em non-reachability} probabilities for a given BMDP are captured
by the {\em greatest fixed point} (GFP),  $g^* \in [0,1]^n$ of a corresponding 
maxPPS (or minPPS) system of Bellman equations.   
We then show that one can approximate 
the GFP solution $g^* \in [0,1]^n$
of a maxPPS (or minPPS), $x=P(x)$, 
in time  polynomial in both the encoding size $|P|$ of the system of
equations and in 
$\log(1/\epsilon)$, where $\epsilon > 0$ is the desired additive error 
bound of the solution.\footnote{It is worth mentioning that
it follows already from results in \cite{rmc} 
that the quantitative {\em decision}
problem for the GFP of a PPS (or max/minPPS) is {\tt PosSLP}-hard.
In other words,  
the problem of deciding whether $g^*_i \geq p$, for a given
probability $p \in [0,1]$, where $g^*$ is the GFP of a given 
PPS, is {\tt PosSLP}-hard.  This follows immediately 
from the proof in \cite{rmc}
(Theorem 5.3)
of the {\tt PosSLP}-hardness of
deciding whether $q^*_i \geq p$, where $q^*$ is the LFP of a given PPS
(equivalently, the termination probabilities of a given 1-exit RMC).
The PPS constructed in that proof is ``acyclic'' and
has a {\em unique} 
fixed point, and thus its LFP is equal to its GFP,  
i.e., $q^*=g^*$.
Thus, we can not hope to obtain a P-time algorithm
in the Turing model for deciding $g^*_i \geq p$, for a given
PPS (or max/minPPS), without
a major breakthrough in the complexity of numerical computation.}
  (The model of 
computation is the standard Turing machine model.)
We also show that the qualitative analysis of determining
the coordinates of the GFP that are 0 and 1, can be done in P-time
(and hence the same holds for the optimal reachability probabilities of BMDPs).

More generally, we study {\em branching simple stochastic games} (BSSGs) 
with (non-)reachability objectives.  These are two player zero-sum
turn based stochastic games, where one player wishes to reach a target type
while the other player wants to avoid that.
These games generalize BPs and BMDPs.
Such games can potentially be used to model adversarially 
some unknown parts of the controlled stochastic
model.  For example, in the setting suggested above for modeling
injection of different drugs in cancer tumors, there could
be some cell types whose offspring generation behavior 
in the presence of the drugs 
is unknown, and these cell types could be modeled in a worst-case fashion
as types in 
the BSSG that are controlled by the adversary, 
where the adversary aims to {\em maximize}
the probability of reaching the bad (malignant) cell types,
whereas the controller wants a drug injection strategy
for the controllable cell types in order to {\em minimize}
this probability.

We show that, firstly, the {\em value} of BSSG (non-)reachability games
(the value exists, i.e., these games
are {\em determined})
is captured by the GFP, $g^*$, of a corresponding 
max-minPPS, $x=P(x)$.   A max-minPPS is a system of equations $x_i = P_i(x)$,
where $P_i(x)$ has either the form $\max_j \{p_j(x)\}$ or the form
$\min_j \{p_j(x)\}$, where $p_j(x)$ are probabilistic polynomials.
We show that the 
{\em quantitative} problem of approximating the value of a BSSG,
or equivalently the GFP of a max-minPPS, is in TFNP.
We also show that  
the {\em qualitative} problems associated with deciding
whether the value 
of a BSSG is $0$ or $1$  (as well as computing
optimal strategies that ``achieve'' these values if one or the other is the case)  are all solvable in polynomial time.  This should be contrasted
with a result  in \cite{rmdp} which shows that,
for a given BSSG {\em extinction} game, the qualitative problem
of deciding whether the value is equal to $1$ 
\footnote{Equivalently, the problem of deciding whether the value is 1 for the {\em termination}
game on a 1-exit Recursive simple stochastic game (1-RSSG).}
is  at least as hard as Condon's long standing open problem of
computing the value of finite state simple stochastic games (or deciding
whether this value is, say, $\geq 1/2$).

Our P-time algorithms for computing the GFP of minPPSs and maxPPSs
to desired precision
make use of a variant of Generalized Newton Method (GNM), adapted for
the computation of the GFP instead of the LFP,
with a key important difference in the preprocessing step
before applying GNM. We first identify and remove only the variables that
have value $1$ in the GFP $g^*$  (we do not remove the
variables with value $0$, unlike the LFP case). 
We show that for maxPPSs, once these variables
are removed, the remaining system with GFP $g^* < \vone$ has a unique fixed point in
$[0,1]^n$,
hence the GFP is equal to the LFP; applying GNM starting from the all-$0$ initial vector
converges quickly (in P-time, with suitable rounding) to the GFP 
(by \cite{esy-icalp12}).
For minPPSs, even after the removal of the variables $x_i$ with $g^*_i=1$,
the remaining system may have multiple fixed points, and we can have LFP $<$ 
GFP.
Nevertheless, we show that with the subtle change in the preprocessing step,
GNM, starting at the all-$0$ vector, remarkably ``skips over''  the LFP 
and converges to the GFP solution $g^*$, in P-time.

We note incidentally that for any monotone operator $P$ from $[0,1]^n$ to itself,
one can define another monotone operator $R:[0,1]^n \rightarrow [0,1]^n$,
where $R(y)= \vone - P(\vone - y)$, such that the GFP $g^*$ of $x=P(x)$
and the LFP $r^*$ of $y=R(y)$ satisfy $g^*=\vone-r^*$. (The second system is
obtained from the first by the change of variables $y=\vone-x$.)
Simple {\em value iteration} starting at $\vzero$ ($\vone$) on $P(x)$ corresponds 
1-to-1 to value iteration starting at $\vone$ ($\vzero$, respectively) on $R(y)$.
However, this does not imply that computing the GFP of a max/minPPS  
is P-time reducible to computing the LFP of a max/minPPS: even if 
$x=P(x)$ is a PPS, the polynomials of $R(y)$  in general have negative coefficients. 
Value iteration on $R$ provably can converge exponentially slowly
(starting at $\vzero$ or $\vone$). 
Moreover, naively applying Newton starting at $\vzero$ to $y=R(y)$ can fail because the Jacobians are no longer non-negative, and the iterates need not even be defined (even after qualitative preprocessing).

Comparing the properties of the LFP and GFP of max/minPPSs, we note that
a difference for the qualitative problems is
that for the GFP, both the value=0 and the value=1 question depend only
on the structure of the model and not on its probabilities (the values of the
coefficients), whereas in the LFP
case the value=1 question depends on the probabilities
while value=0 does not (see \cite{rmc,rmdp}).

It is also worth noting that for BMDPs and BSSGs there is a natural ``duality''
between the objectives of optimizing reachability probability and that of
optimizing extinction probability.  
Namely, we can view a BMDP or BSSG as a random/controlled process
that generates a node-labeled (not necessarily finite) tree.
The objective of optimizing the extinction probability (i.e., the
probability of generating a finite tree), starting from a
given type, can equivalently be rephrased as a 
``{\em universal reachability}'' objective on a slightly modified BMDP, 
where the goal is to
optimize the probability of eventually reaching the 
target type (namely ``death'')
on {\em all}  paths starting at the root of the tree. 
Likewise, the ``universal reachability'' objective for any BMDP
can equivalently be rephrases as the objective of optimizing 
extinction probability on a slightly modified BMDP.
(We will explain these in more detail in Section \ref{sec:defs}.) 
By contrast, the
{\em reachability} objective that we study in this paper is precisely the
``{\em existential reachability}'' objective for BMDPs and BSSGs, 
namely optimizing
the probability of reaching the target type on {\em some} path in the
generated tree.

We shall see that, despite this duality, there 
are some important differences between these two objectives,
in particular when 
it comes to the existence of optimal strategies.
Namely, we show that, unlike optimization of
extinction (termination) probabilities for BMDPs, for which there
always exists a static deterministic optimal strategy (\cite{rmdp}),
there need not exist any optimal strategy at all for maximizing reachability 
probability in a BMDP, i.e., the supremum probability may not
be attainable. If the supremum probability is 1 however (and likewise 
if the value of the BSSG game is 1),
we show that there does exist a strategy (for the player maximizing
reachability probability) that achieves it, although not necessarily any 
static strategy.
For the objective of minimizing reachability probability, 
we show there always exists an optimal deterministic 
and static strategy, both in BMDPs and BSSGs.
Regardless of what the optimal value is, we show that we can compute in P-time 
an $\epsilon$-optimal static (possibly randomized) policy,
for both maximizing and minimizing reachability probability in 
a BMDP.

\medskip
{\bf Related work:}
BMDPs have been previously studied in both operations research (e.g., \cite{pliska76,rotwhit82,denrot05a}) and computer science (e.g., \cite{rmdp,EGKS08,EWY08}).
We have already mentioned the results in \cite{rmdp,esy-icalp12}
concerning the computation of the extinction probabilities of BMDPs and
the computation of the LFP of max/minPPS. 
Branching processes are closely connected to stochastic context-free grammars,
1-exit Recursive Markov chains (1-RMC) \cite{rmc}, and the corresponding 
stateless probabilistic pushdown processes, pBPA \cite{EKM}; their extinction 
or termination
probabilities are interreducible, and they are all captured by the LFP of PPSs.
The same is true for their controlled extensions, for example
the extinction probability of BMDPs and the termination probabilities of 1-exit
Recursive Markov Decision processes (1-RMDP) \cite{rmdp}, are both captured by
the LFP of maxPPS or minPPS. A different type of objective of optimizing the
total expected reward for 1-RMDPs (and equivalently BMDPs) in a setting with 
positive rewards was studied in \cite{EWY08}; in this case the 
optimal values are rational and can be computed exactly in P-time.

The equivalence between BMDPs and 1-RMDPs however does not carry over
to the reachability objective. The {\em qualitative} reachability
problem for 1-RMDPs (equivalently BPA MDPs) and the extension to simple
2-person games 1-RSSGs (BPA games) were studied in \cite{BBFK08} and
\cite{BBKO11} by Brazdil et al.  It is shown in \cite{BBFK08} that
qualitative {\em almost-sure} reachability for 1-RMDPs can be decided
in P-time (both for maximizing and minimizing 1-RMDPs).  However, for
maximizing reachability probability, almost-sure and limit-sure
reachability are {\em not} the same: in other words, the supremum
reachability probability can be $1$, but it may not be achieved by any
strategy for the 1-RMDP.  By contrast, for BMDPs we show that if the
supremum reachability probability is 1, then there is a strategy that
achieves it. This is one illustration of the fact that the equivalence
between 1-RMDP and BMDP does not hold for the reachability objective.
The papers \cite{BBFK08,BBKO11} do not address the limit-sure
reachability problem, and in fact even the decidability of limit-sure
reachability for 1-RMDPs remains open.

Chen et. al. \cite{CDK12} studied model checking of branching
processes with respect to properties expressed by
deterministic parity tree automata and showed that the qualitative
problem is in P (hence this holds in particular for reachability
probability in BPs), and that the quantitative problem of comparing the
probability with a rational is in PSPACE. Although not explicitly
stated there, one can use Lemma 20 of \cite{CDK12} and our algorithm
from \cite{ESY12} to show that the reachability probabilities of BPs
can be approximated in P-time. Bonnet et. al. \cite{BKL14} studied a model of 
``probabilistic Basic Parallel Processes'', which are syntactically
close to Branching processes, except reproduction is
asynchronous and the entity that reproduces in each step is chosen
randomly (or by a scheduler/controller). None of the previous results have 
a
bearing on the reachability problems for BMDPs.

\medskip
{\bf Organization of the paper:}
Section \ref{sec:defs} gives basic definitions and background. 
Section \ref{sec:gfp} characterizes the (non-)reachability
problem for BMDPs, and more general BSSGs, in terms of 
the GFP computation problem for max-minPPS equations,
and discusses the existence of optimal strategies for BMDPs. 
Section \ref{sec:qual-gfp-equal-1} gives a P-time algorithm for 
determining those variables with value $=$ 1 in the GFP of a max-minPPS. 
Section \ref{sec:bp-ld} analyses the GFP of PPSs, and shows we can approximate it in P-time. 
Section \ref{sec:maxpps-gfp-approx} solves the GFP value approximation problem for maxPPSs in P-time,
and also shows how to compute an $\epsilon$-optimal deterministic static strategy for maxPPS in P-time. 
Section \ref{sec:minPPS-gfp-approx} solves the GFP value approximation problem for minPPSs in P-time. 
Section \ref{sec:eps-opt-minpps} concerns the construction, in P-time, of $\epsilon$-optimal 
strategies
for the GFP of a minPPS (this is substantially harder than the maxPPS case).
Section \ref{sec:zero-detect} gives a P-time algorithm for 
determining those variables with value $=$ 0 in the GFP of a max-minPPS (this is substantially
harder than the $=$ 1 case done in Section \ref{sec:qual-gfp-equal-1}).
Section \ref{sec:bssg} shows that we can approximate the value of a BSSG (non-)reachability game,
and the GFP of a max-minPPS, in TFNP.

\section{Definitions and Background}

\label{sec:defs}

We start by providing unified definitions of
multi-type Branching processes (BPs), Branching MDPs (BMDPs), and 
Branching Simple
Stochastic Games (BSSGs).
Although most of our results are focused on BMDPs,
since BSSGs provide the most general of these models we start by 
defining BSSGs, and then specializing them to obtain BMDPs and BPs.
Throughout we use $\vzero$ and $\vone$ to denote
all-$0$ and all-$1$ vectors, respectively, of the appropriate dimensions.

A {\bf Branching Simple Stochastic Game} (BSSG),
consists of a finite set 
$V=\{T_1, \ldots, T_n\}$ of types, a finite non-empty set 
$A_i \subseteq \Sigma$ 
of actions for each type $T_i$ ($\Sigma$ is some finite action alphabet),
and a finite set $R(T_i,a)$ of probabilistic rules associated
with each pair $(T_i,a)$, $i \in [n]$, where $a \in A_i$. 
Each rule $r \in  R(T_i,a)$ 
is a triple $(T_i,p_r, \alpha_r)$, which we denote by
$T_i \stackrel{p_r}{\rightarrow} \alpha_r$, 
where $\alpha_r \subseteq \nat^n$ 
is a $n$-vector of natural numbers that denotes a finite multi-set 
over the set $V$, and where $p_r \in (0,1] \cap \rat$ is the probability of the rule $r$ (which we assume is given by a rational number, for computational purposes),
where we assume that for all $i \in V$ and $a \in A_i$, 
the rule probabilities  in $R(T_i,a)$ sum to $1$, i.e.,  
$\sum_{r \in R(T_i,a) } p_r = 1$. 
For BSSGs, the types are partitioned into two sets: 
$V = V_{\max} \cup V_{\min}$,   $V_{\max} \cap V_{\min} = \emptyset$,
where $V_{\max}$ contains those types ``belonging'' to player $\max$,
and  $V_{\min}$ containing those
belonging to player $\min$.

A {\bf Branching Markov Decision Process} (BMDP) is a BSSG where one
of the two sets $V_{\max}$ or $V_{\min}$ is empty.  Intuitively, a BMDP
(BSSG) describes the stochastic evolution of a population of entities
of different types in the presence of a controller (or two players)
that can influence the evolution.  We can define a 
{\bf multi-type Branching
  Process} (BP), by imposing a further restriction,
namely that all action sets $A_i$ must be singleton sets.  
Hence in a BP, players have no choice of actions, and we can 
simply assume players don't exist: a BP defines a purely stochastic
process.

A play (or trajectory) of a BSSG operates as follows:
starting from an
initial population (i.e., set of entities of given types) $X_0$ at 
time (generation) 0,
a sequence of populations $X_1, X_2, \ldots$ is generated, where $X_{k+1}$ is obtained from
$X_{k}$ as follows. 
Player $\max$ ($\min$) selects for each entity $e$ in set 
$X_{k}$ that belongs to $\max$ (to $\min$, respectively) 
an available action $a \in A_i$ for the type $T_i$ of entity $e$; 
then for each such entity $e$ in $X_{k}$
a rule $r \in R(T_i,a)$ is chosen randomly and independently 
according to the rule probabilities $p_r$,
where $a \in A_i$ is the action selected for that particular entity $e$. 
Every entity is then replaced by a set of entities with the types 
specified by the right-hand side multiset $\alpha_r$ of that chosen rule $r$.
The process is repeated as long as the current population $X_k$ is nonempty, 
and it is said to {\em terminate} (or become {\em extinct})  if there is
some $k \geq 0$ such that 
$X_k = \emptyset$.  When there are $n$ types, we can view a population $X_i$ as a vector $X_i \in \nat^n$, specifying the number of objects of each type.
We say that the process {\em reaches} a type $T_j$, if there is some $k \geq 0$
such that $(X_k)_j > 0$.

We can consider different objectives by the players.
For example, in \cite{rmdp,esy-icalp12} the 
objective considered was that
the two players wish to maximize and minimize, respectively,
the probability of termination (i.e., extinction of
the population). It was shown in \cite{rmdp} that such
BSSG games indeed have a value, and in \cite{esy-icalp12}
a P-time algorithm was developed for approximating this
value in the case of max-BMDPs and min-BMDPs with
the termination objective.

In this paper we consider the {\em reachability} objective:  namely where the
goal of the two players, starting from a given population,
is to maximize/minimize the probability of reaching a population
which contains {\em at least} one entity of a given special type, $T_{f^*}$.
It is perhaps not immediately clear that a BSSG with such a reachability
objective has a {\em value}, but we shall show that this is indeed the case.

Regarding strategies,
at each stage, $k$, each player is allowed, in principle, to select the actions for the entities in $X_k$ that belong to it 
based on the whole past history, may use randomization (a mixed strategy), and may make different choices 
for entities of the same type. 
The ``history'' of the process 
up to time $k-1$ includes not only the populations $X_0, X_1, \ldots, X_{k-1}$,
but also the information on all the past actions
and rules applied and the parent-child
relationships between all the entities up to the generation $X_{k-1}$.
The history can be represented by a forest of depth $k-1$, with internal nodes
labelled by rules and actions, and whose leaves at level $k-1$
form the population $X_{k-1}$.
Thus, a strategy of a player is a function that maps every finite history
(i.e., labelled forest of some finite depth as above)
to a probability distribution on the set of tuples of actions
for the entities in the current population
(i.e. at the bottom level of the  forest) that are controlled by that player. 
Let $\Psi_1, \Psi_2$ be the set of all strategies of players 1, 2.
We say that a strategy is {\em deterministic} if for every history it chooses 
one tuple of actions with probability 1.  We say that a 
strategy is 
{\em static} if 
for each type $T_i$ controlled by that player the strategy always chooses 
the same action $a_i$, or the same probability distribution on actions,
for all entities of type $T_i$ in all histories.\footnote{In \cite{esy-icalp12} we called 
a strategy ``static'' if it was both deterministic and static.  In this paper
we will refer to these as ``deterministic static'' strategies, because
we will also need ``randomized static'' strategies, and want to differentiate
between them.} 
Our notion of an arbitrary strategy is quite general (it can depend on 
 all the details of the entire history, and be randomized, etc.).
However,  it was shown  in \cite{rmdp}
that for the objective of optimizing extinction probability,
both players have optimal static strategies in BSSGs.
We shall see that this is not the case for BMDPs or BSSGs with
the reachability objective.  

Let us now observe, as mentioned in the Introduction, 
a natural ``duality'' between the objective of
optimizing extinction probability and that of optimizing
reachability probability.
A BMDP or BSSG can also be viewed as a random/controlled
process for generating a node-labeled, not necessarily finite, tree
(or a forest, in case the process is started with a population larger than 1).
The nodes of the tree denote objects, nodes are labeled by their type,
and the edges in the tree denote the parent-child relationships:
when a rule $T_i \rightarrow \alpha_r$ is applied to some node $v$ of
type $T_i$ in the tree, the children of node $v$ will be  
in 1-1 correspondence with the multi-set of types given by $\alpha_r$.
For a given BSSG, optimizing the extinction probability (i.e., the
probability of generating a finite tree), starting from an object of a
given type, can be rephrased as a ``{\em universal reachability}''
objective on a slightly modified BSSG, where the objective is to
optimize the probability of eventually reaching a target type
on {\em all}  paths starting at the root of
the generated tree.
Specifically, the target type is a newly introduced type, 
called ${\mathtt{death}}$, and for all types $T_i$, 
every rule $T_i \rightarrow \emptyset$
in the original BSSG 
is replaced by the rule $T_i \rightarrow {\mathtt{death}}$
in the modified BSSG (with the same probability).
Likewise, the ``universal reachability''
objective for any BSSG can be rephrased  as the objective of optimizing 
extinction probability in a slightly modified BSSG.
Namely, for all types $T_i$,  every rule  
$T_i \rightarrow \alpha_r$ in
the original BSSG,  where the multiset $\alpha_r$ is nonempty, is replaced by the rule $T_i \rightarrow \alpha'_r$
(with the same probability) in the revised BSSG,
where the multiset $\alpha'_r$ is the same as 
$\alpha_r$ except that all copies of
the target type have been removed from $\alpha'_r$;
moreover for any non-target type $T_i$,  a rule in the original BSSG of the form
$T_i \rightarrow \emptyset$ is replaced by the rule $T_i \rightarrow 
{\mathtt{dead}}$
(with the same probability) in the revised BSSG, where ${\mathtt{dead}}$
is a new type having only one associated rule: ${\mathtt{dead}} \rightarrow {\mathtt{dead}}$, with probability 1.

By contrast, the
{\em reachability} problem that we study in this paper is precisely the
``{\em existential reachability}'' objective for BMDPs, namely optimizing
the probability of reaching the target type on {\em some} path in the
generated tree.  

Let us now consider in more detail the (non-)reachability objective.
 For a given initial population $\mu \in \nat^n$, with $(\mu)_{f^*} = 0$,
and given
 integer $k \geq 0$, and strategies $\sigma \in \Psi_1$,
 $\tau \in \Psi_2$, we denote by $g^{k}_{\sigma,\tau}(\mu)$
 the probability that the process with initial population $\mu$,
 and strategies $\sigma, \tau$ does {\em not} reach a 
population with an object of type $T_{f^*}$ in at most $k$ steps.
In other words, this is the probability that for all $0 \leq d \leq k$,
we have $(X_d)_{f^*} = 0$.
Let us denote by $g^{*}_{\sigma,\tau}(\mu)$ the probability
that $(X_d)_{f^*} = 0$ for all $d \geq 0$.

We let $g^k(\mu) = \sup_{\sigma \in \Psi_1} \inf_{\tau \in \Psi_2} g^{k}_{\sigma,\tau}(\mu)$,
 and $g^*(\mu) = \sup_{\sigma \in \Psi_1} \inf_{\tau \in \Psi_2} g^{*}_{\sigma,\tau}(\mu)$;
 the last quantity is the {\em value} of the non-reachability game 
for the initial population $\mu$.
Likewise $g^k(\mu)$ is the value of the $k$-step non-reachability game.
We will show that  determinacy holds for these games, 
 i.e. $g^*(\mu) = \sup_{\sigma \in \Psi_1} \inf_{\tau \in \Psi_2} g^{*}_{\sigma,\tau}(\mu)
 = \inf_{\tau \in \Psi_2} \sup_{\sigma \in \Psi_1} g^{*}_{\sigma,\tau}(\mu)$,
 and similarly for $g^k(\mu)$.
However, unlike the case for extinction probabilities (\cite{rmdp}),
it does {\em not} hold that
both players have optimal static strategies.

If $\mu$ has a single entity of type $T_i$, we will write
 $g^*_i$ and $g^k_i$ instead of $g^*(\mu)$ and $g^k(\mu)$.
 
Given a BMDP (or BSSG), the goal is to compute the vector $g^*$
 of the $g^*_i$'s, i.e. the vector of non-reachability values
 of the different types.
 As we will see, from the $g^*_i$'s, we can compute
 the value $g^*(\mu)$ for any initial population $\mu$,
 namely $g^*(\mu) = f(g^*,\mu) :=  \Pi_i (g^*_i)^{\mu_i}$.
 The vector of reachability values $r^*$ is of course $r^* = \vone -g^*$,
 where $\vone$ is the all-1 vector; the reachability value
 for initial population $\mu$ is $r^*(\mu) = \vone- g^*(\mu)$.

We shall associate a system of min/max probabilistic polynomial Bellman
equations, $x=P(x)$, to each given BMDP 
or BSSG, that contains one variable $x_i$ and one equation
$x_i = P_i(x)$ for each type $T_i$,
 such that the vector $g^*$ of values of the BSSG non-reachability 
game for the different starting types is given by
the greatest fixed point (GFP) solution of $x=P(x)$ in $[0,1]^n$.
We need some notation first in order to introduce these Bellman equations.

For an $n$-vector of variables $x = (x_1,\ldots,x_n)$, and 
a vector $v \in \nat^n$,  we use the shorthand notation $x^v$ to denote
the monomial $x_1^{v_1}\ldots x^{v_n}_{n}$.
Let $\langle \alpha_r \in \nat^n \mid r \in R \rangle$ be a finite set
of $n$-vectors of natural numbers, indexed by the set $R$.
Consider a multi-variate polynomial $P_i(x) = \sum_{r \in R} p_r x^{\alpha_r}$,
for some rational-valued coefficients $p_r$, $r \in R$.
We shall call $P_i(x)$ a {\bf\em probabilistic polynomial} if  
$p_r \geq 0$ for all $r \in R$, and $\sum_{r \in R} p_r \leq 1$.

\begin{definition} A {\bf probabilistic polynomial system of equations}, 
$x = P(x)$, which we shall call
a  {\bf PPS},  
is a system of $n$ equations, $x_i = P_i(x)$,  in $n$ variables $x = (x_1,x_2,\ldots,x_n)$,
  where for all $i \in \{ 1,2,\ldots,n\}$,  $P_i(x)$ is a probabilistic polynomial.

A {\bf maximum-minimum probabilistic polynomial system of equations}, $x=P(x)$,
called a  {\bf max-minPPS} 
is a system of $n$ equations in $n$ variables $x= (x_1,x_2,\ldots,x_n)$,  
where for all $i \in \{1,2,\ldots,n\}$, either: 
\begin{itemize}
\item {\tt Max-polynomial}: $P_i(x) = \max  \{q_{i,j}(x): j \in \{1,\ldots,m_i\}\}$,   Or: 

\item {\tt Min-polynomial}: $P_i(x) = \min  \{q_{i,j}(x): j \in \{1,\ldots,m_i\}\}  \ $
\end{itemize}

\noindent where each $q_{i,j}(x)$ is a probabilistic polynomial, for every 
$j \in \{1,\ldots,m_i\}$.

We shall call such a system a {\bf maxPPS} (respectively, a {\bf minPPS})
if for every $i \in \{1,\ldots,n\}$, $P_i(x)$ is a {\tt Max-polynomial} (respectively, a {\tt Min-polynomial}).

Note that we can view a PPS in n variables 
as a maxPPS, or as a minPPS, where $m_i =1$ for every $i \in \{1,\ldots,n\}$.
\end{definition}

For computational purposes we assume that all the coefficients are rational.
We assume that the polynomials in a system are given in sparse form,
i.e., by listing only the nonzero terms, with the coefficient and the nonzero exponents 
of each term given in binary. 
We let $|P|$ denote the total bit encoding length of a system $x=P(x)$
under this representation.

We use  {\bf max/minPPS} to refer to a system of equations, $x=P(x)$, 
that is either a maxPPS or a minPPS. 
We refer to systems of equations containing both max and min equations
as {\bf max-minPPSs}.

It was shown in \cite{rmdp} that any max-minPPS, $x=P(x)$, 
has a {\bf least fixed point} ({\bf LFP}) 
solution, $q^* \in [0,1]^n$, i.e.,  $q^* = P(q^*)$ and if $q = P(q)$ 
for some $q \in [0,1]^n$ then $q^* \leq q$
(coordinate-wise inequality).  
In fact, $q^*$ corresponds to the vector of {\em values} of
a corresponding {\em Branching Simple Stochastic Game} with 
the objective of {\em extinction}, starting at each type.
As observed in \cite{rmc, rmdp},  $q^*$ may in general
contain irrational values, even in the case of pure PPSs
(and the corresponding multi-type Branching process).

In this paper we shall observe that any
max-minPPS, $x=P(x)$, 
also has a {\bf greatest fixed point} ({\bf GFP}) 
solution, $g^* \in [0,1]^n$, i.e.,  such that
$g^* = P(g^*)$ and if $q = P(q)$ 
for some $q \in [0,1]^n$ then $q \leq g^*$
(coordinate-wise inequality).  
In fact, in this case $g^*$ corresponds to the vector of {\em values} of
a corresponding {\em branching simple stochastic game} where
the objective of the two players is
to maximize/minimize the probability of
{\em \underline{not} reaching an undesired type (or set of types)} 
starting at each type.
Again, $g^*$ may contain irrational coordinates, so we in
general want to approximate its coordinates (and the coordinates
of $(\mathbf{1}- g^*)$ which constitute {\em reachability} values)
to desired precision.    For a countable set $S$, let $\Delta(S)$
denote the set of probability distributions on $S$, i.e., the set of
functions $f:S \rightarrow [0,1]$ such that $\sum_{s \in S} f(s) = 1$.

\begin{definition} We define a (possibly \underline{randomized}) {\bf policy} 
for max (min) in a max-minPPS, $x=P(x)$, to be a function 
$\sigma: \{1,\ldots,n\} \rightarrow \Delta(\mathbb{N})$
that assigns a probability distribution to each variable $x_i$
for which $P_i(x)$ is a max- (respectively, min-) polynomial, 
such that the support of $\sigma(i)$ is
a subset of $\{1,\ldots, m_i\}$, the possible 
$m_i = |A_i|$ different actions (i.e., choices of polynomials) available 
in $P_i(x)$.
\end{definition}

Intuitively, policies are akin to static strategies for BMDPs and BSSGs.
For each variable, $x_i$, a policy 
selects 
a probability
distribution over the probabilistic polynomials, $q_{i,\sigma(i)}(x)$,
that appear on the RHS of the equation $x_i = P_i(x)$, and
which $P_i(x)$ is the maximum/minimum over.

\begin{definition} 
For a max-minPPS, $x=P(x)$, and policies $\sigma$ and $\tau$
for the $\max$ and $\min$ players, respectively,
we write $x=P_{\sigma,\tau}(x)$ for the
PPS obtained by fixing both these policies.
 We write
$x=P_{\sigma,*}(x)$ for the minPPS obtained by fixing $\sigma$ for the 
max player, and
$x=P_{*,\tau}(x)$ for the maxPPS obtained by fixing $\tau$ for the min player. 
More specifically, note that for  
policy $\sigma$ for player $\max$,
we define the minPPS $x=P_{\sigma,*}(x)$
by $(P_{\sigma,*})_i(x) := \sum_{a \in A_i} \sigma(i)(a) \cdot q_{i,a}$,
for all $i$ that belong to player $\max$, and 
otherwise $(P_{\sigma,*})_i(x) := P_i(x)$.
We similarly define $x=P_{*,\tau}(x)$ and $x=P_{\sigma,\tau}(x)$.

For a maxPPS (or minPPS), $x=P(x)$, and policy $\sigma$ for
$\max$  (for $\min$),  we shall use the abbreviated notation
$x=P_\sigma(x)$   instead of $x=P_{\sigma,*}(x)$
(instead of $x=P_{*,\sigma}(x)$, respectively).

For a max-minPPS, $x =P(x)$, and a (possibly randomized) policy, $\sigma$
for $\max$,
we use $q^*_{\sigma,*}$ and $g^*_{\sigma,*}$ to denote the LFP and GFP solution vectors
for the corresponding minPPS $x=P_{\sigma,*}(x)$, respectively.
Likewise we use $q^*_{*,\tau}$ and $g^*_{*,\tau}$ to define the LFP and 
GFP solutions
of the maxPPS $x=P_{*,\tau}(x)$.
Similarly, for a maxPPS (or minPPS), $x=P(x)$, and a policy, $\sigma$,
we use $q^*_\sigma$ and $g^*_\sigma$ to denote the LFP and GFP of $x=P_\sigma(x)$.
\end{definition}

\begin{definition} For a max-minPPS, $x=P(x)$, a policy $\sigma^*$ 
is called {\bf optimal} for $\max$ 
for the LFP (respectively, the GFP) if 
$q^*_{\sigma^*,*} = q^*$  (respectively $g^*_{\sigma^*,*} = g^*$).

An optimal policy $\tau^*$ for $\min$ for the LFP and GFP, respectively,
is defined similarly. 

For $\epsilon > 0$, a policy $\sigma'$ 
for $\max$ is called {\bf $\epsilon$-optimal} for the LFP 
(respectively GFP) , if
$||q^*_{\sigma',*} - q^*||_{\infty} \leq \epsilon$
(respectively, $||g^*_{\sigma',*} -g^*||_{\infty} \leq \epsilon$).
An $\epsilon$-optimal policy $\tau'$ for $\min$ is defined
similarly.
\end{definition}

It is convenient to put max-minPPSs in the following simple form.

\begin{definition} A max-minPPS in {\bf simple normal form (SNF)}, $x= P(x)$, 
is a system of $n$ equations in $n$ variables $x_1,x_2,\ldots,x_n$  where each $P_i(x)$ for $i = 1,2,\ldots,n$ 
is in one of three forms:
\begin{itemize}
\item {\tt Form L}: $P_i(x) = a_{i,0} + \sum_{j=1}^n a_{i,j} x_j$,
where  $a_{i,j} \geq 0$ for all $j$, and  such that $\sum_{j=0}^n a_{i,j} \leq 1$
\item {\tt Form Q}:  $P_i(x) = x_j x_k$ for some $j,k$
\item {\tt Form M}: $P_i(x) =\max  \{ x_j ,  x_k \}$ 
or $P_i(x) = \min  \{x_j, x_k \}$, for some $j,k$;\\
we sometimes differentiate these two cases as {\tt Form M$_{\max}$} and 
{\tt M$_{\min}$}, respectively.
\end{itemize}
We define {\bf SNF form} for max/minPPSs analogously: only the definition of ``{\tt Form M}'' changes
(restricting to $\max$ or $\min$, respectively).
\end{definition}

In the setting of a max-minPPSs in SNF form, we will often say that a variable 
has form or type L, Q, or M, to mean
that $P_i(x)$ has the corresponding form.
Also, for simplicity in notation, when 
we talk about a deterministic policy, if $P_i(x)$ has form $M$, 
say $P_i(x) \equiv \max \{ x_j, x_k\}$,
then when it is clear from the context we will use 
$\sigma(i) = k$ to mean that the 
policy $\sigma$ chooses $x_k$ among the two choices $x_j$ and $x_k$ 
available in $P_i(x) \equiv \max \{x_j,x_k\}$.

\begin{proposition}[cf. Proposition 7.3 \cite{rmc}]
\label{prop:snf-form}
Every max-minPPS, $x = P(x)$, can be transformed in P-time
to an ``{\em equivalent}''  max-minPPS,  $y=Q(y)$  in SNF form,
such that  $|Q| \in O( |P|  )$.
More precisely, the variables $x$ are 
a subset of the variables $y$, and both the LFP 
and GFP of  $x = P(x)$
are, respectively, the projection of the LFP and GFP of  $y=Q(y)$,
onto the variables $x$, and furthermore
an optimal policy (respectively, $\epsilon$-optimal
policy) for the LFP (respectively, GFP) of 
 $x = P(x)$ can be obtained in P-time from an
optimal (resp., $\epsilon$-optimal) policy 
for the LFP (respectively, GFP) of $y=Q(y)$.
\end{proposition}
\begin{proof}
We can easily convert, in P-time, any max-minPPS  into SNF form, 
using the following procedure.
\begin{itemize}
\item For each equation $x_i = P_i(x) = \text{max } \{p_1(x), \ldots ,p_m(x)\}$, 
for each $p_j(x)$ on the right-hand-side that is not a variable, add a new variable 
$x_k$, replace $p_j(x)$ with $x_k$ in $P_i(x)$, 
and add the  new equation $x_k = p_j(x)$. 
Do similarly if $P_i(x) = \min \{p_1(x), \ldots, p_m(x) \}$.

\item If $P_i(x) = \text{max } \{x_{j_1},\ldots,x_{j_m}\}$ with $m >2$, then add $m-2$
new variables $x_{i_1}, \ldots, x_{i_{m-2}}$, 
set $P_i(x) =  \text{max } \{x_{j_1},x_{i_1}\}$,
and add the equations $x_{i_1} = \text{max } \{x_{j_2},x_{i_2}\}$, 
$x_{i_2} = \text{max } \{x_{j_3},x_{i_3}\}$, $\ldots$, 
$x_{i_{m-2}} = \text{max } \{x_{j_{m-1}},x_{j_m}\}$. 
Do similarly if $P_i(x) = \min \{x_{j_1},ldots,x_{j_m}\}$ with $m >2$.

\item For each equation $x_i =P_i(x) = \sum_{j=1}^m p_j x^{\alpha_j}$,
where $P_i(x)$ is a probabilistic polynomial that is not just a 
constant or a single monomial, 
replace every (non-constant) monomial $x^{\alpha_j}$ on the right-hand-side
that is not a single variable by a new variable $x_{i_j}$ 
and add the equation  $x_{i_j}=x^{\alpha_j}$.

\item For each variable $x_i$ that occurs in some polynomial with
exponent higher than 1, introduce new variables 
$x_{i_1}, \ldots, x_{i_k}$ where $k$ is the logarithm of
the highest exponent of $x_i$ that occurs in $P(x)$,
and add equations $x_{i_1} = x_i^2$,  $x_{i_2} = x_{i_1}^2$, $\ldots$,
$x_{i_k} = x_{i_{k-1}}^2$.
For every occurrence of a higher power $x_i^l$, $l>1$, of $x_i$ in $P(x)$,
if the binary representation of the exponent $l$ is $a_k \dots a_2 a_1 a_0$,
then we replace $x_i^l$ by the product of the variables $x_{i_j}$  such that the
corresponding bit $a_j$ is 1, and $x_i$ if $a_0=1$.
After we perform this replacement for all the higher powers of
all the variables, every polynomial of total degree >2 is just a product of variables.

\item If a polynomial $P_i(x) = x_{j_1} \cdots x_{j_m}$ in the current system
is the product of $m>2$ variables, then add $m-2$
new variables $x_{i_1}, \ldots, x_{i_{m-2}}$, 
set $P_i(x) = x_{j_1}x_{i_1}$,
and add the equations $x_{i_1} = x_{j_2} x_{i_2}$, 
$x_{i_2} = x_{j_3} x_{i_3}$, $\ldots$, 
$x_{i_{m-2}} = x_{j_{m-1}} x_{j_m}$. 
\end{itemize}
Now all equations are of the form L, Q, or M.

The above procedure allows us to 
convert any max-minPPS into one in SNF form by introducing $O(|P|)$ new variables and 
blowing up the size of $P$ by a constant factor $O(1)$. 
It is clear that both the LFP and the GFP of $x=P(x)$ arise as the projections
of the LFP and GFP of $y=Q(y)$ onto the $x$ variables.
Furthermore, there is an obvious (and easy to compute) bijection between  
policies for the resulting SNF form max-minPPS and the original max-minPPS.
  \end{proof}

Thus from now on, and for the
rest of this paper {\em we may assume, without loss of generality, that all max-minPPSs are in  SNF normal form. }

A non-trivial fact established in \cite{rmdp} is that
for the LFP of a max-minPPS,
both players always  have an optimal deterministic policy:

\begin{theorem}[\cite{rmdp}, Theorem 2] \label{optexist} 
For any max-minPPS, $x=P(x)$,  
for both the maximizing and minimizing player
there always exists an optimal {\em deterministic} policy,
for the LFP.
\end{theorem}
 
As we shall show,
while in general for a max-minPPS $x=P(x)$ there
does exist an optimal deterministic policy $\sigma^*$
for the {\em maximizing} player,
for the GFP,
in general there does {\em not} exist
any optimal policy at all for 
the {\em minimizing} player for the GFP of a 
minPPS $x=P(x)$.

Nevertheless, we shall show that for any $\epsilon > 0$,
there always exists an $\epsilon$-optimal {\em randomized} policy
for the GFP for the minimizing player in 
any max-minPPS. Furthermore, we shall show how to compute such a
policy in P-time for minPPS.

\begin{definition}
The {\bf dependency graph} of a max-min PPS $x=P(x)$,
is a directed graph that has one node for each variable $x_i$,
and contains an edge $(x_i,x_j)$ if $x_j$ appears in $P_i(x)$.
The dependency graph of a BSSG has one node for each type,
and contains an edge $(T_i,T_j)$ if there is an action $a \in A_i$
and a rule $T_i \stackrel{p_r}{\rightarrow} \alpha_r$ in $R(T_i,a)$ such that $T_j$
appears in $\alpha_r$.
\end{definition}

\subsection{Generalized Newton's Method}

The problem of approximating efficiently the LFP of a PPS was solved 
in \cite{ESY12},
by using Newton's method (combined with suitable rounding), applied 
after elimination of the variables with LFP value $0$ and $1$.
We first recall the definition of Newton iteration for PPSs.

\begin{definition} For a PPS $x=P(x)$ we use 
$B(x)$ to denote the Jacobian matrix of partial derivatives of $P(x)$,
i.e., $B(x)_{i,j} := \frac{\partial P_i(x)}{\partial x_j}$. 
For a point $x \in \mathbb{R}^n$, if $(I-B(x))$ is 
non-singular, then we define one Newton iteration at $x$ via the operator:
$$\mathcal{N}(x) = x + (I-B(x))^{-1}(P(x) -x)$$
Given a max/minPPS, $x=P(x)$, and a policy $\sigma$,  
we use $\mathcal{N}_\sigma(x) $ to denote the Newton operator
of the PPS $x=P_{\sigma}(x)$; 
i.e.,  letting $B_\sigma(x)$ denote the Jacobian of $P_\sigma(x)$, 
if $(I-B_\sigma(x))$ is 
non-singular at a point $x \in \mathbb{R}^n$,
then $\mathcal{N}_\sigma(x) = x + (I-B_\sigma(x))^{-1}(P_\sigma(x) -x)$.
\end{definition}

\begin{definition}  For a max/minPPS, $x=P(x)$, with $n$ variables (in SNF form),
the {\bf linearization of $P(x)$ at a point ${\mathbf y} \in \mathbb{R}^n$},   
is a system of max/min linear functions denoted by $P^{y}(x)$, which has the
following form:

if $P_i(x)$ has form L or M,  then $P^y_i(x) = P_i(x)$, and

if $P_i(x)$ has form Q, i.e., $P_i(x) =x_jx_k$ for some $j$,$k$,
then $$P^y_i(x) = y_jx_k + x_jy_k - y_jy_k$$
\end{definition}

We can consider the linearization of a PPS, $x=P_\sigma(x)$, 
obtained as the result of fixing a policy, $\sigma$, for a 
max/minPPS, $x =P(x)$.

\begin{definition} $P^y_\sigma(x) := (P_\sigma)^y(x)$.\end{definition}

Note than the linearization $P^y(x)$ only changes equations of
form Q, and using a policy $\sigma$ only changes equations of form M, so these
operations are independent in terms of the effects they
have on the underlying equations, and thus
$P^y_\sigma(x) \equiv (P_\sigma)^y(x) = (P^y)_\sigma(x)$.

We now recall and adapt from \cite{esy-icalp12} 
the definition of distinct iteration operators for a maxPPS and a minPPS,
both of which we shall refer to with the overloaded notation $I(x)$. 
These operators serve as the basis for
{\em Generalized Newton's Method} (GNM) to be applied to maxPPSs and minPPSs, respectively.
We need to slightly adapt the definition of operator $I(x)$, specifying the conditions
on the GFP $g^*$ under which the operator is well-defined:

\begin{definition}
For a maxPPS, $x=P(x)$, with GFP $g^*$, such that $\vzero \leq g^* < \vone$,
and for a real vector $y$ such
that $\vzero \leq y \leq g^*$, we define the operator $I(y)$ to be the 
{\em unique} optimal solution, $a \in \real^n$, to the following 
mathematical program: $\quad \quad \mbox{\em Minimize:} \  \ \sum_i a_i \ ; \quad  
 \mbox{\em Subject to:} \quad   P^y(a) \leq a$.

For a minPPS, $x=P(x)$, with GFP $g^*$, such that $\vzero \leq g^* < \vone$, 
and for a real vector $y$ such 
that $\vzero \leq y \leq g^*$, we define the operator $I(y)$ to be the 
{\em unique} optimal solution  
$a \in \real^n$ 
to the following mathematical program:
$ \quad \quad \mbox{\em Maximize:} \  \ \sum_i a_i \ ; \quad 
 \mbox{\em Subject to:} \quad   P^y(a) \geq a$.

\end{definition}

\noindent
In both cases, the mathematical programs can be solved using Linear Programming.
In the case of a maxPPS, the constraint $P^y_i(a) \leq a_i$
for each variable $x_i$ of form L or Q is linear,
and the constraint for a variable $x_i$ of form M with $P_i(x)= \max(x_j,x_k)$
can be replaced by the two inequalities $a_j \leq a_i$ and $a_k \leq a_i$.
Similarly, in the case of a minPPS, the constraints for variables of form
L and Q are linear, and the constraint $P^y_i(a) \geq a_i$
for a variable $x_i$ of form M with $P_i(x)= \min(x_j,x_k)$
can be replaced by the two inequalities $a_j \geq a_i$ and $a_k \geq a_i$.
A priori, it is not clear whether the mathematical programs
have a unique solution, and hence whether the  above ``definitions'' of $I(x)$ for
maxPPSs and minPPSs are well-defined.  We will see that they are
(again, adapting facts for GNM applied to LFP computation from \cite{esy-icalp12}). 

We require a {\em rounded} version of GNM,  
defined in \cite{esy-icalp12} as follows. \\
{\bf\em GNM, with
rounding parameter $h$}:  Starting
at $x^{(0)} := {\mathbf{0}}$,  For $k \geq 0$, compute
$x^{(k+1)}$ from $x^{(k)}$ as follows: first calculate $I(x^{(k)})$,
then  for each coordinate $i=1,2,\ldots,n$, set $x_i^{(k+1)}$ to be the 
maximum (non-negative) multiple of $2^{-h}$ which is $\leq 
\text{max}\{0,I(x^{(k)})_i\}$. 
(In other words, round $I(x^{(k)})$ {\em down} to the nearest $2^{-h}$ and ensure it is non-negative.)

\section{Greatest Fixed Points capture non-reachability values}

\label{sec:gfp}

For any given BSSG, $\mathcal{G}$, with a specified special target type $T_{f^*}$, 
we will construct
a max-minPPS, $x=P(x)$, and show that the vector $g^*$ of non-reachability
values for  $(\mathcal{G}, T_{f^*})$ is precisely the {\em greatest fixed point} $g^* \in [0,1]^n$
of $x=P(x)$.

The system $x=P(x)$ will have one variable $x_i$ and one equation 
$x_i = P_i(x)$, for 
each type $T_i \neq T_{f^*}$.
For each $i \neq f^*$,  the min/max probabilistic polynomial $P_i(x)$
is constructed as follows.
For all $j \in A_i$, let 
$R'(T_i,j) := \{r \in R(T_i,j) : (\alpha_r)_{f^*} = 0 \}$
denote the set of rules for type $T_i$ and action $j$ 
that generate a multiset $\alpha_r$ not containing any 
element of type $T_{f^*}$.
$P_i(x)$
contains one probabilistic polynomial
$q_{i,j}(x)$ for each action $j \in A_i$, with 
$q_{i,j}(x) = \sum_{r \in R'(T_i,j)} p_r x^{\alpha_r}$.
In particular, note that we {\em do not} include, in the sum that defines
$q_{i,j}(x)$, any monomial $p_{r'} x^{\alpha_{r'}}$ 
associated with a rule $r'$ which generates at least one object of the special
type $T_{f^*}$.  
Then, if type $T_i$ belongs to the $\max$ player, who aims to {\em minimize}
the probability of {\em not} reaching an object of type 
$T_{f^*}$, we define $P_i(x) \equiv  \min_{j \in A_i} q_{i,j}(x)$.
Likewise, if type $T_i$ belongs to the $\min$ player, whose aim is to {\em maximize}
the probability of {\em not} reaching an object of type $T_{f^*}$, then we define
$P_i(x) \equiv \max_{j \in A_i} q_{i,j}(x)$.

Note the swapped roles that $\max$ and $\min$
play in the equations, versus the goal of the corresponding player in terms of the reachability objective.  
This swap is necessary because, whereas the objectives
of the players are to maximize or minimize reachability probabilities,  the equations we have
constructed will capture, in their {\em greatest fixed point} (GFP) solution, 
the optimal {\em non-reachability} values $g^*$.

The following theorem, which  is key, is analogous to a theorem proved in
\cite{rmdp} which proves a similar relationship between the LFP
of a max-minPPS and the extinction values of a BSSG:

 \begin{theorem}
 \label{bssg-equations}
 The non-reachability value vector
 $g^* \in [0,1]^n$ of the BSSG is equal to the Greatest Fixed Point (GFP) of the operator 
$P(\cdot)$ in $[0,1]^n$. Thus,
$g^* = P(g^*)$, and for all fixed points $g'=P(g')$, $g' \in [0,1]^n$,
$g' \leq g^*$.
 Furthermore, for any initial population $\mu$,
 the optimal non-reachability values satisfy
 $g^*(\mu) = \Pi_i (g^*_i)^{\mu_i}$
 and $g^*(\mu) = \sup_{\sigma \in \Psi_1} \inf_{\tau \in \Psi_2} g^{*,\sigma,\tau}(\mu)
 = \inf_{\tau \in \Psi_2} \sup_{\sigma \in \Psi_1} g^{*,\sigma,\tau}(\mu)$.
In particular, such games are {\em determined}.
 \end{theorem}
 \begin{proof}
 Let $x^k$ denote the $k$-fold application of $P$ on the all-1 vector,
 i.e. $x^0= \mathbf{1}$, and $x^k =P(x^{k-1})$ for $k >0$.
$P(\cdot)$ defines a monotone operator,
 $P : [0,1]^n \rightarrow [0,1]^n$, that maps $[0,1]^n$ to itself.
 Thus, the sequence $x^k$ is (component-wise) monotonically
 non-increasing as a function of $k$, bounded from below by the all-$0$ vector,
 and thus by Tarski's theorem %
it converges to the GFP, $x^* \in [0,1]^n$,
of the monotone operator $P(\cdot)$, as $k \rightarrow \infty$.
 We will first show the following lemma. 

 \begin{lemma}
\label{claim:k-step-value-reach}
 For any integer $k \geq 0$ and
 any finite non-empty initial population $\mu$ 
(expressed as an $n$-vector)
which does 
not contain any element of
type of $T_{f^*}$, the {\em value} 
 $g^k(\mu) := \sup_{\sigma \in \Psi_1} \inf_{\tau \in \Psi_2} g^{k}_{\sigma,\tau}(\mu)$
 of not reaching an element of type $T_{f^*}$ in  $k$ steps
is $g^k(\mu) = f(x^k, \mu) := \Pi_{i=1}^n (x^k_i)^{(\mu)_i}$.
 Furthermore, there are strategies of the two players
 (in fact deterministic strategies), $\sigma_k \in \Psi_1$ and $\tau_k \in \Psi_2$,
 that achieve this value, i.e,
 $g^k(\mu) = \inf_{\tau \in \Psi_2} g^{k}_{\sigma_k,\tau}(\mu) = 
 \sup_{\sigma \in \Psi_1} g^{k}_{\sigma,\tau_k}(\mu) $.
 \end{lemma}
 \begin{proof}
 We show the claim by induction on $k$.
 The basis, $k=0$, is trivial: namely we only have 
variables $x_i$ for each type $T_i \neq T_{f^*}$.
Thus, clearly starting with any finite non-empty 
population of objects of types $T_i \neq T_{f^*}$ the (optimal) probability
of not reaching an object of type $T_{f^*}$ within $0$ steps is $1$. 
 For the induction part, consider the generation of population $X_1$
 from $X_0$ in step 1. We show first that $g^k(\mu) \geq f(x^k,\mu) := 
\Pi_{i=1}^n (x^k_i)^{(\mu)_i}$.
 Consider the following strategy $\sigma_k$ for the max player
(the player trying to maximize the probability of {\em not} reaching the type $T_{f^*}$).
 For each entity in the initial population $X_0 = \mu$ of a max type $T_i$,
 the max player selects in step 1 (deterministically) an action $a \in A_i$
 that maximizes the expression $\sum_{r \in R'(T_i,a)}  p_r f(x^{k-1},\alpha_r)$
 on the right side of the equation $x^k_i =P_i(x^{k-1})$. 
 Once the min player also selects actions for the entities of min type
 in $X_0$, and rules for all the entities are chosen probabilistically
 to generate the  population $X_1$ for time 1, 
 the max player thereafter follows an optimal $(k-1)$-step strategy
 $\sigma_{k-1}$ starting from $X_1$. If we assume inductively that
 $\sigma_{k-1}$ is deterministic, then $\sigma_k$ is also
 deterministic. (It is not static however; the action chosen
 for an entity of a given type in a population $X_i$ 
 in the process may depend on the time $i$.)

 Let $\tau$ be any strategy of the min player.
 Consider a combination of actions chosen with nonzero probability
 by the min player in step 1 for the entities of min type in $X_0=\mu$.
 After this, a combination of rules is chosen 
randomly and independently for all the
 entities of $\mu$ and the population $X_1$ is generated accordingly
 with probability that is the  product of the rule probabilities that were
 applied (because the rules are chosen independently).
 By the induction hypothesis, the  value with which the population $X_1$ 
does not reach a type $T_{f^*}$ in the  next $k-1$ steps (i.e. by time $k$) 
is  $g^{k-1}(X_1) = f(x^{k-1},X_1)$. If, 
for each possible set $X_1$ 
(there are finitely many possibilities), 
we multiply $f(x^{k-1},X_1)$ with the probability of
the combination of rules that can be used in step 1 to generate $X_1$
from $X_0$, and 
we sum this over all possible $X_1$, 
we can write the result as a product
 of $|\mu|$ terms, one for each entity in $\mu$.
 The  term for an entity of max or min type $T_i$
 is $\sum_{r \in R'(T_i,a)}  p_r f(x^{k-1},\alpha_r)$, where $a$ is the action
 selected for this entity by the min or max player in step 1.
 For the max player, we selected an action $a \in A_i$ that maximizes 
 this expression, therefore the term for a $\max$ entity is equal 
to $P_i(x^{k-1})= x^k_i$.
 
For an entity that belongs to the min player, no matter which
 action the player chose, the term is greater than or equal to the minimum
 value over all available actions, which is $P_i(x^{k-1})= x^k_i$.
 Hence, for any combination of actions chosen by the
 min player in step 1, the probability that the process does not reach
an object of type $T_{f^*}$
 by step $k$ under the strategies $\sigma_k, \tau$ is at least $f(x^k, \mu)$.
 Therefore, this holds also if $\tau$ makes a randomized selection
 in step 1, i.e., assigns nonzero probability to more than one
 combinations of actions for the min entities in $\mu$.
 Thus, $\inf_{\tau \in \Psi_2} g^{k}_{\sigma_k,\tau}(\mu) \geq f(x^k, \mu)$
 and hence $g^k(\mu) \geq f(x^k, \mu)$.

 We can give a symmetric argument for the min player to
 prove the reverse inequality.
 Define strategy $\tau_k$ for the min player as follows.
 In step 1, the min player chooses for each entity of min type $T_i$
 in the initial population $\mu$, an action $a \in A_i$ that
 minimizes the expression $\sum_{r \in R'(T_i,a)}  p_r f(x^{k-1},\alpha_r)$ 
 on the right side of the equation $x^k_i =P_i(x^{k-1})$,
 and then, once the max player has chosen actions
 for the max entities of $\mu$, and rules are selected and
 applied to generate the population $X_1$, the  min player follows the
 optimal deterministic strategy $\tau_{k-1}$ starting from $X_1$ (assumed
to exist by induction).
 By a symmetric argument to the max player case,
 it is easy to see that 
 $\sup_{\sigma \in \Psi_1} g^{k,\sigma,\tau_k}(\mu) \leq  f(x^k, \mu)$
 and hence $g^k(\mu) \leq f(x^k, \mu)$.
 It follows that 
 $g^k(\mu) = \inf_{\tau \in \Psi_2} g^{k}_{\sigma_k,\tau}(\mu) = 
 \sup_{\sigma \in \Psi_1} g^{k}_{\sigma,\tau_k}(\mu) = f(x^k, \mu)$.
 
 \end{proof}

 In particular, for singleton initial populations, 
 the Lemma implies that $ g^k_i = x^k_i$ for all types $T_i \neq T_{f^*}$,
 and for all $k \geq 0$.

 Let $x^* = \lim_{k \rightarrow \infty} x^k$ denote 
the Greatest Fixed Point (GFP) 
of the equation $x=P(x)$.
 We will show that for any initial population $\mu$,
 the ``value''  $g^*(\mu) := \sup_{\sigma \in \Psi_1} \inf_{\tau \in \Psi_2} g^{*}_{\sigma,\tau} (\mu)$
of not ever reaching a population containing an object of type 
$T_{f^*}$ satisfies
$g^*(\mu) =
\inf_{\tau \in \Psi_2} \sup_{\sigma \in \Psi_1} g^{*}_{\sigma,\tau}(\mu) = f(x^*, \mu) $.
 In particular, these games are indeed determined.
For singleton populations, 
this  implies that $g^*_i = x^*_i$ for all types $T_i \neq T_{f^*}$.

 Since $x^k$ converges to $x^*$ from above as $k \rightarrow \infty$,  
 the sequence $f(x^k, \mu)$ converges to $f(x^*, \mu)$ from above.
 Thus, for every $\epsilon >0$ there is a $k(\epsilon)$ such that
 $f(x^*,\mu) \leq f(x^{k(\epsilon)}, \mu) < f(x^*, \mu) +\epsilon$.
 
From the proof of Lemma 
\ref{claim:k-step-value-reach},
the strategy $\tau_{k(\epsilon)}$ of the $\min$ player
(who is 
 minimizing the probability of {\em not} reaching $T_{f^*}$ in $k(\epsilon)$ 
rounds), 
satisfies, for all strategies $\sigma \in \Psi_1$,
$g^{*}_{\sigma,\tau_{k(\epsilon)}}(\mu) \leq g^{k(\epsilon)}_{\sigma, \tau_{k(\epsilon)}}(\mu)
 \leq \sup_{\sigma \in \Psi_1} g^{k(\epsilon)}_{\sigma,\tau_{k(\epsilon)}}(\mu) = 
f(x^{k(\epsilon)}, \mu) < f(x^*, \mu)  + \epsilon$.
 Since this holds for every 
$\epsilon >0$,
 it follows that 
 $g^*(\mu) = \sup_{\sigma \in \Psi_1} \inf_{\tau \in \Psi_2} g^{*}_{\sigma,\tau} (\mu)
\leq \inf_{\tau \in \Psi_2} \sup_{\sigma \in \Psi_1} g^{*}_{\sigma,\tau}(\mu)
 \leq f(x^*, \mu)$.

 For the converse inequality,
 let $\sigma^*$ be the static deterministic strategy for the max player
(who is trying to maximize the probability of {\em not}
reaching $T_{f^*}$),
 which always chooses for each
 entity of max type $T_i$ an action $a \in A_i$ that maximizes
 the expression $\sum_{r \in R(T_i,a)}  p_r f(x^{*},\alpha_r)$.
 If we fix the actions for all the $\max$ types according to $\sigma^*$, 
 the BSSG $G$ becomes
 a minimizing BMDP $G'$ where all the max types of $G$ become
 now choice-less or ``random'' types (meaning that no choice is available to 
the max player: it has only one action it can take at every type that
belongs to it). 
Let $x=P'(x)$ be the set of equations for $G'$;
 for the min types $T_i$ of $G'$, the equation is the same,
 i.e., $P'_i=P_i$; whereas for max types $T_i$ 
the function on the right-hand side
 changes from $P_i(x) =  \max_{a \in A_i} \sum_{r \in R(T_i,a)} p_r  f(x,\alpha_r)$
 to $P'_i(x) =  \sum_{r \in R(T_i,a_i)} p_r  f(x,\alpha_r)$, for some specific 
action $a_i \in A_i$.
 Thus, $P'(x) \leq P(x)$ for all $x \in [0,1]^n$.
 Let $y^k, k=0,1,\ldots$ be the vector resulting from the $k$-fold application
 of the operator $P'$ on the all-1  vector.
 Then $y^k \leq x^k$ for all $k$, and therefore the GFP $y^*$ of $P'$
satisfies $y^* \leq x^*$, where
 $x^*$ is the GFP of $P$.
 However, $x^*$ is a fixed point of $P'$, since we have chosen actions
 for all the max types $T_i$ that achieve the maximum in $P_i(x^*)$.
 Therefore, $x^*=y^*$, and both $x^*$ and $y^*$ are 
the GFP of both $P'$ and $P$.

Consider any fixed strategy $\tau$ of the min player starting from 
initial population $\mu$.
Applying Lemma \ref{claim:k-step-value-reach} to the BMDP $G'$, 
we know that for every $k$, the
probability, using strategy $\tau$ in $G'$,
of {\em not} reaching the type $T_{f^*}$ in $k$ steps,
starting in population $\mu$ is at least
$f(y^k,\mu)$.
 Therefore, the optimal (infimum) probability of not reaching 
a type $T_{f^*}$ in any number of  steps
is  at least $\lim_{k \rightarrow \infty} f(y^k,\mu) = f(y^*,\mu) =f(x^*,\mu)$.
 That is, $\inf_{\tau \in \Psi_2} g^{*}_{\sigma^*,\tau} (\mu) \geq f(x^*,\mu)$.
 Combining with the previously established inequality,
 $g^*(\mu) \leq f(x^*, \mu)$, 
 and since
 clearly $g^*(\mu) = \sup_{\sigma \in \Psi_1} \inf_{\tau \in \Psi_2} g^{*}_{\sigma,\tau} (\mu)
 \geq \inf_{\tau \in \Psi_2} g^{*}_{\sigma^*,\tau} (\mu)$, 
 we conclude that $\sigma^*$ is actually an optimal (static) strategy
for the player maximizing the non-reachability probability of $T_{f^*}$, and
that
 $f(x^*,\mu) = \inf_{\tau \in \Psi_2} g^{*}_{\sigma^*,\tau}(\mu) =
 \sup_{\sigma \in \Psi_1} \inf_{\tau \in \Psi_2} g^{*}_{\sigma,\tau} (\mu)= 
g^*(\mu) = 
 \inf_{\tau \in \Psi_2} \sup_{\sigma \in \Psi_1} g^{*}_{\sigma,\tau}(\mu)$.
 
 \end{proof}

\noindent A direct corollary of the proof of Theorem \ref{bssg-equations}
is that the player maximizing non-reachability probability 
always has an optimal static strategy:

\begin{corollary}
\label{cor:opt-maxPPS}
In any Branching Simple Stochastic Game, $G$,
where the objective of the players is to maximize and minimize,
respectively, the probability of {\em not} reaching a
type $T_{f^*}$, the player trying to {\em maximize}
this probability always has a deterministic static optimal strategy $\sigma^*$.

In particular, for any max-minPPS, $x=P(x)$, with GFP $g^*$,
the max player has an optimal deterministic policy, $\sigma^*$,
for the GFP,
such that $g^* = g^*_{\sigma^*,*}$  (where, recall, $g^*_{\sigma^*,*}$ is the
GFP of $x = P_{\sigma^*,*}(x)$).
\end{corollary}

\begin{proof}
Just use the deterministic static optimal strategy $\sigma^*$ 
for the maximizing player defined 
in the proof 
of Theorem \ref{bssg-equations}, which for each type $T_i$
controlled by the max player chooses an action $a \in A_i$ 
which maximizes the expression $\sum_{r \in R'(T_i,a)} p_r f(x^*,\alpha_r)$.

Clearly, this also implies the existence of a deterministic optimal policy,
$\sigma^*$, 
for the max player, for the GFP $g^* = g^*_{\sigma^*,*}$ in any max-minPPS $x=P(x)$.
\end{proof}

The same is {\em not} true for the player
trying to {\em minimize} this non-reachability probability.
In other words, the same is not true for
the player trying to {\em maximize} the
probability of reaching a type $T_{f^*}$.  
This is illustrated by the
following two examples:

\noindent 
\begin{example}[{\em In general, 
there is no randomized static optimal strategy for maximizing 
the reachability probability
in BMDPs, even when the supremum probability is 1.}]
\label{example1}
{\rm
Consider a BMDP with three types: $\{A,B,C\}$.
Type $C$ is the goal type (i.e., $C= T_{f^*}$).
The BMDP is described by the following rules for types $A$ and $B$.
The only controlled type is $A$.  The type $B$ is purely ``random''.
The symbol ``$\emptyset$'' denotes that one of the rules
for type $B$ generates, with probability $1/2$, the empty
set, containing no objects, 
from an object of type $B$.

\begin{eqnarray*}
A  & \rightarrow &   AA\\
A  & \rightarrow &   B\\ \vspace*{0.1in}
B  &\stackrel{1/2}{\rightarrow} & C\\
B  &\stackrel{1/2}{\rightarrow} & \emptyset
\end{eqnarray*}

It is easy to see that for this BMDP, the controller who
wishes to maximize the probability of reaching type $C$,
starting with one object of type $A$, can do so with probability
$1-\epsilon$, for any $\epsilon > 0$.  The strategy for
doing so is the following:  first create sufficiently
many copies of $A$, 
namely $k = \lceil \log( 1/\epsilon) \rceil$ copies, 
by using the rule $A \rightarrow AA$. 
Then, for each of the created copies, choose the ``lottery'' $B$.
Each ``lottery'' $B$ will, independently,  with 1/2 probability, reach $C$.
This assures that the total probability of {\em not} reaching
a $C$ is  $\frac{1}{2^{k}} \leq \epsilon$.

Thus, the supremum {\em value} of reaching $C$ in this BMDP is clearly
$1$.  However, it is also easy to see that there is no 
randomized static optimal
strategy that achieves this supremum value of $1$.  
This is because any randomized static strategy 
which places positive probability on the rule $A \rightarrow B$ 
would with
positive probability $p^*$ bounded away from $0$ go extinct starting
from a bounded population of $A$'s (without hitting $C$).

The minPPS for this BMDP has two variables $a, b$ and two equations
$a=\min(a^2,b)$ and $b=1/2$. This system has clearly only one fixed point:
$a^*=0, b^*=1/2$. However, there is no policy (whether
deterministic or randomized) that gives
$(0,1/2)$ as the GFP of the resulting PPS, for the same reason given above
that the BMDP does not have any optimal static strategy. 
Note in particular, that if a policy
selects for $a$ the first choice, $a^2$, then the resulting PPS is
$a=a^2, b=1/2$, and $a$ has value 1 in its GFP, not 0.

On the other hand, for this BMDP there is a \underline{non-static}
optimal strategy that achieves the reachability value $1$, namely, do as follows:
starting from one $A$, first use $A \rightarrow AA$ to create two $A$'s.
Then apply $A \rightarrow B$ to the ``left'' $A$ and apply 
$A \rightarrow AA$ to the ``right'' $A$. 
Now we have two $A$'s and a $B$.  The $B$ gives us a chance
to reach $C$.  On the two $A$'s, we again take the left $A$ to $B$
and the right $A$ to $AA$.  Repeat.
This way, the population will repeatedly contain two $A$'s and one 
$B$ forever, and each time $B$ is created it gives us a positive
chance to reach $C$, so we reach $C$ with probability $1$.

It turns out, as we will show later,
that for any BSSG, {\em if the reachability value is $1$},
then the player maximizing the probability of reachability
always has a {\em not necessarily static}, optimal strategy that
achieves this value 1.

This is {\em not} the case if the reachability value is strictly less than $1$,
as we shall show in the next example, Example \ref{example:no-opt-strat}.

On the other hand, if the goal was to {\em minimize} the probability
of reaching $C$, then starting from $A$ 
there is a simple strategy in this BMDP that achieves this:  deterministically
choose the rule $A \rightarrow AA$ from all copies of $A$.
This ensures that the process never reaches $C$, i.e.,
reaches $C$ with probability $0$.
This is clearly an optimal strategy.  Indeed, this holds in general:
as shown in Corollary \ref{cor:opt-maxPPS},
there always exists a deterministic static optimal strategy for
minimizing the probability of reaching a given type
(i.e., maximizing the probability of not reaching it),  in a BMDP or BSSG.}
\qed 
\end{example}

\begin{example}[{\em No optimal strategy at all for maximizing reachability 
probability in
a BMDP}]
\label{example:no-opt-strat}
{\rm 
We now give an example of a BMDP where the supremum reachability
probability of the designated type $T_{f^*}$ is $< 1$, and such
that there does not exist {\em any} optimal strategy (regardless of
the memory or randomness used) that achieves the value.

Consider the following BMDP, where the goal is to 
maximize the probability of reaching type $D$:

\begin{eqnarray*}
A  & \stackrel{2/3}{\rightarrow} &   BB\\
A  & \stackrel{1/3}{\rightarrow} &   \emptyset\\ \vspace*{0.1in}
B  & \rightarrow &   A \\
B  & \rightarrow &   C\\
C  &\stackrel{1/3}{\rightarrow} & D\\
C  &\stackrel{2/3}{\rightarrow} & \emptyset
\end{eqnarray*}

We claim that:

\begin{enumerate}
\item  The supremum probability, starting with one $A$, of
eventually reaching an object of type $D$ is $1/2$.

\item There is no strategy of any kind that achieves probability $1/2$.
\end{enumerate}

\begin{proof}
\begin{enumerate}
\item First, to see that the supremum probability starting at $A$
is $1/2$, consider the following sequence of strategies:
strategy $\tau^k$, for $k \geq 1$,  chooses  $B \rightarrow A$
for all objects in every multiset $X_i$ until
a multiset is reached in which there are at least $k$ B's. 
Then, in the next step, $\tau^k$ chooses $B \rightarrow C$
for all copies of $B$. 
In other words,  the strategy waits until there are ``enough''
$B$'s, and then switches to $B \rightarrow C$ for all $B$'s.
Note firstly that, with probability at least
$1/2$ we will eventually have a population of $B$'s exceeding $k$,
for any $k$.
Thereafter the probability of not hitting $D$ will be at most $(2/3)^k$.
We can make $k$ as large as we like, and thus we can make the probability
of not hitting $D$, conditioned on reaching population $k$, as small
as possible.   So we can make the probability of hitting $D$ as close
as we like to $1/2$.

This can be seen also from the corresponding minPPS using Theorem \ref{bssg-equations}.
The minPPS has three variables $a,b ,c$ and equations
$a= \frac{2}{3} b^2 + \frac{1}{3}$, $b=\min(a,c)$, $c=\frac{2}{3}$.
It is easy to see that the system has only one fixed point, 
$a^* = b^* = \frac{1}{2}, c^*=\frac{2}{3}$, which is thus the GFP.
Hence, by Theorem \ref{bssg-equations}, the reachability value of the BMDP is
$1-a^* = 1/2$. However, there is no policy of the minPPS 
(and correspondingly, no static strategy of the BMDP) that achieves this value.
In particular, note that the policy that selects for $b$ the first choice $a$,
yields a PPS $\{a= \frac{2}{3} b^2 + \frac{1}{3}, b=a, c=\frac{2}{3} \}$ 
with a GFP in which $a$ has value 1, instead of 1/2.

\item To see that in fact there is no strategy (whether static or not) of the BMDP
that achieves probability $1/2$ , assume, for contradiction, that there
does exist a strategy $\sigma$ that achieves probability $1/2$. 

Consider any occurrence of $B$ in the history $X_0, X_1, \ldots$ of
configurations,  such that the rule $B \rightarrow C$ is applied
with positive probability
to that occurrence of $B$ by the strategy $\sigma$.  
It is without loss of generality to assume that such a $B$ exists,
because otherwise the probability of reaching $D$ would be $0$.

We claim that the total probability of
reaching type $D$ would strictly increase if, instead of 
applying action $B \rightarrow C$ with positive
probability $p'$ on that copy of $B$,  
the strategy $\sigma$ instead is changed to a strategy $\sigma'$ where
that positive probability $p'$ on action $B \rightarrow C$ 
is shifted entirely to the pure action
$B \rightarrow A$, and thereafter, in the next step, if on that resulting $A$ 
the random rule $A \stackrel{2/3}{\rightarrow} BB$ happens to get chosen,
the strategy $\sigma'$ then 
(with the shifted probability $p'$) immediately applies the rule 
$B \rightarrow C$  to both resulting copies of $B$.

To see why this switch to strategy $\sigma'$ 
would strictly increase the probability of
reaching $D$, note that for any given $B$ by choosing $B \rightarrow C$
deterministically 
the probability of reaching $D$ from that copy of $B$ becomes exactly $1/3$.
On the other hand, by choosing $B \rightarrow A$ from
that copy of $B$ and
thereafter (with $2/3$ probability) choosing $B \rightarrow C$
on the resulting two copies of $B$, the new probability of hitting
$D$ is  $2/3 \cdot (1- (2/3)^2) = 10/27 > 1/3$.   
The same analysis shows that even if the original
strategy $\sigma$ only chose $B \rightarrow C$ with positive 
probability $p > 0$ then shifting
that probability over to the two-step strategy,
first choosing $B \rightarrow A$, achieves strictly greater
probability of reaching $D$.
Since this analysis holds for {\em any} copy of $B$ that occurs in 
the trajectory $X_0, X_1, \ldots $ of the process, we see that
we can always strictly increase the probability of reaching $D$
by {\em indefinitely delaying} the application of the rule $B \rightarrow C$.
However, note that we can not delay application of the rule 
$B \rightarrow C$ forever:
if we do so then the probability of reaching $D$ is actually $0$.

Thus, the supremum probability of reaching $D$ is only achieved 
in the limit by a sequence
of strategies, which delay the use of $B \rightarrow C$ longer and longer,
but is never attained by any single strategy.

We have already seen that the supremum probability 
of reaching $D$ is at least $1/2$, using the sequence
of strategies described
in part (1.) above.
Now, to see why the supremum value is indeed $1/2$, 
note that if we do indeed delay forever
using $B \rightarrow C$, then starting with one $B$ or one $A$ 
the process becomes extinct
with probability $1/2$ (without ever seeing a $D$).
Thus, if we delay using $B \rightarrow C$ for ``long enough'',
then the process becomes extinct with probability $1/2 - \epsilon$
without seeing $D$,
for an arbitrarily small positive $\epsilon > 0$.
So, the supremum value of the reachability probability can be at
most $1/2$, and thus is equal to $1/2$.   
Moreover, we have already argued that this supremum value
is not achieved by any strategy, because we can always achieve
strictly higher probability of reaching $D$ by delaying the use
of $B \rightarrow C$ one step further.  Thus, $1/2$ is the
supremum value, but is not achieved by any strategy.
\end{enumerate}
\end{proof}
}
\end{example}

\section{P-time detection of  GFP  $g^*_i=1$   for max-minPPSs and BSSGs}

\label{sec:qual-gfp-equal-1}

In this section we will show that there are (easy) P-time algorithms
to compute for a given max-minPPS the variables that have value 1 in
the GFP, and thus also for deciding, for a given BSSG (or BMDP),
whether $g^*_i = 1$ (i.e., whether the {\em non}-reachability value,
starting from a given type $T_i$ is $1$).  The algorithm does not
require looking at the precise values of the coefficients of the
polynomials in the max-minPPS (respectively, it does not depend on
probabilities labelling the transitions of the BSSG): it only depends
on the qualitative ``structure'' of the max-minPPS (the BSSG).
As we show, it reduces to an AND-OR graph reachability problem.

Recall that in the AND-OR graph reachability problem, we are given
a directed graph $G$, whose nodes are partitioned into a set $T$ of target nodes,
a set $V_1$ of OR nodes and a set $V_2$ of AND nodes.
The set of nodes that can {\em AND-OR reach} $T$ is defined to be
the (unique) smallest set $S$ of nodes that
includes $T$ and which has the property that (i) an OR-node $v$ is in $S$ iff at least
one of its immediate successors is in $S$, and (ii) an AND-node $v$ is in $S$ iff all
its immediate successors are in $S$.
This set can be computed easily by an iterative algorithm that initializes
$S$ to $T$, and then repeatedly adds to $S$ any OR-node $v$ that
has an immediate successor already in $S$, and any AND-node all of whose
immediate successors are already in $S$, until there are no more changes to $S$.
As is well-known, the algorithm can be implemented in linear time.
Equivalently, the AND-OR reachability problem can be viewed as
a two-person zero-sum reachability game, where the OR-nodes belong to
player 1 who wants to reach some node in the target set $T$, and the AND-nodes
belong to player 2 who wants to avoid this.
The set of winning nodes for player 1 is precisely the set $S$ of nodes that
can AND-OR reach $T$; a winning strategy $\tau$ for player 1 from each 
OR-node in $S$
is to pick an immediate successor that was added earlier to $S$.
The complementary set of nodes is winning for player 2; a winning strategy $\sigma$ for
player 2 from each AND-node that is not in $S$ is to pick an immediate successor
that is not in $S$ 
(there must be one, otherwise the AND-node would have been added to $S$).

\begin{proposition}
\label{lem:prob1-ptime}
There is a P-time algorithm that given a max-minPPS
(and thus also a maxPPS or minPPS), $x=P(x)$,
with $n$ variables, and
with GFP $g^* \in [0,1]^n$, and given $i \in [n]$,
decides whether $g^*_i = 1$, or $g^*_i < 1$.
The same result holds for determining for a given BSSG with non-reachability
objective, whether the value of the game is 1. 
Moreover, in the case where $g^*_i = 1$ the algorithm
computes a deterministic policy 
(i.e., deterministic static strategy in the BSSG case) $\sigma$,
for the max player which forces $g^*_i = 1$,
Likewise,
if $g^*_i < 1$, the algorithm computes a deterministic static policy $\tau$
for the min player which forces $g^*_i < 1$.
\end{proposition}
\begin{proof}
For simplicity, we assume w.l.o.g., that the max-min PPS, $x=P(x)$
is in SNF form.
Consider the dependency graph $G = (V,E)$ on the variables $V =
\{x_1,\ldots, x_n\}$ of $x=P(x)$.
The edges $E$ are defined as follows:   $(x_i, x_j) \in E$ if and
only if $x_j$ appears in one of the monomials with positive
coefficient that appear on the right hand side of $P_i(x)$.

Let us call a variable $x_i$ {\em deficient} if  $P_i(x)$
has form {\tt L} and the coefficients and
constant term in $P_i(x)$ sum to strictly less than
$1$; equivalently, $x_i$ is deficient iff $P_i(\vone) <1$.  
Let $\mathcal{Z} \subseteq \{x_1,\ldots,x_n\}$ denote the set of
deficient variables.

Let $X = V \setminus \mathcal{Z}$, denote the remaining set of non-deficient
variables.  
We partition the remaining variables $X = \mathtt{L} \cup
\mathtt{Q} \cup \mathtt{M}$
according to the form of the 
corresponding SNF-form equation $x_i = P_i(x)$.   
In fact, we further partition the variables  
$\mathtt{M}$  as $\mathtt{M} = \mathtt{M_{\max}} \cup \mathtt{M_{\min}}$,
according to whether the corresponding RHS for that variable
has the form $\max\{x_j,x_k\}$ or  $\min\{x_j, x_k\}$.

We can now view the dependency graph $G$  as a (non-probabilistic) 
AND-OR game graph, namely a 2-player {\em reachability game graph}, in which
the goal of player 1 is to reach a node in $\mathcal{Z}$,
whereas the goal of player 2 is to avoid this.
The nodes of the game graph belonging to player 1 are
$\mathtt{L} \cup \mathtt{Q} \cup \mathtt{M_{\min}}$ (these are the OR nodes),
the nodes of the game graph belonging to player 2 are
$\mathtt{M_{\max}}$ (these are the AND nodes),  and finally the nodes in $\mathcal{Z}$ are
the {\em target} nodes (from which player 1 wins automatically).

Let $S$ be the set of nodes that can AND-OR reach $\mathcal{Z}$,
i.e. the set of nodes from which player 1 can win, let ${\bar S}$
be the complementary set of nodes from which player 2 wins,
and let $\tau, \sigma$ be winning (deterministic, static)
strategies for the two players
from their respective winning sets, as described before
the proposition (the definition of the strategies on their
sets of losing nodes is irrelevant).
As we mentioned earlier, the sets $S, {\bar S}$ and
the strategies $\tau, \sigma$  can be computed
in P-time (in fact, in linear time).

We claim that for every variable $x_i$, we have $g^*_i < 1$
if and only if $x_i \in S$.

For the one direction, we can show that $g^*_i < 1 $, 
and in fact $(g^*_{*,\tau})_i <1$,
for all $x_i \in S$, by induction on the time that $x_i$ was added to $S$
in the iterative algorithm.
For the basis case, $x_i \in \mathcal{Z}$ is a deficient node, i.e. $P_i(\vone) <1$,
and hence clearly $g^*_i \leq (g^*_{*,\tau})_i =P_i(g^*_{*,\tau}) \leq P_i(\vone)<1$.
For the induction step, if $x_i$ is of type $\mathtt{M_{\min}}$
and $\tau$ chooses $x_j \in P_i(x)$ for $x_i$, then $x_j$ was added earlier to $S$,
thus $g^*_i \leq (g^*_{*,\tau})_i = (g^*_{*,\tau})_j < 1$.
The other cases when $x_i$ is of type $\mathtt{L}, \mathtt{Q}, \mathtt{M_{\max}}$
are similar.

To see the other direction, $g^*_i = 1$,
and in fact $(g^*_{\sigma,*})_i=1$, for all $x_i \in {\bar S}$,
note that the dependency graph of the minPPS $x=P_{\sigma,*}(x)$ has
no edges from ${\bar S}$ to $S$: all variables of
type $\mathtt{L} \cup \mathtt{Q} \cup \mathtt{M_{\min}}$ of ${\bar S}$
depend only on variables in ${\bar S}$ (otherwise, they would have been added to $S$),
and for variables of type $\mathtt{M_{\max}}$, policy $\sigma$ selected a successor in ${\bar S}$.
Furthermore, ${\bar S}$ does not contain any deficient node,
thus $P_i(\vone)=1$ for all $x_i \in {\bar S}$.
Therefore, the subsystem of $x=P_{\sigma,*}(x)$ induced by
${\bar S}$ has the all-1 vector as a fixed point,
hence $(g^*_{\sigma,*})_i=1$ (and thus $g^*_i=1$), for all $x_i \in {\bar S}$.
  \end{proof}

We will consider detection of $g^*_i=0$ for max-minPPSs
with GFP $g^*$  later in the paper.   We shall see that
for maxPPSs, after detection and removal of
variables $x_i$ such that $g^*_i =1$, so that $g^* < \vone$,
the GFP $g^*$ of the residual maxPPS is equal to the LFP $q^*$ of
the residual maxPPS, and thus detecting whether $g^*_i = q^*_i = 0$
can be done in P-time via simple AND-OR graph analysis
using the algorithm given in \cite{rmdp}.

For minPPSs, however, the above reduction does not hold,
and in fact the P-time algorithm for detecting whether $g^*_i = 0$
is substantially more complicated (but still does not
involve knowing the actual coefficients of the polynomials
in the minPPS, or the probabilities labeling rules of the BMDP,
only its structure).
We provide such a P-time algorithm for deciding whether $g^*_i = 0$,
not only for minPPSs, but also for the more general max-minPPSs, in Section 
\ref{sec:zero-detect}.

\section{Reachability for BPs, and 
linear degeneracy}

\label{sec:bp-ld}

In this section we study the reachability problem for purely stochastic
BPs.   Along the way, we establish several Lemmas which will be crucial
for our analysis of BMDPs.     
We start by defining the notion of
linear degeneracy.

A PPS $x=P(x)$ is called {\em linear degenerate} 
if every polynomial  $P(x)$ is 
linear, with {\em no} constant term, and all coefficients sum to $1$. 
Thus $x=P(x)$ is linear degenerate if $P_i(x) \equiv \sum^n_{j=1} p_{ij} x_j$,
where $p_{ij} \in [0,1]$ for all $i \in [n]$, and $\sum_j p_{ij} = 1$.
We refer to a linear degenerate PPS as an LD-PPS.

Note that for any LD-PPS, $x=P(x)$, we have  
$P(\vzero) = \vzero$ and $P(\vone) = \vone$, so the LFP is
$q^*=\vzero$ and the GFP is $g^* = \vone$. The Jacobian $B(x)$ of an LD-PPS is
a constant stochastic matrix B (independent of $x$), where every row of $B$ is non-negative and 
sums to 1.   During the evolution of the
associated BP, the size of the population remains constant. Thus, if
we start with a single object, the MT-BP trajectory $X_0, X_1, \ldots$ 
is simply the trajectory of a finite-state
Markov chain whose states correspond to types, and where the
singleton set $X_i$ corresponds to the one object in
the population at time $i$.   Note that the Jacobian $B(x)= B$ is the 
transition matrix of the corresponding finite-state Markov chain.
Furthermore, observe that for any LD-PPS we have $P(x) = B x$.

Given a PPS, we can construct its dependency graph and decompose it
into strongly connected components (SCCs). 
A bottom SCC is an SCC that has no
outgoing edges.
The following Lemma is immediate:

\begin{lemma} \label{lem:ld-or-ldf} For any PPS,  $x=P(x)$, 
exactly one of the following two cases holds:
\begin{itemize}
\item[(i)] $x=P(x)$ contains a 
linear degenerate bottom strongly-connected component
(BSCC), $S$, i.e., $x_S = P_S(x_S)$ is a LD-PPS,
and $P_S(x_S) \equiv B_S x_S$, for a stochastic matrix $B_S$.  

\item[(ii)] every variable $x_i$ either is, or depends 
(directly or indirectly) on, a variable $x_j$ where $P_j(x)$ has one of
the following properties:
\begin{enumerate}
\item $P_j(x)$ has a term of degree 2 or more,
\item $P_j(x)$ has a non-zero constant term i.e. $P_j(\vzero) > 0$ or
\item $P_j(\vone) < 1$.
\end{enumerate}
\end{itemize}
\end{lemma}

A PPS $x=P(x)$ is called
a {\em linear-degenerate-free} PPS (LDF-PPS) if it satisfies condition {\em (ii)} of 
Lemma \ref{lem:ld-or-ldf}.

\begin{lemma} \label{lem:inbetween-ldf} If a PPS, $x=P(x)$, has either GFP $g^* < \vone$, or LFP $q^* > \vzero$, 
then $x=P(x)$ is a LDF-PPS. 
\end{lemma}
\begin{proof} Suppose that for a PPS, $x=P(x)$ condition $(i)$ of 
Lemma \ref{lem:ld-or-ldf} holds, i.e., there is a bottom SCC $S$ with $P_S(x_S)=B_Sx_S$ for 
a stochastic matrix $B_S$. Then $P_S(\vzero)=\vzero$ and $P_S(\vone)=\vone$. So $g^*_S =\vone$ and $q^*_S=\vzero$, 
which contradicts the assumptions. So, condition $(ii)$ must hold, i.e. $x=P(x)$ is a LDF-PPS.
  \end{proof}

We use $\rho(A)$ to denote the spectral radius of a matrix $A$.
A basic property that we use is that, if $A$ is a non-negative matrix and $\rho(A)<1$,
then the matrix $I-A$ is nonsingular, and
its inverse $(I-A)^{-1}=\sum_{k=0}^{\infty} A^k$ is non-negative (see e.g. \cite{HornJohnson85}).

We will often use also the following lemma from \cite{ESY12} (stated there more
generally for monotone polynomial systems).

\begin{lemma} \label{lem:3.3ESY12}
(\cite{ESY12}, Lemma 3.3.) 
Let $x=P(x)$ be a PPS, with $n$ variables, in SNF form,
and let $a, b \in \real^n$. Then:
$P(a)-P(b) = B(\frac{a + b}{2})(a-b)$. 
\end{lemma}

The following is a strengthened variant of Lemma 2.12 from \cite{esy-icalp12}.

\begin{lemma}[cf. Lemma 2.12 of \cite{esy-icalp12}] \label{lem:ldf-Ndefined} For any (w.l.o.g., quadratic) 
LDF-PPS, $x=P(x)$ with LFP $q^*$, and for $\vzero \leq y < \frac{1}{2}(\vone +q^*)$, 
we have $\rho(B(y)) < 1$ and so $(I-B(y))^{-1}$ exists and is non-negative, and 
thus $\mathcal{N}(y)$ is well-defined. \end{lemma}
\begin{proof} The spectral radius $\rho(A)$ of a square non-negative
  matrix, $A$, is equal to the maximum of the spectral radii of its
  principal irreducible submatrices (see, e.g., \cite{HornJohnson85},
  Chapter 8).  Any principal irreducible submatrix of $B(y)$ is a
 principal irreducible submatrix of $B_S(y)$ for some SCC $S$ of the
 dependency graph of $x=P(x)$ ($B_S(y)$ itself might not be
 irreducible, since we do not assume $y > \vzero$).  So to show that
 $\rho(B(y)) < 1$, it suffices to show that for any SCC $S$,
 $\rho(B_S(y)) < 1$.

 For a trivial SCC, one where $S=\{x_i\}$ for a single variable $x_i$
 which does not appear in $P_i(x)$, $B_S(y)$ is the zero matrix so
 $\rho(B_S(y)) = 0 < 1$.
 
 Now we consider SCCs which are non-trivial and contain an equation of
 form Q, $x_i=P_i(x)$. Here $P_i(x) \equiv x_jx_k$ for some $j$,$k$
 must contain at least one term, say w.l.o.g., $x_j$ which is also in $S$
 or we would have the above trivial case. We have $B_S(y) \leq
 B_S(\frac{1}{2}(\vone +q^*))$ by monotonicity of $B(x)$. But
 $(B_S(y))_{i,j} = y_k < \frac{1}{2}(1 +q^*_k) =
 (B_S(\frac{1}{2}(\vone +q^*)))_{i,j}$.  So the inequality $B_S(y)
 \leq B_S(\frac{1}{2}(\vone +q^*))$ is strict in the $i,j$
 entry. Since the matrix $B_S(\frac{1}{2}(\vone +q^*))$ is
 irreducible, $\rho(B_S(y)) < \rho(B_S(\frac{1}{2}(\vone +q^*)))$ (again, see e.g., \cite{HornJohnson85}).  So
 it suffices to show that $\rho(B_S(\frac{1}{2}(\vone +q^*))) \leq
 1$. 
 
 There are two cases. Firstly suppose $q^*_S = \vone$. Then any
 SCC $D$ that $S$ depends on also has $q^*_D= \vone$.  So
 $B_S(\frac{1}{2}(\vone +q^*)) = B_S(\vone) = B_S(q^*)$.  But we know
 (\cite{rmc}, \cite{ESY12})  that $\rho(B(q^*)) \leq 1$ so we have that
 $\rho(B_S(\frac{1}{2}(\vone +q^*))) = \rho(B_S(q^*)) \leq
 \rho(B(q^*)) \leq 1$. 
 
 Secondly suppose that $q^*_S \not= \vone$. Then
 $q^*_S < \vone$. Applying Lemma \ref{lem:3.3ESY12}
 with $a=\vone$ and $b=q^*$, we have
 that $B(\frac{1}{2}(\vone +q^*)) (\vone-q^*)=P(\vone)-P(q^*) \leq (\vone-q^*)$. 
 Since
 $B(\frac{1}{2}(\vone +q^*))$ is non-negative and $\vone-q^* \geq \vzero$,
 we have that $B_S(\frac{1}{2}(\vone +q^*))(\vone-q^*_S) \leq
 (\vone-q^*_S)$.  By standard facts of Perron-Frobenius theory, since
 $\vone-q^*_S > \vzero$ and $B_S(\frac{1}{2}(\vone +q^*))(\vone-q^*_S) \leq
 (\vone-q^*_S)$, it follows that $\rho(B_S(\frac{1}{2}(\vone +q^*)) \leq 1$.  So in
 either case we have $\rho(B_S(y)) < \rho(B_S(\frac{1}{2}(\vone +q^*))
 \leq 1$.
 
 Finally we consider SCCs which contain only equations of form L. Here $B_S(y)$ is irreducible since $B_S(x)$ is a constant matrix and so if $i$ depends on $j$, $B_{i,j}(y) \not= 0$.
 $B_S(y)$ is also substochastic since  all the entries in the $i$'th row are coefficients in $P_i(x)$ and $x=P(x)$ is a PPS.
 Since $x=P(x)$, is a LDF-PPS, $B_S(y)$ is not stochastic since otherwise $S$ would be a bottom linear degenerate SCC.
 So there is an irreducible stochastic matrix $A$ with $B_S(y) \leq A$ with strict inequality in some entry. 
This implies  $\rho(B_S(y)) < \rho(A) = 1$.
   \end{proof}

\begin{lemma} \label{lem:ldf-uniquefp} For any LDF-PPS, $x=P(x)$, and
  $y < \vone$, if $P(y) \leq y$ then $y \geq q^*$ and if $P(y) \geq y$,
  then $y \leq q^*$.  In particular, if $q^* < \vone$, then $q^*$ is the
  only fixed-point $q$ of $x=P(x)$ with $q < \vone$. \end{lemma}
\begin{proof} Since $y < \vone$, $\frac{1}{2}(y+q^*) <
  \frac{1}{2}(\vone+q^*)$. By Lemma \ref{lem:ldf-Ndefined},
  $(I-B(\frac{1}{2}(y+q^*)))^{-1}$ exists and is non-negative. Lemma
  \ref{lem:3.3ESY12} yields that $P(y)- q^* =
  B(\frac{1}{2}(y+q^*))(y-q^*)$. Re-arranging this gives $q^*-y =
  (I-B(\frac{1}{2}(y+q^*)))^{-1}(P(y)-y)$. So when $P(y) -y \geq \vzero$
 we also have $q^*-y \geq \vzero$, and when $P(y)-y \leq \vzero$ we also have $q^*-y \leq \vzero$. That is if
  $P(y) \leq y$ then $y \geq q^*$ and if $P(y) \geq y$, then $y \leq
  q^*$.

Suppose $q < \vone$ is a fixed point, i.e. $P(q)=q$. Then both $P(q) \geq q$ and $P(q) \leq q$, 
so both $q \leq q^*$ and $q \geq q^*$. Thus $q=q^*$.   \end{proof}

We shall need the following fact about BPs later.

\begin{lemma}
\label{lem:either-extinct-or-infty}
For a BP, if the PPS associated with its extinction
probabilities (see \cite{ESY12}) is an LDF-PPS, $x=P(x)$,
and if all
types have extinction probability $q^*_i <1$, then
for any  population $z$ and any initial population, the
probability that $z$ occurs infinitely often is 0.
Consequently, starting with any initial population,
with probability 1 either the process becomes extinct
or the population goes to infinity.
\end{lemma}
\begin{proof}
Let $G$ be the dependency graph of the branching process.
Suppose first that $G$ is strongly connected.
We claim then that almost surely (with probability 1) the process either becomes
extinct or grows without bound (for any initial population). 
This can be shown easily using the results in \cite{Harris63}
in the so called positive regular (primitive) moment matrix case.
We give a direct proof. Suppose first that all types have positive
extinction probability, $q^*_i >0$. Let $X_k$ denote the population at time $k$,
for $k \geq 0$.
Then for every population $z \neq 0$, the probability
$P( X_k =0 | X_0=z) >0$ for some large enough $k$, and for all $k' \geq k$. 
Hence the population $z$ is a {\em transient} state of the underlying
countable-state Markov chain of the BP, that is, the probability
that $z$ occurs infinitely often is 0. Since this holds for every $z \neq 0$,
 the process almost surely either becomes extinct or grows without bound. 

Suppose now that there are some types $i$ with extinction probability $q^*_i =0$,
and let $Z$ be the set of all such types.
Then every rule of every type in $Z$ includes in the offspring at least one element of $Z$.
So the population of objects with type in $Z$ can never go down.
Since the process is not linear degenerate, at least one type $i^*$ of $Z$
has a rule $r^*$ with two or more offspring. Since $G$ is strongly connected,
if we start with an object of any type, with positive probability, the
process will generate within $n$ steps an object of type $i^*$, apply rule $r^*$, and within another $n$
steps, the (at least) two offspring can generate two objects with type in $Z$.
If the process does not go extinct, this happens infinitely often almost surely,
and since the number of objects with type $Z$ never goes down, this implies that
the size goes to infinity. 
Hence, with probability 1, the process either goes extinct or grows without bound.
Thus, the lemma holds if $G$ is strongly connected.

Consider now a branching process with a dependency graph $G$ that is not
strongly connected. Suppose that there is positive probability that a population $z$
occurs infinitely often. 
Let $i$ be the type of an object in $z$
and let $j$ be a type reachable from $i$ that is in a bottom strongly connected component $S$.
Every time there is an object of type $i$ in the population, there is positive
probability that it will generate later on an object of type $j$.
Since $z$ occurs infinitely often, 
almost surely the process will contain also infinitely often objects of type $j$.
Since $q^*_j <1$, the process starting with a single object of type $j$,
grows without bound with positive probability.
Since objects of type $j$ occur infinitely often, 
the probability that the process stays bounded is 0.
 
\end{proof}

\begin{lemma}
\label{lem:PPS:unique-fixpoint}
If $x=P(x)$ is a PPS with GFP $g^*$ such that $\vzero \leq g^* < \vone$,   then $g^*$ is
          the unique fixed point solution of $x=P(x)$ in $[0,1]^n$.   In
          other words, $g^* = q^*$, where $q^*$ is the LFP of $x=P(x)$.
\end{lemma}

\begin{proof}
Since $g^* < \vone$, by Lemma \ref{lem:inbetween-ldf}, $x=P(x)$ is a LDF-PPS.
Thus, since $P(g^*) \geq g^*$, it follows by  Lemma \ref{lem:ldf-uniquefp} that
$q^* = g^*$.
  \end{proof}

\begin{proposition}{(cf. also \cite{CDK12}, Proposition 5\&6, and Lemma 20; and \cite{ESY12})}
 Given a PPS, $x=P(x)$, 
with GFP $g^*$, and given any integer $j>0$, there is an algorithm that
computes a rational vector $v \leq g^*$ with $\|g^* -v\|_\infty \leq 2^{-j}$,
in time polynomial in $|P|$ and $j$.
\end{proposition}
\begin{proof}
By Proposition \ref{lem:prob1-ptime}, 
it is without loss of generality to assume that $g^* < \vone$,
because
we first preprocess $x=P(x)$, and remove
the variables $x_j$ such that $g^*_j = 1$, plugging in $1$ in their place on
RHSs of other equations.  So, we assume wlog that PPS $x=P(x)$ satisfies $g^* < \vone$.
By Lemma \ref{lem:PPS:unique-fixpoint},   $x=P(x)$ has a unique fixed point in
$[0,1]^n$, and $g^* = q^*$, where $q^*$ is the LFP.    We can then simply
apply the algorithm from \cite{ESY12}, to approximate the LFP $q^* = g^*$ 
of $x=P(x)$ within $j$ bits of precision in time polynomial in $|P|$ and $j$. 
  \end{proof}

\section{Approximating the GFP of a maxPPS in P-time}

\label{sec:maxpps-gfp-approx}

In this section, we will show that we can approximate the GFP of a maxPPS
and compute an $\epsilon$-optimal deterministic policy in polynomial time.
We show also that we can determine easily if the value is 0.

We call a policy $\sigma$ for a max/minPPS, $x=P(x)$, 
{\em linear degenerate free} (LDF)  if its 
associated PPS $x=P_\sigma(x)$ is an LDF-PPS.

\begin{lemma}
\label{lem:maxPPS:unique-fixpoint}
For any maxPPS, $x=P(x)$, if GFP $g^* < \vone$  then  $g^*$ is
          the unique fixed point of $x=P(x)$ in $[0,1]^n$.   In
          other words, $g^* = q^*$, where $q^*$ is the LFP of $x=P(x)$.
\end{lemma}

\begin{proof}
Suppose $x=P(x)$ is a maxPPS with  GFP $g^* < \vone$.

We know, by Corollary \ref{cor:opt-maxPPS}, that 
there is a deterministic optimal policy for achieving the GFP
for $x=P(x)$,  i.e., there is a deterministic policy $\sigma$  such that
$g^* = g^*_{\sigma}$, where $g^*_{\sigma}$ is the GFP of the PPS 
$x= P_{\sigma}(x)$.
(Namely, $\sigma$ just picks, from each type, an
 action that maximizes the RHS of the corresponding equation
   evaluated at $g^*$.)

Let  $\sigma$ be such an optimal policy.
Then $\vzero \leq g^*_\sigma = g^* < \vone$.
By Lemma \ref{lem:PPS:unique-fixpoint} this implies   
$\vzero \leq q^*_\sigma = g^*_\sigma < \vone$.
Next, we observe the following easy fact:

\begin{lemma}
\label{claim:ineq-maxPPS}
For all $z , z' \in [0,1]^n$,  if $z \leq z'$ then
         $P_\sigma(z)  \leq  P(z')$.
\end{lemma}
\begin{proof}
This holds because   $P(z) \leq P(z')$ by monotonicity of $P(x)$, and
because each expression $(P(z))_i$ in $P(z)$ consists of the $\max$ 
operator applied to a set of monotone polynomial terms,
which include among them the monotone polynomial $(P_\sigma(z))_i$,
and thus $P_\sigma(z) \leq P(z)$.
  \end{proof}

Now we consider ``value iteration'' starting from the all-0 vector,
on {\em both}  the PPS $P_\sigma(x)$ and  the maxPPS $P(x)$.  
Let $x^0  := y^0 := {\mathbf 0}$. 
For $i \geq 1$, let  $x^i :=  P^i_\sigma(\vzero)$ and let $y^i :=  P^i(\vzero)$.
Note that  $x^i \leq x^{i+1}$  and  $y^i \leq y^{i+1}$, for all $i \geq 0$.

We claim that $x^i \leq y^i$  for all $i \geq 0$.
This holds by induction on $i$: base case $i=0$ is by definition.
For $i \geq 0$, assuming $x^i \leq y^i$,
we have $x^{i+1} = P_\sigma(x^i) \leq P(y^i)  = y^{i+1}$, where
the middle inequality follows by Lemma \ref{claim:ineq-maxPPS}.

By Lemma
\ref{lem:PPS:unique-fixpoint},  and since $\sigma$ is optimal,
we know that 
 $(\lim_{i \rightarrow \infty}  x^i) = q^*_\sigma = g^*_\sigma = g^*$.
We also have that $(\lim_{i \rightarrow \infty}  y^i) =  q^*$,
where $q^*$ is the LFP of the maxPPS $x=P(x)$.
But then since  $x^i \leq y^i$ for all $i$,  
it follows that $g^* \leq q^*$.  But since we always have 
$q^* \leq g^*$, this implies  $g^* = q^*$.
  \end{proof}

\begin{theorem}
 Given a maxPPS, $x=P(x)$, 
with GFP $g^*$,
\begin{enumerate}
\item  Given $i \in [n]$, 
there is an algorithm that determines in P-time whether $g^*_i=0$,
and if $g^*_i > 0$ computes a deterministic policy
for the max player that achieves this.

\item Given any integer $j>0$, there is an algorithm that
computes a rational vector $v \leq g^*$ with $\|g^* -v\|_\infty \leq 2^{-j}$,
and also computes a deterministic policy $\sigma$,
such that $\|g^* - g^*_\sigma\| \leq 2^{-j}$, both 
in time polynomial in $|P|$ and $j$.
\end{enumerate}
\end{theorem}
\begin{proof}
\begin{enumerate}
\item First apply Proposition \ref{lem:prob1-ptime},
to remove variables $x_k$ with $g^*_k =1$,
and record the partial strategy for max on those types $T_k$ that achieves
 $g^*_k=1$. The residual maxPPS has $q^*=g^*$ by Lemma \ref{lem:maxPPS:unique-fixpoint}.
Thus, in order to decide whether $g^*_i = q^*_i = 0$,
we only need to apply the P-time algorithm from \cite{rmdp}
to decide whether the extinction probability $q^*_i > 0$.
And the AND-OR graph algorithm for this from \cite{rmdp} also 
 supplies a deterministic policy
 to achieve $q^*_i > 0$, if this is the case.
\item Again, we first apply Proposition \ref{lem:prob1-ptime},
so that, wlog, we can assume $g^* < \vone$.
Then by Lemma \ref{lem:maxPPS:unique-fixpoint},  $g^* = q^*$,
so that we only need to approximate the LFP $q^*$ 
of a maxPPS, $x=P(x)$, to within $j$ bits of
precision, and compute a $(2^{-j})$-optimal deterministic policy,
 in time polynomial in $|P|$ and $j$.
Algorithms that achieve precisely these two things were given in \cite{esy-icalp12}.  
\end{enumerate}
 \end{proof}

\section{Approximating the GFP of a minPPS in P-time}

\label{sec:minPPS-gfp-approx}

In this section we will show the following.

\begin{theorem} 
\label{thm:minPPS-approx-gfp}
Given a minPPS, $x=P(x)$ with $g^* < \vone$. If we use
  Generalized Newton's method, starting at $x^{(0)} := \vzero$, 
  with rounding parameter $h=j+2+4|P|$, then after $h$
  iterations, we have $\|g^*-x^{(h)} \|_\infty \leq 2^{-j}$.
\end{theorem}

In order to prove this theorem, we need some structural lemmas
about the GFPs of minPPSs, and their relationship to policies.
There need not exist any policies $\sigma$ with $g^*_\sigma=g^*$,
so we need policies that can, in some sense, act as ``surrogates''
for it. 
Recall that a policy $\sigma$ for a max/minPPS, $x=P(x)$, 
is called linear degenerate free (LDF)  if its 
associated PPS $x=P_\sigma(x)$ is an LDF-PPS.
When we consider the minPPS, $x=P(x)$, obtained from 
a BMDP for (non)reachability, after
eliminating types which cannot reach the target, the LFP 
$q^*_\sigma$ of $x=P_\sigma(x)$
for an LDF policy, $\sigma$, 
has $(q^*_\sigma)_i$ equal to $1$ minus the probability that,
starting with one object of type $i$, we reach the
target or else generate an infinite number of objects that can reach
the target under policy $\sigma$. It turns out that there is an LDF policy
$\sigma^*$ whose associated LFP is the GFP of the minPPS. 
Furthermore, it turns out that we can get an $\epsilon$-optimal policy by
following this LDF policy $\sigma^*$ with high probability and with
low probability following some policy that can reach the target from
anywhere.

\begin{lemma}  
\label{lem:minpps-cases-g-less-1}
If a minPPS $x=P(x)$ has $g^* < \vone$ then:

\begin{enumerate}
\item \label{lem:reach-pol-part}
 There is a deterministic
LDF policy $\sigma$ with $g^*_\sigma <\vone$,

\item \label{lem:ldf-lfp-geq-gfp-part} $g^* \leq q^*_\tau$, for any LDF policy $\tau$, and

\item \label{lem:min-ldf-part} 
There is a deterministic LDF policy $\sigma^*$ whose associated LFP, $q^*_{\sigma^*}$, has  
$g^*=q^*_{\sigma^*}$.\footnote{We remark for the reader's intuition (although we shall not prove it) that it can be shown that
any LDF policy $\sigma^*$ for a minPPS 
that satisfies $q^*_{\sigma^*} = g^* < \vone$ has the property that 
in the underlying BMDP $\sigma^*$ maximizes
the probability of the event of either reaching the target type or else
growing the population of types that can reach the target to infinity.}
\end{enumerate}
\end{lemma}

Rather than proving Lemma \ref{lem:minpps-cases-g-less-1} here,
we will instead prove later on a 
result for max-minPPSs
(Lemma \ref{lem:max-minpps-cases-g-less-1} of Section \ref{sec:zero-detect}),
which directly generalizes Lemma \ref{lem:minpps-cases-g-less-1}.

Note that the policy $\sigma^*$ 
described in part (3.) of Lemma \ref{lem:minpps-cases-g-less-1} 
is not necessarily optimal because even though
$g^*=q^*_{\sigma^*}$, there may be an $i$ with
$g^*_i=(q^*_{\sigma^*})_i < (g^*_{\sigma^*})_i=1$.

We will need also the following lemma from \cite{esy-icalp12}
on linearizations of max/minPPS.

\begin{lemma} \label{lem:3.5icalp12}
(\cite{esy-icalp12}, Lemma 3.5.) 
Let $x=P(x)$ be any max/minPPS.
Suppose that the matrix inverse 
$(I-B_\sigma(y))^{-1}$
  exists and is non-negative, for some policy $\sigma$, 
and some $y \in  \mathbb{R}^n$, where $B_\sigma$ is the Jacobian of $P_\sigma$. 
Then
\begin{itemize}
\item[(i)] $\mathcal{N}_\sigma(y)$ is defined, and is equal to the unique point $a \in \real^n$ 
such that  $P^y_\sigma(a) = a$.
\item[(ii)] For any vector $x \in \real^n$:\\
If $P^y_\sigma(x) \geq x$, then $x \leq \mathcal{N}_\sigma(y)$.\\
If $P^y_\sigma(x) \leq x$, then $x \geq \mathcal{N}_\sigma(y)$.
\end{itemize}
\end{lemma}

We will show now that Generalised Newton's Method (GNM) is well-defined.

\begin{lemma} \label{lem:gnm-ldf} Given a minPPS, $x=P(x)$, with GFP $g^* < \vone$, and 
given $y$ with $\vzero \leq y \leq g^*$,  there exists a 
deterministic LDF policy $\sigma$ with 
$P^y(\mathcal{N}_\sigma(y))=\mathcal{N}_\sigma(y)$,
the GNM operator $I(x)$ is defined at $y$, 
and for this policy $\sigma$, $I(y)=\mathcal{N}_\sigma(y)$. \end{lemma}
\begin{proof} We first show that there is an LDF policy $\sigma$ with 
$P^y(\mathcal{N}_\sigma(y)) = \mathcal{N}_\sigma(y)$. 
We will follow a proof structure similar to Lemma 3.14 from \cite{esy-icalp12}.

As there, we will be using policy improvement to show existence of a policy with desired properties
(but not as an algorithm to compute such a policy). 
Lemma \ref{lem:minpps-cases-g-less-1} (\ref{lem:reach-pol-part}.) 
gives the existence of a deterministic LDF policy given our assumption that $g^* <\vone$. So we start with such an LDF policy $\sigma_1$.  
Initially $i=1$, and we increment $i$ after each
policy improvement step.

In the general step $i$ we have a deterministic LDF policy $\sigma_i$.
By Lemma \ref{lem:minpps-cases-g-less-1} (2.), $g^* \leq q^*_{\sigma_i}$.
Since $y \leq g^* <\vone$, we have 
$y < \frac{1}{2}(\vone + g^*) \leq \frac{1}{2}(\vone+q^*_{\sigma_i})$.
Thus, we can apply Lemma \ref{lem:ldf-Ndefined} to the LDF PPS $x=P_{\sigma_i}(x)$
to conclude that $(I-B_{\sigma_i}(y))^{-1}$ exists and thus 
$\mathcal{N}_{\sigma_i}(y)$ is well-defined.
Let $z=\mathcal{N}_{\sigma_i}(y)$. By Lemma \ref{lem:3.5icalp12}, 
$P^y_{\sigma_i}(z) = z$. So $P^y(z) \leq z$. If
$P^y(z)=z$, then stop as we have a policy $\sigma$ with
$P^y(\mathcal{N}_{\sigma}(y)) = \mathcal{N}_{\sigma}(y)$. Otherwise,
there is a $j$ with $(P^y(z))_j < z_j$. $P_j(x)$ has form {\tt M} since
$(P^y(z))_j < (P^y_{\sigma_i}(z))_j$. 
Thus $P_j(x) = \min \{ x_k, x_{\sigma_i(j)} \}$ for some variable $x_k$,
 and $z_k < z_{\sigma_i(j)}$. 
Define $\sigma_{i+1}$ to be
$$\sigma_{i+1}(l) = \begin{cases} 
\sigma_i(l) \> \textrm{if  } l \not= j \\
k \> \textrm{if  } l = j \textrm{.}\end{cases}$$
We will first show that $\sigma_{i+1}$ is LDF, which
implies (as we argued for $\sigma_i$) that $\mathcal{N}_{\sigma_{i+1}}(y)$ is well-defined, and then we will show that $\mathcal{N}_{\sigma_{i+1}}(y) \leq z$ and $\mathcal{N}_{\sigma_{i+1}}(y) \not= z$. 

\begin{claim}
$\sigma_{i+1}$ is LDF.
\end{claim}
\begin{proof}
Suppose for a contradiction that $\sigma_{i+1}$ is not LDF. Then there
is a bottom SCC $S$ of $x=P_{\sigma_{i+1}}(x)$,  with
$(P_{\sigma_{i+1}})_S(x_S) \equiv B_Sx_S$ where $B_S$ is a stochastic
irreducible
matrix. %
$S$ must include $j$ and $k$ since
otherwise $\sigma_i$ would not be LDF. 
Note that since $S$ is a linear degenerate bottom SCC, 
for coordinates $j \in S$ we have $P^y_{\sigma_{i+1}}(x) = P_{\sigma_{i+1}}(x)$.
Now we have $(P_{\sigma_{i+1}}(z))_j = (P^y_{\sigma_{i+1}}(z))_j <
z_j$, but for every other coordinate $l \in S$ such that $l \not= j$,
$(P_{\sigma_{i+1}}(z))_l = (P^y_{\sigma_{i+1}}(z))_l = (P^y_{\sigma_i}(z))_l=z_l$.  

Thus $(P_{\sigma_{i+1}}(z))_S = (B_S z_S) \leq z_S$,  but  with 
$(P_{\sigma_{i+1}}(z))_j = (B_S z_S)_j < z_j$.
However, if we let $j'' \in S$ be a coordinate of $z_S$ with minimum value,
we see that $(P_{\sigma_{i+1}}(z))_{j''}$ is just a convex combination
of the other coordinates of $z_S$.  Thus the
coordinates of $z_S$ that appear in $(P_{\sigma_{i+1}}(z))_{j''}$
must all also have the minimum value, and thus
they are equal to $z_{j''}$.  Repeating this argument,
since $S$ is strongly connected, this implies that
all coordinates of $z_S$ are equal to $z_{j''}$.
But this contradicts the strict inequality 
$(P_{\sigma_{i+1}}(z))_j < z_j$ in the $j$ coordinate.
Thus, $\sigma_{i+1}$ must be LDF.
\end{proof}
Therefore, $\mathcal{N}_{\sigma_{i+1}}(y)$ is well-defined.

We know $(P^y_{\sigma_{i+1}}(z))_j < z_j$, but for every coordinate $l \not=
j$, $(P^y_{\sigma_{i+1}}(z))_l = z_l$. So we have
$P^y_{\sigma_{i+1}}(z) \leq z$. Lemma \ref{lem:3.5icalp12} (ii)
yields that $\mathcal{N}_{\sigma_{i+1}}(y) \leq z$. But
$\mathcal{N}_{\sigma_{i+1}}(y) \not= z$ because $P^y_{\sigma_{i+1}}(z)
\not= z$ whereas by Lemma \ref{lem:3.5icalp12} (i) from \cite{esy-icalp12}, we have
$P^y_{\sigma_{i+1}}(\mathcal{N}_{\sigma_{i+1}}(y)) =
\mathcal{N}_{\sigma_{i+1}}(y)$.

Thus the algorithm gives us a sequence of deterministic LDF policies $\sigma_1$,
$\sigma_2$, $\ldots$, with $\mathcal{N}_{\sigma_1}(y) \geq
\mathcal{N}_{\sigma_2}(y) \geq \mathcal{N}_{\sigma_3}(y) \geq \ldots$, 
where each step must decrease at least one coordinate of
$\mathcal{N}_{\sigma_i}(y)$. It follows that $\sigma_i \not= \sigma_j$
unless $i=j$. There are only finitely many deterministic policies. So the sequence
must be finite and the algorithm terminates. But it only terminates
when we reach a (deterministic) LDF policy $\sigma_i$ with
$P^y(\mathcal{N}_{\sigma_i}(y)) = \mathcal{N}_{\sigma_i}(y)$.

\smallskip
Recall that $I(x)$ is defined to be the unique optimal solution to the following 
LP:\footnote{As we explained in Section 2,
the constraints $P^y(a) \geq a$ can be written as linear inequalities.}
\begin{equation*} %
 \mbox{\em Maximize:} \  \ \sum_i a_i \ ; \quad \quad \quad
 \mbox{\em Subject to:} \quad   P^y(a) \geq a
\end{equation*}
We want to establish that $I(y)$ is well defined, i.e. that the LP has
a unique optimal solution, and that solution is $\mathcal{N}_{\sigma_i}(y)$.
First, $\mathcal{N}_{\sigma_i}(y)$ is a feasible solution to the LP
since $P^y(\mathcal{N}_{\sigma_i}(y)) = \mathcal{N}_{\sigma_i}(y)$. 
Furthermore, if $a$ is any feasible solution, i.e., if $P^y(a) \geq a$, 
then $P^y_{\sigma_i}(a) \geq P^y(a) \geq a$,
 so by Lemma \ref{lem:3.5icalp12} (ii), $a \leq \mathcal{N}_{\sigma_i}(y)$.
So $\mathcal{N}_{\sigma_i}(y)$ has the maximum value in every coordinate
among all feasible solutions.
Thus, it is the unique optimal solution to the LP, and  $
I(y) = \mathcal{N}_{\sigma_i}(y)$.   \end{proof}

Now we can show a halving result 
for GFPs of minPPS, similar to the following lemma that
was shown in \cite{esy-icalp12} for LFPs of PPS:

\begin{lemma} \label{lem:3.17icalp12}
(\cite{esy-icalp12}, Lemma 3.17)
If $x=P(x)$ is a PPS and we are given $a,b \in \mathbb{R}^n$ with $\vzero \leq a \leq b \leq P(b) \leq \vone$, and if 
the following conditions hold:
\begin{equation*}
\lambda>0 \quad  \mbox{and} \quad   b-a \leq \lambda (\textbf{1}-b) 
\quad \mbox{and}
\quad (I - B(a))^{-1} \ \mbox{exists and is non-negative,}
\end{equation*}
then $b - \mathcal{N}(a) \leq \frac{\lambda}{2} (\textbf{1}-b)$. 
\end{lemma}

We show an analogous lemma for GFP of minPPS.

\begin{lemma} \label{lem:minPPS-gfp=halving} Let $x=P(x)$ be a 
minPPS with GFP $g^* < \vone$. For any $\vzero \leq y \leq g^*$ and $\lambda > 0$, 
we have $I(y) \leq g^*$, and if:
$$ g^* - y \leq \lambda (\vone-g^*) $$
then
$$g^* - I(y) \leq \frac{\lambda}{2}(\vone-g^*)$$
\end{lemma}
\begin{proof}
By Lemma \ref{lem:gnm-ldf}, there is a deterministic LDF policy $\sigma$ with
$I(y)=\mathcal{N}_\sigma(y)$.  
We  apply Lemma \ref{lem:3.17icalp12} to the PPS $x=P_\sigma(x)$, 
with its variable $a$ replaced by our $y$, and with its variable $b$ replaced by
$g^*$.  Observe that $P_\sigma(g^*) \geq P(g^*)=g^*$ and that
$(I-B_\sigma(y))^{-1}$ exists and is non-negative.  Thus the
conditions of Lemma \ref{lem:3.17icalp12} hold and we can
conclude that $g^* - \mathcal{N}_\sigma(y) \leq \frac{\lambda}{2}
(\vone-g^*)$. 

All that remains is to show that $I(y) =
\mathcal{N}_\sigma(y) \leq g^*$. 
By Lemma \ref{lem:minpps-cases-g-less-1}(\ref{lem:min-ldf-part}.), there is
an LDF policy $\tau$ with $g^*=q^*_\tau$.
By Lemma \ref{lem:ldf-Ndefined} applied to the PPS $x=P_{\tau}(x)$,
the matrix $(I - B_\tau(y))^{-1}$ exists and is non-negative, and
$\mathcal{N}_\tau(y)$ is well-defined.
Any solution $a$ to the LP defining $I(y)$ has $P^x_{\tau}(a) \geq P^y(a) \geq a$,
so $a \leq \mathcal{N}_\tau(y)$ by Lemma \ref{lem:3.5icalp12}.  So
$I(y) \leq \mathcal{N}_\tau(y)$.
But we know from Lemma 3.4 of \cite{ESY12}, that for any PPS with LFP
$q^* < \vone$, if $y \leq q^*$, and $\mathcal{N}(y)$ is defined, then
$\mathcal{N}(y) \leq q^*$. 
Applying this lemma to the PPS $x=P_{\tau}(x)$, and since $q^*_\tau = g^* <\vone$
and $y \leq g^*$, we conclude that $\mathcal{N}_\tau(y) \leq q^*_\tau$.
Therefore, $I(y) \leq \mathcal{N}_\tau(y) \leq q^*_\tau = g^*$.
   \end{proof}

\begin{proof}[Proof of Theorem  \ref{thm:minPPS-approx-gfp}]
The theorem 
now follows by directly applying exactly the same inductive
argument as given in \cite{esy-icalp12} for the proof
of Theorem 3.21 there for the LFP.
Specifically, we start GNM at $x^{(0)} := {\mathbf 0}$,  and we let $x^{(k)}$ denote
the $k$'th iterate of GNM (applied on the minPPS, $x=P(x)$,  which has $g^* < \vone$), 
with rounding parameter $h := j + 2 + 4|P|$.
In all iterations we have $\vzero \leq x^{(k)} \leq g^*$.
Let $(\textbf{1} - g^*)_{\min} = \min_j (\textbf{1}-g^*)_j$.
As in \cite{esy-icalp12},  we claim, by induction on $k$, that for all $k \geq 0$: 

$$g^*-x^{(k)} \leq (2^{-k} + \sum_{i = 0}^{k-1} 2^{-(h+i)}) \frac{1}{(\textbf{1} - g^*)_{\min}} (\textbf{1} - g^*)$$

For the base case, $k=0$, we have
$$g^* - x^{(0)} = g^* \leq \textbf{1}   \leq \frac{1}{(\textbf{1} - g^*)_\text{min}} (\textbf{1} - g^*). $$

For the induction step, let us write for simplicity the claimed inequality
as $g^*-x^{(k)} \leq \lambda_{k} (\vone-g^*)$.
The induction hypothesis $g^*-x^{(k-1)} \leq \lambda_{k-1} (\vone-g^*)$ implies
by Lemma \ref{lem:minPPS-gfp=halving} that
$g^* - I(x^{(k-1)}) \leq \frac{\lambda_{k-1}}{2} (\vone-g^*)$.
The $k$th iterate $x^{(k)}$ satisfies $x_i^{(k)} \geq I(x^{(k-1)})_i -2^{-h}$
in every coordinate $i$. Therefore,
$g^*-x^{(k)} \leq \frac{\lambda_{k-1}}{2}(\vone-g^*)+2^{-h} \vone$
$\leq (\frac{\lambda_{k-1}}{2} +\frac{2^{-h}}{(\textbf{1} - g^*)_\text{min}})(\vone-g^*)$
$ = \lambda_{k} (\vone-g^*)$.

This shows the claimed inequality.
Since $\sum_{i = 0}^{k-1} 2^{-(h+i)} \leq 2^{-h+1}$, the inequality
implies that $g^*-x^{(k)} \leq (2^{-k} + 2^{-h+1}) \frac{\vone -q^*}{(\vone-q^*)_{\min}}$ for all $k$.

Let $\sigma^*$ be the (deterministic) LDF policy of Lemma \ref{lem:minpps-cases-g-less-1} (3.) with $q^*_{\sigma} =g^*$.
It was shown in \cite{ESY12} (Lemma 3.12), that if the LFP of a PPS $x=P(x)$ is $<1$
then the difference from 1 is at least $2^{-4|P|}$ in every coordinate.
Applying this lemma to the PPS $x=P_{\sigma^*}(x)$ and noting that $|P_{\sigma^*}| \leq |P|$,
we have that $\frac{\|(\vone-g^*)\|_\infty}{(\vone-g^*)_{\min}} \leq  \frac{1}{2^{-4|P|}} = 2^{4|P|}$. Therefore, $\|g^*-x^{(k)}\|_\infty \leq (2^{-k} + 2^{-h+1})2^{4|P|}$.
If we then let $k =h =j+4|P|+2$, 
we get that  $\| g^* - x^{(h)} \|_\infty \leq 2^{-j}$. \end{proof}

\section{Computing $\epsilon$-optimal 
policies for the GFP of  minPPSs  in P-time.}

\label{sec:eps-opt-minpps}

In this section we show how to construct an $\epsilon$-optimal
randomized policy for the GFP of a minPPS, $x=P(x)$,
in time polynomial in the input encoding size $|P|$ and $\log(1/\epsilon)$; 
note that there may not exist any
deterministic $\epsilon$-optimal policy (recall Example 3.2).  We also
consider BMDPs with the minimum non-reachability (i.e., maximum
reachability) objective and show how to construct a deterministic
non-static $\epsilon$-optimal strategy, again in time polynomial
in $|P|$ and $\log(1/\epsilon)$.

Given a minPPS, $x=P(x)$, with $n$ variables, we first preprocess it to identify and 
remove all variables with value 1 in the GFP; 
the policy can be set arbitrarily for all 
these nodes of type {\tt M} that have value 1.
So assume henceforth that $g^* <\vone$.
We first show how to find a deterministic LDF policy $\sigma$ 
with $\|g^*-q^*_\sigma\|_\infty \leq \frac{1}{2} \epsilon$. We will then use this policy to construct an $\epsilon$-optimal (randomized) policy.
Both steps are conducted in time polynomial in $|P|$ and $\log(1/\epsilon)$.

We use the following algorithm to construct a deterministic LDF policy $\sigma$ 
with $\|g^*-q^*_\sigma\|_\infty \leq \frac{1}{2} \epsilon$.
Note that each step of the algorithm runs in time polynomial in
$|P|$ and $\log(1/\epsilon)$.
\medskip

{\em Algorithm minPPS-$\epsilon$-policy1}
\begin{enumerate}
\item Compute, using GNM, a $\vzero \leq y \leq g^*$ with $\|g^*-y\|_\infty \leq 2^{-14|P|-3} \epsilon$.
\item Let $k:= 0$, and let $\sigma_0$ be a policy that has $P_{\sigma_0}(y)=P(y)$, i.e.,  $\sigma_0$ chooses for each type {\tt M} variable $x_i$
a variable $x_j$ of $P_i(x)$ that has the minimum value in the vector $y$.
\item \label{step-k} Compute $F_{\sigma_k}$, the set of variables that, in the dependency graph of $x=P_{\sigma_k}(x)$,
either are or can reach a variable $x_i$  
which either has form {\tt Q} or else $P_i(\vone) < 1$ or $P_i(\vzero) > 0$. 
Let $D_{\sigma_k}$ be the complement of $F_{\sigma_k}$ .
\item If $D_{\sigma_k} \neq \emptyset$, find a variable\footnote{
We will show that such a variable $x_i$ always exists whenever we reach this step.}  $x_i$ of type {\tt M} in $D_{\sigma_k}$ that has a choice $x_j$ in $F_{\sigma_k}$,
which isn't its current choice, such that $|y_i - y_j | \leq 2^{-14|P|-2} \epsilon$. Let $\sigma_{k+1}$ be the policy which chooses 
$x_j$ at $x_i$, and otherwise agrees with $\sigma_k$.  Let $k := k+1$, and return to step \ref{step-k}.
\item Else, i.e., if $D_{\sigma_k}$ is empty, output $\sigma_k$ and terminate.
\end{enumerate}

We will show that the final policy $\sigma$ computed by this algorithm
has the desirable property.
To start, we will extend the following lemma from \cite{esy-icalp12} to GFPs of minPPS.
\begin{lemma}[Lemma 4.4 from \cite{esy-icalp12}] \label{lipschitz} 
If $x=P(x)$ is a max/minPPS, and if $\vzero \leq y \leq q^*$, 
then
$\|P(y) - y\|_\infty \leq 2 \|q^* -y\|_\infty$. \end{lemma}

\begin{lemma} \label{lipschitz-GFP} 
If $x=P(x)$ is a minPPS, and if $\vzero \leq y \leq g^* < \vone$, 
then
$\|P(y) - y\|_\infty \leq 2 \|g^* -y\|_\infty$. \end{lemma}

\begin{proof}
Let $\sigma^*$ be the (deterministic) LDF policy of
Lemma \ref{lem:minpps-cases-g-less-1}(\ref{lem:min-ldf-part}.)
that has $q^*_{\sigma^*}=g^*$.
We apply lemma \ref{lipschitz} to the PPS $x=P_{\sigma^*}(x)$. 
This yields $\|P_{\sigma^*}(y) - y\|_\infty \leq 2 \|g^* -y\|_\infty$.

So for any $x_i$ not of form {\tt M}, we have $|P_i(y)-y_i|=|(P_{\sigma^*}(y) - y)_i| \leq 2 \|g^* -y\|_\infty$.
For $x_i$ of form {\tt M}, we have $P_i(x) \equiv \min \{x_j,x_k\}$ for some variables $x_j$, $x_k$. Suppose wlog that $y_j \leq y_k$ and thus $P_i(y) = y_j$. Then we have $P_i(y) = y_j \geq g^*_j - \|g^* -y\|_\infty \geq g^*_i - \|g^* -y\|_\infty$. Since $P(y) \leq P(g^*)=g^*$, $P_i(y) \leq g^*_i$. 
For $y_i$, we also have $g^*_i - \|g^* -y\|_\infty \leq y_i \leq g^*_i$. 
Therefore, $|P_i(y)-y_i| \leq \|g^* -y\|_\infty$.   \end{proof}

We use this lemma to bound $\|P_\sigma(y)-y\|_\infty$ for the policy $\sigma$
output by the algorithm.

\begin{lemma} \label{lem:algorithm-output}  Algorithm minPPS-$\epsilon$-policy1
always terminates in at most $n$ iterations of steps (3.)-(4.), 
and  outputs a deterministic LDF policy $\sigma$ 
with $\|P_\sigma(y)-y\|_\infty \leq 2^{-14|P|-2} \epsilon$. 
Since each iteration runs in time polynomial in $|P|$ and $\log(1/\epsilon)$,
so does the entire algorithm.\end{lemma}
\begin{proof} We first note that if the algorithm terminates, then it outputs an LDF policy since every variable in $F_{\sigma_k}$ satisfies condition (ii) of Lemma \ref{lem:ld-or-ldf} applied to the PPS $x=P_{\sigma_k}$.
We need to show that the algorithm terminates 
in the specified number of iterations, and that the final policy satisfies
the claimed bound.

At step 1 of the algorithm, we have $\|g^* - y \|_\infty \leq
2^{-14|P|-3} \epsilon$. Thus, by Lemma \ref{lipschitz-GFP}, we have
$\|P(y) - y\| \leq 2^{-14|P|-2} \epsilon$. It follows by the choice of
$\sigma_0$ that $\|P_{\sigma_0}(y) - y\| \leq 2^{-14|P|-2} \epsilon$.
Whenever we switch $x_i$ of form {\tt M} from $x_l$ to $x_j$ at an iteration
$k$, we have $|(P_{\sigma_{k+1}}(y)-y)_i| = |y_j - y_i| \leq
2^{-14|P|-2} \epsilon$ since we required that $|y_i - y_j| \leq  2^{-14|P|-2} 
\epsilon$.
So
for all $k$, $\|P_{\sigma_k}(y)-y\|_\infty \leq 2^{-14|P|-2} \epsilon$. Thus,
if the algorithm terminates, it outputs an LDF policy $\sigma$ with
$\|P_\sigma(y)-y\|_\infty \leq 2^{-14|P|-2} \epsilon$.

Next we show that if $D_{\sigma_k}$ is non-empty 
in some iteration $k$, then it contains an $x_i$ of form {\tt M} 
which has a choice $x_j$ in $F_{\sigma_k}$ with $|y_i - y_j| \leq 2^{-14|P|-2} \epsilon$.
Consider any $x_l$ in $D_{\sigma_k}$. 
Let $\sigma^*$ be a (deterministic) LDF policy such that $g^* = q^*_{\sigma^*}$ (which exists by 
Lemma \ref{lem:minpps-cases-g-less-1}(\ref{lem:min-ldf-part}.)).
$\sigma^*$ is an LDF policy so there is a path 
in the dependency graph of $x = P_{\sigma^*}(x)$ from $x_l$ to some $x_m$ which is not of 
form {\tt M} and is either of form {\tt Q} or has $P_m(\vone) < 1$ or $P_m(\vzero) > 0$. Thus $x_m$ is in $F_{\sigma_k}$.
So there must be a variable $x_i$ on the path from $x_l \in D_{\sigma_k}$ to $x_m \in F_{\sigma_k}$, 
with $x_i \in D_{\sigma_k}$, 
which depends directly on an $x_j$ which is next in the path and such that $x_j \in F_{\sigma_k}$.
So $(P_{\sigma^*}(x))_i$ contains a term with $x_j$ and $(P_{\sigma_k}(x))_i$ does not. 
Thus $x_i$ is of form {\tt M} and $(P_{\sigma_k}(x))_i \equiv x_j$.
But applying Lemma \ref{lipschitz} to the PPS $x=P_{\sigma^*}(x)$  gave us that  $\|P_{\sigma^*}(y) - y\|_\infty \leq 2 \|g^* -y\|_\infty$.
So $|y_i - y_j| \leq 2 \|g^* -y\|_\infty \leq 2^{-14|P|-2} \epsilon$. We can thus switch $x_i$ to $x_j$ in step \ref{step-k}.

Since no variable in $F_{\sigma_k}$ depends on a variable in $D_{\sigma_k}$, we have that $F_{\sigma_{k+1}} \supseteq F_{\sigma_k} \cup \{ x_j \}$. Since there are only $n$ variables, this means that for some $k \leq n$, all are in 
$F_{\sigma_k}$ and the algorithm terminates in at most $n$ iterations
of the steps (3.) and (4.).  
\end{proof}

\noindent We now show that the policy $\sigma$ has the desired property.

\begin{lemma} 
\label{lem:policy1}
The output policy $\sigma$ of Algorithm minPPS-$\epsilon$-policy1 
satisfies $q^*_\sigma <\vone$ and  $\|g^*-q^*_\sigma\|_\infty \leq \frac{1}{2} \epsilon$. 
\end{lemma}
\begin{proof}
We will show the lemma in two steps.
In Step 1, we will show that  $q^*_\sigma < \vone$.
In Step 2 we will use this to show that $\|g^*-q^*_\sigma\|_\infty \leq \frac{1}{2} \epsilon$.

\smallskip
{\bf Step 1}: $q^*_\sigma < \vone$.

This section of the proof is essentially identical to part of the proof of Theorem 4.7 in \cite{esy-icalp12}.
Suppose, for contradiction, that for some $i$, $(q^*_\sigma)_i = 1$. Then by results in \cite{rmc}, $x=P_\sigma(x)$ has a bottom strongly connected component $S$ with $q^*_S = \vone$. If $x_i$ is in $S$ then only variables in $S$ appear in $(P_\sigma)_i(x)$, so we write $x_S=P_S(x)$ for the PPS which is formed by such equations. We also have that $B_S(\vone)$ is irreducible and that the least fixed point solution of $x_S=P_S(x_S)$ is $q^*_S = \vone$.
Take $y_S$ to be the subvector of $y$ with coordinates in $S$. 

We will apply Theorem 4.6 (ii) from \cite{esy-icalp12},
which states that if a PPS $x=P(x)$ is strongly connected, has LFP $q^*=\vone$,
and a vector $y$ satisfies $\vzero \leq y < \vone=q^*$, then $(I-B(y))^{-1}$ exists, is nonnegative, and
$\| (I-B(y))^{-1} \|_\infty \leq 2^{4|P|}/(\vone-y)_{\min}$.
Applying this theorem to the PPS $x_S=P_S(x_S)$ with $\frac{1}{2}(y_S + \vone)$ in place of $y$,
gives that
$$\|(I -B_S(\frac{1}{2}(y_S + \vone)))^{-1}\|_\infty \leq  \frac{2^{4|P_S|}}{\frac{1}{2}(\vone-y_S)_{\min}}$$
But $|P_S| \leq |P|$ and $(\vone-y_S)_{\min} \geq (\vone-g^*)_{\min} \geq 2^{-4|P|}$. Thus
$$\|(I -B_S(\frac{1}{2}(y_S + \vone)))^{-1}\|_\infty \leq 2^{8|P| + 1}$$
From Lemma \ref{lem:3.3ESY12},
$P_S(y_S)-P_S(\vone)= P_S(y_S)-\vone= B_S(\frac{1}{2}(\vone + y_S))(y_S-\vone)$.
Hence, $(I - B_S(\frac{1}{2}(\vone + y_S)))(\vone-y_S)=\vone-y_S+P_S(y_S)-\vone = P_S(y_S) - y_S$,
and therefore:
 $$\vone-y_S = (I - B_S(\frac{1}{2}(\vone + y_S)))^{-1}(P_S(y_S) - y_S)$$
 Taking norms and re-arranging gives:
 $$\|P_S(y_S) - y_S)\|_\infty \geq \frac{\|\vone-y_S\|_\infty}{\|(I -B_S(\frac{1}{2}(y_S + \vone)))^{-1}\|_\infty} \geq \frac{2^{-4|P|}}{2^{8|P| + 1}} \geq 2^{-12|P|-1}$$
 However $\|P_S(y_S) - y_S)\|_\infty \leq \|P_\sigma(y) -y\|_\infty$ and  $\|P_\sigma(y)-y\|_\infty \leq 2^{-14|P|-2} \epsilon$ by Lemma \ref{lem:algorithm-output}. 
 This is a contradiction and so $q^*_\sigma < \vone$.

\medskip
{\bf Step 2}: $\|g^*-q^*_\sigma\|_\infty \leq \frac{1}{2} \epsilon$.

Now that we have $q^*_\sigma < \vone$, we can apply the following 
generalisation of Theorem 4.6 (i) of \cite{esy-icalp12}.
\begin{lemma}[cf Theorem 4.6 (i) of \cite{esy-icalp12}] \label{lem:ldf-normbounds}
If $x=P(x)$ is an LDF PPS with $q^* < \vone$ and $\vzero \leq y < \vone$, then
$B(\frac{1}{2}(y+q^*)))^{-1}$ exists, is nonnegative, and
$$\|(I-B(\frac{1}{2}(y+q^*)))^{-1}\|_\infty \leq  2^{10|P|} \text{max } \{2(\vone-y)_{\min}^{-1}, 2^{|P|}\}$$.
\end{lemma}
\begin{proof}
The {\em only} difference between Lemma \ref{lem:ldf-normbounds},
and the corresponding Theorem 4.6 (i) of \cite{esy-icalp12}, is that instead of
assuming that $\vzero < q^* < \vone$ there,  here we assume that $q^* < \vone$ {\em and} 
that $x=P(x)$ is an LDF PPS.
Furthermore, the {\em only} part of the proof of Theorem 4.6 (i) 
which employs the assumption that $q^* > \vzero$,
is Lemma C.8 of \cite{esy-icalp12}, for which we now establish the analogous Lemma 
\ref{lem:nowhere-near-the-end} below, 
under the alternative assumption that $x=P(x)$ is an LDF PPS.

\begin{lemma}\label{lem:nowhere-near-the-end}
 For any LDF-PPS, $x=P(x)$, with LFP $q^* < \vone$, for
any variable $x_i$ either
\begin{itemize}
\item[(I)]  the equation $x_i = P_i(x)$ is of form Q, or else $P_i(\vone) < 1$, or

\item[(II)] $x_i$  depends (directly or indirectly) on a variable $x_j$, 
such that $x_j=P_j(x)$ is of form Q,
or else $P_j(\vone) < 1$.
\end{itemize}\end{lemma}
\begin{proof} Consider the set $S$ of $x_i$ which do not satisfy either (I) or (II);
i.e., $S$ is the set of variables that cannot reach in the dependency graph any
node $x_j$ that has type {\tt Q} or is deficient ($P_j(\vone) <1$).
Suppose for a contradiction that $S$ is non-empty. No element $x_i$ in $S$ can depend on an element outside of $S$ since otherwise by transitivity of dependence it would satisfy (II).
Consider the LDF-PPS $x_S=P_S(x_S)$. Since this has no variables of form Q, $P_S(x_S)$ is affine i.e. we have $P_S(x_S) \equiv B_S(\vzero)x_S + P_S(\vzero)$. So for any fixed point $q_S$ of $x_S=P_S(x_S)$, we have $q_S=B_S(\vzero)q_S+P_S(\vzero)$. Since $x=P(x)$ is LDF, Lemma \ref{lem:ldf-Ndefined} yields that $(I-B_S(\vzero))^{-1}$ exists and is non-negative.
So we get $q_S = (I-B_S(\vzero))^{-1} P_S(\vzero)$, i.e. the linear system is non-singular,
it has a unique solution, so   
$x =P_S(x_S)$ has a unique fixed point. But because (I) does not hold for any variable in $S$, we have $P_S(\vone)=\vone$. So the unique fixed-point is $q^*_S=\vone$. This contradicts the assumption that $q^* < \vone$ and so $S$ is empty.
  \end{proof}
The rest of the proof of Lemma \ref{lem:ldf-normbounds} 
is word-for-word identical to the rest of the proof of Theorem 4.6 (i) from \cite{esy-icalp12}
(using Lemma \ref{lem:nowhere-near-the-end} instead of Lemma C.8 there),
so we will not repeat it here.\footnote{This is the part that starts after the proof of Lemma C.8 on page 37 of \cite{esy-icalp12} and finishes with the desired norm bound
inequality at the top of page 39 that completes the proof of 
part (i) of Theorem 4.6 there.}  
\end{proof}

\begin{corollary} \label{cor:what-we-actually-need} If $x=P(x)$ is an LDF PPS with $\vzero \leq q^* < \vone$, then
$$\|(I-B(q^*))^{-1}\|_\infty \leq 2^{14|P|+1}$$
\end{corollary}
\begin{proof}
We substitute $y:=q^*$ in Lemma \ref{lem:ldf-normbounds} along with the bound $(\vone-q^*)_{\min} \geq 2^{-4|P|}$ from Theorem 3.12 of \cite{ESY12}.
\end{proof}

We can now complete Step 2 of the proof of Lemma \ref{lem:policy1}.
By Lemma \ref{lem:3.3ESY12}, $B_\sigma(\frac{1}{2}(q^*_\sigma+y))(q^*_\sigma-y)=q^*_\sigma-P_\sigma(y)$. Rearranging this gives $q^*_\sigma-y=(I-B_\sigma(\frac{1}{2}(q^*_\sigma+y)))^{-1}(P_\sigma(y)-y)$. Taking norms, using the fact that $y \leq g^* \leq q^*_\sigma$
(since $\sigma$ is LDF), and applying Corollary \ref{cor:what-we-actually-need} 
on the PPS $x=P_\sigma(x)$ and
Lemma \ref{lem:algorithm-output}, we have:
\begin{eqnarray*} \|q^*_\sigma-y \|_\infty & \leq & \|(I-B_\sigma(\frac{1}{2}(q^*_\sigma+y)))^{-1}\|_\infty \|P_\sigma(y)-y \|_\infty \\
& \leq & \|(I-B_\sigma(q^*_\sigma))^{-1}\|_\infty \|P_\sigma(y)-y \|_\infty \\
& \leq & 2^{14|P|+1} 2^{-14|P|-2} \epsilon \\
& \leq & \frac{1}{2} \epsilon \end{eqnarray*}

\noindent By Lemma \ref{lem:minpps-cases-g-less-1}(\ref{lem:ldf-lfp-geq-gfp-part}.), we have $g^* \leq q^*_\sigma$. We have $y \leq g^* \leq q^*_\sigma$ and 
so $\|q^*_\sigma-g^*\|_\infty \leq \|q^*_\sigma-y\|_\infty \leq \frac{1}{2} \epsilon$.
\end{proof}

We define a randomized policy $\upsilon$ for the minPPS as follows.
Let $\sigma$ be the policy computed by Algorithm minPPS-$\epsilon$-policy1
and let $\tau$ be a (LDF) deterministic policy that satisfies 
$g^*_\tau < \vone$ (which can be computed in P-time by Proposition \ref{lem:prob1-ptime}).
For each type {\tt M} variable, the policy $\upsilon$ follows
with probability $2^{-28|P|-4} \epsilon$ the choice of policy $\tau$,
and with the remaining probability $1-2^{-28|P|-4} \epsilon$ the
choice of policy $\sigma$.

\begin{theorem} \label{thm:upsilon}
The policy $\upsilon$ satisfies $\|g^*-g^*_\upsilon\|_\infty \leq \epsilon$, i.e., it is $\epsilon$-optimal. \end{theorem}
\begin{proof}
We will show that $q^*_\upsilon$ is close to $q^*_\sigma$ and $g^*$,
and that $g^*_\upsilon=q^*_\upsilon$. 
First note that $P_\upsilon(g^*) \geq g^*$: for variables $x_i$ of the minPPS 
that have type {\tt L} or {\tt Q},
$(P_\upsilon(g^*))_i =P_i(g^*) = g^*_i$, and for variables $x_i$ of type {\tt M},
e.g. $x_i = \min(x_j,x_k)$, we have $g^*_i= \min(g^*_j,g^*_k)$,
and thus $(P_\upsilon(g^*))_i \geq g^*_i$.
Since $P_\upsilon(g^*) \geq g^*$, we have $q^*_\upsilon \geq g^*$ by Lemma \ref{lem:ldf-uniquefp}. We seek a $z$ close to $g^*$ such that $g^* \leq q^*_\upsilon \leq z$.

\begin{lemma} \label{lem:strictly-decreasing-point} For an LDF-PPS $x=P(x)$ with LFP $q^* < \vone$, let $z = q^*+ \delta (I-B(q^*))^{-1} \vone$ where $0 \leq \delta \leq 2^{-28|P|-3}$. Then $P(z) \leq z - \frac{1}{2} \delta \vone$. \end{lemma}
\begin{proof}
From Lemma \ref{lem:3.3ESY12}
\begin{equation} \label{eq1} B(\frac{1}{2}(q^*+z))(z-q^*)=P(z)-q^* \end{equation}
From the definition of $z$ we have $(I-B(q^*))(z-q^*)= \delta \vone$ and so 
\begin{equation} \label{eq2} B(q^*)(z-q^*) = z-q^*-\delta \vone \end{equation}
Subtracting (\ref{eq2}) from (\ref{eq1}), we obtain 
$$(B(\frac{1}{2}(q^*+z))- B(q^*)) (z-q^*) = P(z)-z + \delta \vone$$
If $P_i(x)$ is of form {\tt L}, the $i$th row of $B(x)$ does not depend on $x$ so we have $P_i(z)-z_i + \delta=0$ as required.
 
If $P_i(x)$ is of form {\tt Q}, wlog $P_i(x)=x_jx_k$ then we have $((B(\frac{1}{2}(q^*+z))- B(q^*)) (z-q^*))_i = \frac{1}{2}(z_j-q^*_j)(z_k - q^*_k)+ \frac{1}{2}(z_k - q^*_k)(z_j-q^*_j)= (z_j - q^*_j)(z_k - q^*_k)$. 
Thus we have $P_i(z)-z_i + \delta \leq \|z-q^*\|_\infty^2$. 
But here $\|z-q^*\|_\infty^2 \leq \delta^2 \|(I-B(q^*))^{-1} \|^2 \leq \delta^2 2^{28|P|+2} \leq \frac{1}{2} \delta$.
So we have $P_i(z) \leq z_i - \frac{1}{2}\delta$. 
\end{proof}
We apply this Lemma on the PPS $x=P_\sigma(x)$ with $\delta=2^{-28|P|-4} \epsilon$.
We get that for $z=q^*_\sigma + 2^{-28|P|-4} \epsilon (I-B_\sigma(q^*_\sigma))^{-1} \vone$, $P_\sigma(z) \leq z - 2^{-28|P|-3} \epsilon$.
For any $x \in [0,1]^n$, $P_\sigma(x) \in [0,1]^n$ and $P_\tau(x) \in [0,1]^n$,
so  $\|P_\sigma(x) - P_\tau(x)\|_\infty \leq 1$.
So, by definition of $\upsilon$, 
$\|P_\sigma(x) - P_\upsilon(x)\|_\infty  = 2^{-28|P|-3} \epsilon \|P_\sigma(x) - P_\tau(x)\|_\infty \leq 2^{-28|P|-3} \epsilon$.
In particular $\|P_\sigma(z) - P_\upsilon(z)\|_\infty  \leq 2^{-28|P|-3} \epsilon$.
And so we have $P_\upsilon(z) \leq P_\sigma(z) + 2^{-28|P|-3} \epsilon \leq z$. So by Lemma \ref{lem:ldf-uniquefp}, $q^*_\upsilon \leq z$.
Now we have $g^* \leq q^*_\upsilon \leq z$, and so
using Lemma \ref{lem:policy1} and Corollary \ref{cor:what-we-actually-need}, we get:
\begin{eqnarray*} \|q^*_\upsilon - g^*\|_\infty 	& \leq & \|z- g^*\|_\infty \\
												& \leq & \|q^*_\sigma-g^*\|_\infty + \|z-q^*_\sigma\|_\infty \\
												& \leq & \frac{1}{2} \epsilon + 2^{-28|P|-3} \epsilon \|(I-B_\sigma(q^*_\sigma))^{-1} \|_\infty \\
												& \leq & \frac{1}{2} \epsilon + 2^{-28|P|-3} \epsilon 2^{14|P|+1} \\
												& \leq & \epsilon \end{eqnarray*}
Recall that a PPS $x=P(x)$ has $g^*_i < 1$ if and only if either $P_i(\vone)<1$ or there is a path in the dependency graph from $x_i$ to an $x_j$ with $P_j(\vone) < 1$. If there is a path from $x_i$ to $x_j$ in the dependency graph of $x=P_\tau(x)$, then the same path exists in $x=P_\upsilon(x)$. Then by the same graph analysis that gave us $g^*_\tau < \vone$, we have $g^*_\upsilon < \vone$.  And so by Lemma \ref{lem:ldf-uniquefp}, $q^*_\upsilon=g^*_\upsilon$. So we have $\|g^*_\upsilon - g^*\| \leq \epsilon$. That is, $\upsilon$ is an $\epsilon$-optimal policy.  
\end{proof}

So, in a BMDP with minimum non-reachability (i.e., maximum reachability) objective,
we can construct efficiently, in time
polynomial in the encoding size of the BMDP
and $\log(1/\epsilon)$, a randomized static $\epsilon$-optimal strategy.
The following theorem shows that we can also construct a deterministic 
non-static strategy.

\begin{theorem} For a BMDP with minPPS  $x=P(x)$, and minimum non-reachability probabilities given by
the GFP $g^* < \vone$, the following 
deterministic non-static strategy $\alpha$ is also $\epsilon$-optimal starting
with one object of any type:

\begin{quote}Use policy $\sigma$ that is the output of Algorithm minPPS-$\epsilon$-policy1, 
until the population has size at least 
$\frac{2^{4|P|+1}}{\epsilon}$ for the first time; thereafter use a deterministic static policy $\tau$ 
such that $g^*_\tau  < \vone$.
\end{quote}
\end{theorem}
\begin{proof}
It follows from Lemma \ref{lem:either-extinct-or-infty} that
if we start the BP with an initial population of a single object with type corresponding to $x_i$, 
$1-(q^*_\sigma)_i$ is the probability that we either reach the target or else the population tends to infinity as time tends to infinity. 
So under the 
strategy $\alpha$, with at least probability $1-(q^*_\sigma)_i$, we either reach a population of more than $\frac{2^{4|P|+1}}{\epsilon}$ or we reach the target. 

Let $p$ be the probability that we reach the population $\frac{2^{4|P|+1}}{\epsilon}$ under $\sigma$ without reaching the target. Then $1-(q^*_\sigma)_i-p$ is the probability that we reach the target while staying under $\frac{2^{4|P|+3}}{\epsilon}$ population.

We claim that the probability of reaching the target from any population of size $m \geq \frac{2^{4|P|+1}}{\epsilon}$ using $\tau$ is at least $1-\frac{1}{2}\epsilon$. For a single object of type corresponding to $x_j$, this probability is $1-(g^*_\tau)_j \geq 2^{-4|P|}$. Since we can consider descendants of each member of the population independently, the probability that any of them reach the target is at least $1-(1-2^{-4|P|})^m \geq 1 - m2^{-4|P|} \geq \frac{1}{2}\epsilon$

The probability of reaching the target using $\alpha$ is then at least $1-(q^*_\sigma)_i-p + p(1-\frac{1}{2}\epsilon) \geq (1-g^*_i - \frac{1}{2} \epsilon) +p \frac{1}{2}\epsilon \geq 1 - g^*_i - \epsilon$. So $\alpha$ is $\epsilon$-optimal.
\end{proof}

\begin{corollary} Given a BMDP with a minimum non-reachability
(i.e. maximum reachability) objective, and any $\epsilon>0$, we can 
compute a static randomized $\epsilon$-optimal strategy or a deterministic non-static $\epsilon$-optimal strategy in time polynomial
in both the encoding size of the BMDP and in $\log(1/\epsilon)$. \end{corollary}

\section{P-time detection of GFP $g^*_i = 0$ for max-minPPSs
and reachability value 1 for BSSGs.}

\label{sec:zero-detect}

In this section we give a P-time algorithm for deciding whether the value of a BSSG 
reachability game is equal to $1$  (i.e., whether $g^*_i=0$ for a given max-minPPS), in which
case we show that the value is actually achieved
by a specific, memoryful but deterministic, strategy for the maximizing player,
which we can compute in P-time.  Thus 
there is no distinction between limit-sure vs. almost-sure reachability for BSSG.
Recall however that, as shown by Example \ref{example1},  
for a BSSG (or even BMDP) with reachability value equal to 1
there need not exist 
a {\em static} (even randomized) strategy that achieves almost-sure reachability.

Before presenting the algorithm, we need to extend the concept of LDF
policies to max-minPPSs and prove a basic lemma about them.  We define
a policy $\tau$ for the min player to be LDF if for all policies
$\sigma$ of the max player, $x=P_{\sigma,\tau}(x)$ is an LDF PPS.  The
following Lemma directly generalizes Lemma
\ref{lem:minpps-cases-g-less-1} to max-minPPSs, and indeed its proof
also provides the missing proof of Lemma
\ref{lem:minpps-cases-g-less-1} for minPPSs.

\begin{lemma}  
\label{lem:max-minpps-cases-g-less-1}
If a max-minPPS $x=P(x)$ has $g^* < \vone$ then:

\begin{enumerate}
\item 
\label{lem:mm-reach-pol-part}
 There is a deterministic 
LDF policy $\tau$ for the min player with $g^*_{*,\tau} <\vone$,

\item 
\label{lem:mm-ldf-lfp-geq-gfp-part} 
$g^* \leq q^*_{*,\tau'}$ for any LDF policy $\tau'$ for the min player, and

\item 
\label{lem:mm-min-ldf-part} 
There is a deterministic LDF policy $\tau^*$ for the min player
whose associated LFP, $q^*_{*,\tau^*}$, has  
$g^*=q^*_{*,\tau^*}$.
\end{enumerate}
\end{lemma}
\begin{proof} 
\mbox{}
\begin{enumerate}
\item 
Recall the P-time
  algorithm to detect whether $g^*_i=1$
(see Proposition \ref{lem:prob1-ptime} and its proof). That algorithm 
yields a deterministic policy
  $\tau$ with $g^*_{*,\tau} < \vone$. For all max player
  policies $\sigma$, we have $g^*_{\sigma,\tau} < \vone$. Lemma
  \ref{lem:inbetween-ldf} gives that all such PPSs
  $x=P_{\sigma,\tau}(x)$ are LDF. Thus, $\tau$ is LDF.

\item To prove part (\ref{lem:mm-ldf-lfp-geq-gfp-part}.), 
let $\tau'$ be any LDF policy for the min player.
Note that  $g^* = P(g^*) \leq P_{*,\tau'}(g^*)$.
So
 there exists a $\sigma$  with $g^* \leq P_{*,\tau'}(g^*) = P_{\sigma,\tau'}(g^*)$.
Namely, $\sigma$ simply chooses, for each equation $x_i = \max \{x_j, x_k \}$,
the neighbor $x_j$ if $g^*_i = \max \{g^*_j, g^*_k\} = g^*_j$, and otherwise chooses $x_k$,
since then $g^*_i =  \max \{g^*_j, g^*_k\} = g^*_k$.
By applying Lemma \ref{lem:ldf-uniquefp} 
to the LDF-PPS $x=P_{\sigma,\tau'}(x)$, with $y := g^*$,
we get $g^*\leq q^*_{\sigma,\tau'} \leq q^*_{*,\tau'}$.

\item %
We will first show there exists a deterministic LDF policy $\tau^*$ such that $P(q^*_{*,\tau^*})=q^*_{*,\tau^*}$,
and we will then argue that $g^* =  q^*_{*,\tau^*}$.
This proof is
somewhat similar to the proof of Lemma 3.14 from \cite{esy-icalp12},
as well as 
 the proof of Lemma \ref{lem:gnm-ldf} in this paper. 
The proof uses policy improvement to demonstrate the existence
of the claimed policy (but not as an algorithm to compute it).

Part (\ref{lem:mm-reach-pol-part}.) of this Lemma yields that there is a
deterministic LDF policy $\tau$ with $g^*_{*,\tau} < \vone$.
Thus we have  $q^*_{*,\tau} \leq g^*_{*,\tau} < \vone$.  

At step 1, we start policy improvement with $\tau_1 := \tau$.  
At step $i$, we have a deterministic LDF policy $\tau_i$ with $q^*_{*,\tau_i} < \vone$. 
If $P(q^*_{*,\tau_i})=q^*_{*,\tau_i}$, stop
(because then, as we will see, policy $\tau_i$ satisfies $g^* = q^*_{*,\tau_i}$).
Otherwise, there must be an $x_j$ with $P_j(q^*_{*,\tau_i}) < (q^*_{*,\tau_i})_j$, 
because $P(q^*_{*,\tau_i}) \leq P_{*,\tau_i}(q^*_{*,\tau_i}) = q^*_{*,\tau_i}$.
Note that $x_j$ belongs to min, because otherwise we would have $P_j(q^*_{*,\tau_i}) = (P_{*,\tau_i}(q^*_{*,\tau_i}))_j$. So we must have 
$P_j(x) = \min \{ x_k , x_{\tau_i(j)} \}$ for some $x_k$.
Then set $\tau_{i+1}$ to be the policy that selects $x_k$ at $x_j$ for some $x_j$ with $P_j(q^*_{*,\tau_i}) < (q^*_{*,\tau_i})_j$, 
but is otherwise identical.
We need first to show that $\tau_{i+1}$ is LDF. 

\begin{claim}
The policy $\tau_{i+1}$ is LDF. 
\end{claim}
\begin{proof}
Suppose for a contradiction that $\tau_{i+1}$ is not LDF.
Then there  exists a policy $\sigma$ for max such that a bottom SCC 
$S$ of $x=P_{\sigma, \tau_{i+1}}(x)$ is linear degenerate.
This SCC must contain $x_j$ and $x_k$ since otherwise $S$ would also be a linear degenerate SCC of $x=P_{\sigma, \tau_{i}}(x)$ and so $\tau_i$ would also not 
be LDF.

By construction, $P_{*,\tau_{i+1}}(q^*_{*,\tau_i}) \leq q^*_{*,\tau_i}$
with strict inequality $(P_{*,\tau_{i+1}}(q^*_{*,\tau_i}))_j < (q^*_{*,\tau_i})_j$ in the coordinate $j \in S$. 

Let $j'' = \arg\min_{j' \in S}  (q^*_{*,\tau_i})_{j'}$ be any coordinate of the vector $(q^*_{*,\tau_i})_S$
which has minimum value.
We have $(P_{*,\tau_{i+1}}(q^*_{*,\tau_i}))_{j''} \leq (q^*_{*,\tau_i})_{j''}$.

We claim that any $x_{j'} \in S$ that appears in $(P_{*,\tau_{i+1}}(x))_{j''}$ must also have this minimum value, i.e. $(q^*_{*,\tau_i})_{j'}=(q^*_{*,\tau_i})_{j''}$ .
 If $(P_{*,\tau_{i+1}}(x))_{j''}$ has form {\tt L}, then  $(P_{*,\tau_{i+1}}(q^*_{*,\tau_i}))_{j''}$ is just a convex combination of coordinates 
of $(q^*_{*,\tau_i})_S$. If any of these are bigger than their minimum value then we would have $(P_{*,\tau_{i+1}}(q^*_{*,\tau_i}))_{j''} > (q^*_{*,\tau_i})_{j''}$ which is a contradiction.
If $(P_{*,\tau_{i+1}}(x))_{j''}$ belongs to min, then it is equal to $x_{j'} \in S$. Again we have $(q^*_{*,\tau_i})_{j'} \leq  (q^*_{*,\tau_i})_{j''}$ which is an equality by minimality of $j''$.
If $(P_{*,\tau_{i+1}}(x))_{j''}$ belongs to max, then we must have $(q^*_{*,\tau_i})_{j'} \leq (P_{*,\tau_{i+1}}(x))_{j''} \leq (q^*_{*,\tau_i})_{j''}$ which again is an equality by minimality of $j''$.
Lastly  $(P_{*,\tau_{i+1}}(x))_{j''}$ can not have form {\tt Q} since 
$S$ is linear degenerate in $x=P_{\sigma,\tau_{i+1}}(x)$. 
This completes the proof of the claim that such $x_{j'}$ are also minimal.

Since $S$ is strongly-connected in $x=P_{\sigma, \tau_{i+1}}(x)$, $x_j$ and $x_k$ depend (directly or indirectly) on $x_{j''}$ in $x=P_{\sigma, \tau_{i+1}}(x)$ and so in 
$x=P_{*, \tau_{i+1}}(x)$ as well. By induction, we have that $(q^*_{*,\tau_i})_{j}=(q^*_{*,\tau_i})_{j''}$ and $(q^*_{*,\tau_i})_{k} = (q^*_{*,\tau_i})_{j''}$. But now we have $(q^*_{*,\tau_i})_{j}=(q^*_{*,\tau_i})_{k}$. This contradicts $(q^*_{*,\tau_i})_{k} < (q^*_{*,\tau_i})_{j}$ which is why we switched $x_j$ to $x_k$ in $\tau_{i+1}$.  Thus $\tau_{i+1}$ is LDF.
\end{proof}

By construction of $\tau_{i+1}$, we have $P_{*,\tau_{i+1}}(q^*_{*,\tau_i}) \leq q^*_{*,\tau_i}$ with strict inequality in the $j$ coordinate.

There is a policy $\sigma$ for max that has $q^*_{*,\tau_{i+1}}= q^*_{\sigma,\tau_{i+1}} $
For such a $\sigma$, we have $P_{\sigma,\tau_{i+1}}(q^*_{*,\tau_i}) \leq q^*_{*,\tau_i}$ with strict inequality 
in the $j$ coordinate. 
By Lemma \ref{lem:ldf-uniquefp}, applied to the LDF-PPS, $x=
P_{\sigma,\tau_{i+1}}(x)$ with $y:=   q^*_{*,\tau_i}$, 
this implies $q^*_{\sigma,\tau_{i+1}} \leq q^*_{*,\tau_i}$. So $q^*_{*,\tau_{i+1}} \leq q^*_{*,\tau_i}$.

This cannot be an equality since $P_{*,\tau_{i+1}}(q^*_{*,\tau_i}) \not= q^*_{*,\tau_i}$. So the algorithm cannot 
revisit the same policy, i.e., for all $k \neq i$,  we have $\tau_k \neq \tau_i$.
Since there are only finitely many deterministic policies, the algorithm must terminate. 

So the algorithm terminates with a deterministic LDF policy $\tau^*$ with $P(q^*_{*,\tau^*})=q^*_{*,\tau^*}$.
All that remains is to show that  $g^*=q^*_{*,\tau^*}$. 
$P(q^*_{*,\tau^*})=q^*_{*,\tau^*}$, so $q^*_{*,\tau^*}$ is a fixed point of $x=P(x)$ and the GFP $g^*$ satisfies
$g^* \geq q^*_{*,\tau^*}$.
By part (\ref{lem:mm-ldf-lfp-geq-gfp-part}.) of this Lemma, $g^* \leq q^*_{*,\tau^*}$.
Therefore, $g^* = q^*_{*,\tau^*}$.
\end{enumerate}
 \end{proof}

We are now ready to give the algorithm.
First, we identify and remove 
all variables $x_i$ with $g^*_i =1$ 
(which we can do in P-time, by
Proposition \ref{lem:prob1-ptime}).
Let $X$ be the set of all variables in the remaining max-minPPS $x = P(x)$ 
in SNF form,  with GFP $g^* < \vone$.   The algorithm is described in 
Figure \ref{fig:alg-g1-max-minpps},
and Theorem \ref{thm:minpps-gfp-zero-detection} shows that it computes the set $\{x_i \in X \mid g^*_i = 0\}$.

\begin{figure}
\begin{enumerate}
\item Initialize $S := \{ \ x_i \in X \mid$ $P_i(\vzero) > 0$, i.e., $P_i(x)$ contains a constant term 
$\}$.
\item \label{and-or-1} Repeat the following until neither are applicable:
\begin{enumerate}
\item If a variable $x_i$ is of form {\tt L} or {\tt M$_{\max}$}
and $P_i(x)$ contains a variable that is already in $S$, add $x_i$ to $S$.
\item If a variable $x_i$ is of form {\tt Q} or {\tt M$_{\min}$} and both 
variables in $P_i(x)$ are already in $S$, add $x_i$ to $S$.
\end{enumerate}
\item Let $F := \{ \ x_i \in X-S \mid $   
$P_i(\vone) < 1$, or $P_i(x)$ has form ${\tt Q}$ $\}$.
\item \label{reach2} repeat the following until no more variables can be added:
\begin{enumerate}
\item If a variable $x_i \in X-S$ is of form {\tt L} or {\tt M$_{\min}$} and $P_i(x)$ contains a term whose variable is in $F$, add $x_i$ to $F$.
\item If a variable $x_i \in X-S$ is of form {\tt M$_{\max}$} and both variables in $P_i(x)$ are in $F$, add $x_i$ to $F$. 
\end{enumerate}
\item If $X=S \cup F$, terminate and output $F$.
\item \label{iterate} Otherwise set $S := X - F$ and return to step \ref{and-or-1}.
\end{enumerate}
\caption{P-time algorithm for computing $\{x_i \in X \mid g^*_i = 0 \}$ for a max-minPPS with GFP $g^* < \vone$.  \label{fig:alg-g1-max-minpps}}
\end{figure}

\begin{theorem} 
\label{thm:minpps-gfp-zero-detection}
The procedure in Figure \ref{fig:alg-g1-max-minpps}, applied to a max-minPPS $x=P(x)$ with $g^* <\vone$, 
always terminates and outputs precisely the set of variables $\{ x_i \in X \mid g^*_i=0 \}$, in time polynomial in $|P|$. 
Furthermore we can compute in P-time 
a deterministic policy $\sigma$ for the max player such that 
$(g^*_{\sigma,*})_i > 0$ for all the variables $x_i$ in $\{ x_i \in X \mid g^*_i>0 \}$.
\end{theorem}
\begin{proof}
Firstly we show that all variables $x_i$ in the output $F$ have $g^*_i=0$. To do this 
we construct an LDF policy $\tau^*$ for the min player such that $(q^*_{*,\tau^*})_i=0$, and then argue that $g^*_i=0$.

By Lemma \ref{lem:max-minpps-cases-g-less-1}(\ref{lem:reach-pol-part}.),
there is an LDF policy $\tau$ with $g^*_{*,\tau} < \vone$. We define $\tau^*$ 
so that it agrees
with $\tau$ on variables in $S$. For a variable $x_i$ of form {\tt M$_{\min}$} in
$F$, policy $\tau^*$ chooses a variable of $P_i(x)$ that was already in $F$
and which caused $x_i$ to be added to $F$ in step \ref{reach2}. 
So, for any fixed policy $\sigma$ for the max player, every variable 
$x_i$ in $F$
depends (directly or indirectly) in the PPS $x=P_{\sigma,\tau^*}(x)$ on a variable $x_j$ in $F$ with $P_j(\vone)
< 1$ or of form {\tt Q}. So every variable in $F$ satisfies one of two
out of the three conditions in Lemma \ref{lem:ld-or-ldf} part (ii), with respect to $x=P_{\sigma,\tau^*}(x)$. Now consider
a variable $x_i$ in $S$.  $\tau$ is LDF, so for the fixed 
policy $\sigma$ for the max player,  there is a path in the
dependency graph of $x=P_{\sigma,\tau}(x)$ from $x_i$ to an $x_j$ which
satisfies one of the three conditions in Lemma \ref{lem:ld-or-ldf} part (ii). If
this path does not contain any variable in $F$, then it is also a path
in the dependency graph of $x=P_{\sigma,\tau^*}(x)$. If it does, then $x_i$
depends on a variable in $F$, so by transitivity of dependence, it
also depends on a variable which satisfies one of the conditions in
Lemma \ref{lem:ld-or-ldf} (ii). The policy $\sigma$ for the max player
was chosen arbitrarily. So $\tau^*$ is LDF.
 
Next we need to show that $(q^*_{*,\tau^*})_F = \vzero$. 
Since for all variables $x_i$ in $F$, 
$P_i(x)$ does not contain a constant term, we have $(P_{*,\tau^*}(\vzero))_F=\vzero$. Note 
that all variables in $F$ of type {\tt L}, {\tt M$_{\min}$} 
and {\tt M$_{\max}$} depend directly only on variables
of $F$ in $x=P_{*,\tau^*}(x)$, and every variable of type {\tt Q} depends on some variable in $F$ 
(otherwise it would have been added to $S$ in the previous step 2). 
It follows then
by an easy induction that $(P^k_{*,\tau^*}(\vzero))_F=\vzero$ for all $k$. So $(q^*_{*,\tau^*})_F = \vzero$.

Since $\tau^*$ is an LDF policy,
 Lemma \ref{lem:max-minpps-cases-g-less-1}(\ref{lem:ldf-lfp-geq-gfp-part}.) tells us that $q^*_{*,\tau^*} \geq g^*$. 
Since $g^* \geq \vzero$, we have that $g^*_F=\vzero$ as required.

Finally, we need to show that $g^*_S > \vzero$, and specify a policy $\sigma$
for the max player that ensures this. We will only specify the
policy for the {\tt M$_{\max}$} nodes in $S$; the choice for
the other {\tt M$_{\max}$} nodes does not matter and can be arbitrary.  
To show the claim, we need to show
inductively that when we add a variable $x_i$ to $S$, if all variables
$x_j$ already in $S$ have $g^*_j > 0$ (and $(g^*_{\sigma,*})_j >0$),
then $g^*_i > 0$ (and $(g^*_{\sigma,*})_i >0$).

For the basis case, note that for the variables $x_i$ added to $S$ in step 1,
$P_i(x)$ contains a positive constant term, hence $g^*_i \geq P_i(\vzero) > 0$.
Consider now the variables $x_i$ added in an execution of step 2.
If $P_i(x)$ is of form {\tt L} then it contains an $x_j$ that was added
earlier to $S$; hence $g^*_j > 0$ (and $(g^*_{\sigma,*})_j >0$), 
and thus $g^*_i=P_i(g^*) > 0$ (and $(g^*_{\sigma,*})_i >0$). 
If $x_i=x_jx_k$ or $x_i = \min \{x_j,
x_k\}$ for some $x_j$,$x_k$, then both $x_j$,$x_k$ were added earlier
to $S$; hence $g^*_j >0$ and $g^*_k > 0$, and thus $g^*_i=P_i(g^*) > 0$
(and similarly, $(g^*_{\sigma,*})_i >0$).  
If $x_i = \max \{x_j, x_k\}$, then at least one of $x_j, x_k$  
was added earlier to $S$, say $x_j$, hence $g^*_j > 0$ and $(g^*_{\sigma,*})_j >0$.
Let the policy $\sigma$ choose $\sigma(x_i) = x_j$;
then $g^*_i \geq (g^*_{\sigma,*})_i =  (g^*_{\sigma,*})_j >0$.

Consider now the set $R$ of variables added to $S$ in an execution of step 6
and assume inductively that all the variables $x_j$ assigned so far to $S$
have $g^*_j >0$ (and $(g^*_{\sigma,*})_j >0$). 
Since the variables $x_i$ of $R$ were not added to 
$F$ in steps 3-4, they all satisfy $P_i(\vone)=1$, they are not of type {\tt Q},
every variable of type {\tt L} or {\tt M$_{\min}$} does not depend directly on 
any variable in $F$,
and every variable of type {\tt M$_{\max}$} depends 
directly on at least one variable that is not in $F$. 
Let the policy $\sigma$ choose actions for variables
in $S$ as before, and for each 
variable $x_i$ of type {\tt M$_{\max}$} in $R$ let $\sigma$ choose an
arbitrary variable of $P_i(x)$ that is not in $F$.
(For the variables $x_i$ of type {\tt M$_{\max}$} that are in $F$, 
the choices of $\sigma$ do not matter at this point.)
Then the dependency graph  of $x=P_{\sigma,*}(x)$ has no edges from $R$ to $F$.

We claim that $(g^*_{\sigma,*})_R >\vzero$. 
Let $\tau'$ be an LDF policy for the min player
in the minPPS, $x=P_{\sigma,*}(x)$,
such that $g^*_{\sigma,*} = q^*_{\sigma,\tau'}$
(we know $\tau'$ exists by Lemma \ref{lem:minpps-cases-g-less-1}(\ref{lem:reach-pol-part})).
Let $U$ be the set of variables $x_i \in R$ with
$(q^*_{\sigma,\tau'})_i=0$. We need to show that $U$ is empty.
Consider any $x_j \in U$.  We claim that any variable $x_k$ appearing 
in $(P_{\sigma,\tau'}(x))_j$ is in $U$.  We know that the dependency graph
of $x=P_{\sigma,*}(x)$ has no edges from $R$ to $F$, so $x_k \notin
F$. It remains to show that $(q^*_{\sigma,\tau'})_k=0$, since then by
inductive assumption $x_k \notin S$, and so we must have $x_k \in R$, and thus 
$x_k \in U$.  
If $x_j$ is
of type {\tt M}, then $(P_{\sigma,\tau'}(x))_j \equiv x_k$, so
$(q^*_{\sigma,\tau'})_k=(q^*_{\sigma,\tau'})_j=0$.  If $x_j$ has type
{\tt L}, then if $(q^*_{\sigma,\tau'})_k > 0$ then
$(q^*_{\sigma,\tau'})_j > 0$ so we must have $(q^*_{\sigma,\tau'})_k =
0$.  Since $x_j \in U$, it can not have type {\tt Q} since such
variables not in $S$ were put in $F$. So $x_k \in U$.

We have that variables in $U$ depend in the PPS
$x=P_{\sigma,\tau'}(x)$ only on other variables in $U$.  However, no
variable $x_j$ that satisfied one of the three conditions of Lemma
\ref{lem:ld-or-ldf} (ii) is in $U \subseteq R$ since it would have
been put in $S$ or $F$ in an earlier step.  Since
$x=P_{\sigma,\tau'}(x)$ is an LDF-PPS, for any $x_i$, by Lemma
\ref{lem:ld-or-ldf}, there is a path from $x_i$ to such an $x_j$.  If
$x_i \in U$, then this path must remain entirely in $U$ which is a
contradiction.  Therefore $U$ is empty and we have that
$(g^*_{\sigma,*})_R >\vzero$ as required.

The fact that the algorithm runs in P-time follows easily from 
the fact that each iteration of the outer loop adds at least
one element to $S$, and no element is ever removed. The
individual steps of the algorithm are each easily computable in P-time,
by performing AND-OR reachability on the dependency graph.
  \end{proof}

We remark that the policy $\tau^*$ for the min player constructed in the
proof of Theorem \ref{thm:minpps-gfp-zero-detection} does not necessarily
ensure value 0 in the GFP for a variable $x_i$ with $g^*_i=0$ 
(i.e., it is possible that $(g^*_{*,\tau^*})_i >0$). In fact, there may not
exist any such policy (deterministic or randomized) ensuring value 0 for the min player
in a max-minPPS (or even a minPPS).
Similarly, in a BSSG (or even RMDP) with
optimal non-reachability value 0 (i.e. reachability value 1),
there may not exist any optimal static strategy 
for the player that wants to minimize the non-reachability probability; 
recall Example \ref{example1}.
We show however that we can construct a non-static optimal deterministic strategy.

\begin{theorem} There is a non-static deterministic optimal strategy for 
the player minimizing
the probability of {\em not} reaching a target type in a BSSG,
if the value of not reaching the target is $0$. \end{theorem}

\begin{proof}
Let $x=P(x)$ be the max-minPPS for the given BSSG, whose GFP $g^*$
gives the non-reachability values.
Let $Z = \{ x_i | g^*_i =0 \}$ be the final value of the set $F$ that is
returned by the algorithm of Fig. 1.
Let $\tau^*$ be the LDF policy for player $\min$ constructed in
the proof of Theorem \ref{thm:minpps-gfp-zero-detection}
that has the property that
$g^*_i = 0$ iff $(q^*_{*,\tau^*})_i = 0$. 
Recall that  $\tau^*$ selects for each type {\tt M$_{\min}$} variable
$x_i \in Z$ a variable $x_j$ of $P_i(x)$ that was
added earlier to $F$ (and hence is also in $Z$). 
From Proposition \ref{lem:prob1-ptime}, we can also compute in P-time 
an LDF policy $\tau$ with $g^*_{*,\tau} < \vone$. We combine $\tau^*$ and $\tau$ 
in the following non-static policy:

We designate one member of our initial population with type in $Z$ to
be the queen. The rest of the population are workers. We use policy
$\tau^*$ for the queen and $\tau$ for the workers. In following
generations, if we have not reached an object of the target type, we
choose one of the children in $Z$ of the last generation's queen
(which we next show must exist) to be the new queen. Again, all other
members of the population are workers.

We first show the policy is well defined, i.e., we can always find a new queen
as prescribed. If $g^*_i=0$, then $P_i(g^*)=(P_{*,\tau^*}(g^*))_i=g^*_i=0$. 
If $P_i(x)$ has form {\tt L} then all $x_j$ appearing in $P_i(x)$ have 
$g^*_j=0$ and there is no constant term. If $P_i(x)$ has form {\tt Q} then 
at least one $x_j$ in $P_i(x)$ will have $g^*_j=0$. 
If $P_i(x)$ has form {\tt M}$_{\min}$, 
then the $x_j=\tau^*(x_i)$ in $(P_{*,\tau^*}(x))_i$ has $g^*_j = 0$. 
Finally, if $P_i(x)$ has form {\tt M}$_{\max}$, then 
for all variables $x_j$ in $P_{*,\tau^*}(x)$ we have $g^*_j = 0$.
In other words, using 
$\tau^*$, an object of a type in $Z$ has offspring which either 
includes the target or an object of a type in $Z$. 
Thus the next generation always includes a potential choice of queen.

Next we show that if we never reach the target type, the queen has
more than one child infinitely often with probability 1. Indeed we
claim that with probability at least $2^{-|P|}$ within the next $n$
steps, either the queen has more than one child or we reach the
target. For this purpose, we define inductively for every variable $x_i \in Z$
a (directed) tree $T_i$ with root $x_i$, which shows why $x_i$ was added to
$F$ in the final iteration of the algorithm.
If $P_i(\vone) <1$ or $x_i$ has type {\tt Q} then $T_i$ is a single node
labeled $x_i$.
If $x_i$ has type {\tt L} (respectively {\tt M$_{\min}$}) and was added in step 4 
because of variable $x_j \in P_i(x)$ that was already in $F$
(resp., where $x_j = \tau^*(x_i)$), then $T_i$ consists of the
edge $x_i \rightarrow x_j$ and the subtree $T_j$ rooted at $x_j$.
If $x_i$ has type {\tt M$_{\max}$} then $T_i$ contains edges $x_i \rightarrow x_j$
for all $x_j \in P_i(x)$ and a subtree $T_j$
hanging from each $x_j$.

Suppose that in some step the queen is an object corresponding to $x_i \in Z$.
Then with positive probability (in fact probability at least $2^{-|P|}$),
in the next (at most) $n$ steps, the process will follow a root-to-leaf path
of the tree $T_i$, regardless of the strategy of the max player: 
whenever the path is at a node of type {\tt L}, 
the process follows the edge to the (unique) child (which becomes the new queen) 
with the probability of the corresponding transition of the BSSG;
when it is at a node of type {\tt M$_{\min}$}, it follows necessarily the edge 
to its child because we are using policy ${\tau^*}$ for the queen; and when it is at 
a node of type {\tt M$_{\max}$}, it follows an edge selected by the max player.
Thus, with probability at least $2^{-|P|}$, the process arrives at a leaf
of $T_i$. If the leaf corresponds to a variable $x_j$ with $P_j(\vone) <1$ then
the process has reached a target type. If the leaf corresponds to a variable
of type {\tt Q} then the queen generates two children.

Thus, if the queen never reaches the target throughout the process,
then the queen will generate more than one child infinitely often
with probability 1.

By our choice of the policy $\tau$ followed by the workers, $ g^*_{*,\tau} < \vone$.  
The descendants of a worker of type $x_i$ have positive probability $(\vone-g^*_{*,\tau})_i>0$
of reaching the target regardless of the strategy of the max player
(this probability is $ \geq 2^{-4|P|}$  by Lemma 3.20 of \cite{esy-icalp12} 
applied to the maxPPS $x=P_{*,\tau}(x)$). 
For each worker descended from the queen these
probabilities are independent. 
So with probability 1, one of them will have descendants
that reach the target. Thus we reach the target with probability
1. \end{proof}

\section{Approximating the value of BSSGs and the GFP of max-minPPSs}

\label{sec:bssg}

In this section we build on the prior results to 
show that the {\em value} of a Branching Simple Stochastic
Game (BSSG) with reachability as the objective can be approximated in TFNP.  
Equivalently, we show that the GFP, $g^*$, of a max-minPPS can be 
approximated in TFNP.

We will first show that we can test in polynomial time whether a
(deterministic) policy $\tau$ for the min player in a max-minPPS is LDF.
Recall, from Section \ref{sec:zero-detect},
what it means for a policy $\tau$ for the min player to be LDF.

We borrow the concept of a {\em closed set},
studied in \cite{CY98}, which we adapt
for maxPPSs as follows:

\begin{definition} A {\em closed set}  of a maxPPS, $x=P(x)$, is a subset of variables $S$ such that: 
(i) the dependency subgraph induced by $S$ 
is nontrivial (i.e., contains at least one edge,
but not necessarily more than one variable), 
and is strongly connected (i.e., every variable in $S$ depends on every 
variable in $S$
{\em via} a directed path going only through variables in $S$);
(ii) $S$ contains only variables of type ${\mathtt M}$ and ${\mathtt L}$; and
(iii) for all variables $x_i$ in $S$ of type {\tt L},  $P(x)_i$ contains only variables in
$S$, and furthermore $P_i({\mathbf 0})=0$ and $P_i(\vone)=1$.
\end{definition}

\begin{lemma} \label{ldf-criterion}
A policy $\tau$ 
for the min player is LDF if and only if the maxPPS $x=P_{*,\tau}(x)$ contains no closed sets.
\end{lemma}
\begin{proof}
To prove the $(\Rightarrow)$ direction,
suppose that $S$ is a closed set. Since $S$ is strongly-connected, every variable $x_i$ of form {\tt M$_{\max}$} 
must have a choice in $S$, because otherwise no variable in $S$ would depend on it via $S$. 
Let $\sigma$ be a policy that picks a choice in $S$ for every variables $x_i$ in $S$. 
No variable in $S$ in $x=P_{\sigma,\tau}(x)$ depends on any variable outside of $S$. So there must be a bottom SCC, $T$, of $x=P_{\sigma,\tau}(x)$ with $T \subseteq S$. $T$ is a linear degenerate SCC, since it contains 
no variables of form {\tt Q} and has $(P_{\sigma,\tau})_T(\vzero)=\vzero$ and $(P_{\sigma,\tau})_T(\vone)=\vone$. 
So $\tau$ is not LDF.

To prove the $(\Leftarrow)$ direction, 
suppose that $\tau$ is a policy for min in $x=P(x)$ which is not LDF, i.e., 
there exists a $\sigma$ 
such that there is a bottom SCC $S$ of $x=P_{\sigma,\tau}(x)$ that is linear degenerate. 
We claim that $S$ is a closed set in $x=P_{*,\tau}(x)$. It is strongly-connected because it is an SCC of 
$x=P_{\sigma,\tau}(x)$. It contains no variables of form {\tt Q} and every variable 
of form {\tt L} satisfies (iii).
\end{proof}

\begin{lemma} \label{ldf-alg}
Given a max-minPPS $x=P(x)$ and a policy $\tau$ for the min player, we
can determine in linear time whether $\tau$ is LDF.
\end{lemma}
\begin{proof}
By Lemma \ref{ldf-criterion}, $\tau$ is LDF if the maxPPS $x=P_{*,\tau}(x)$ 
contains no closed sets.
A P-time algorithm was given in \cite{CY98} to find
the maximal closed subsets (called the closed components) of a finite state MDP
(and an improved algorithm was given in \cite{Chat-Hen14}).
These algorithms could be readily adapted to our setting,
in order to compute all maximal closed subsets of the maxPPS.
However, our problem is simpler here: we only need to determine
if there is any closed set.
We can test this condition more directly as follows.
Let $G$ be the dependency graph of the maxPPS $x=P_{*,\tau}(x)$.
Note that the variables $x_i$ of type Min have become now type L variables
in the maxPPS $x=P_{*,\tau}(x)$ and the corresponding polynomial $(P_{*,\tau})_i(x)$
satisfies $(P_{*,\tau})_i(\vzero)=0$ and $(P_{*,\tau})_i(\vone)=1$.

Perform AND-OR reachability on $G$, where the set $T$ of target nodes
includes all nodes (variables) of type Q and all nodes $x_i$ of type L
where $P_i(\vzero)>0$ or $P_i(\vone)<1$; the set of OR nodes consists of
all type L variables $x_i$ where $P_i(\vzero)=0$ and $P_i(\vone)=1$
(this includes all type Min nodes); and the set of AND nodes consists
of all type Max variables.
Recall, from the second paragraph of Section
\ref{sec:qual-gfp-equal-1}, the definition of the
set of nodes that can {\em and-or reach} the set $T$. 
Let $U$ be the set of nodes that {\em cannot} AND-OR reach the set $T$ of target nodes.
We claim that $\tau$ is LDF if and only if $U$ is empty.

Suppose first that $\tau$ is not LDF. 
Then the maxPPS $x=P_{*,\tau}(x)$ contains a closed set $S$.
By the definition of a closed set, every type Max node of $S$ has a successor in $S$ (because $S$ is strongly connected), and every type L node $x_i$ of $S$ has all its successors in $S$ and satisfies $P_i(\vzero)=0$ and $P_i(\vone)=1$.
Therefore, when we perform AND-OR reachability, no node of $S$ will
be accessed, i.e., no node of $S$ can AND-OR reach the set $T$ of target nodes.
Hence $S \subseteq U$ and thus $U$ is not empty.

On the other hand, suppose that $U$ is not empty, and let $S$ be a bottom
SCC of the subgraph $G[U]$ of $G$ induced by $U$.
Then $S$ satisfies the conditions of a closed set.
Hence $\tau$ is not LDF.
\end{proof}

We can show now the main result of this section.

\begin{theorem} The problem of approximating the GFP of a max-minPPS $x=P(x)$, i.e. 
computing a vector $\tilde{g} \in [0,1]^n$ such that $\|g^* - \tilde{g}\|_\infty \leq \epsilon$, is in TFNP.  
 \end{theorem}
 \begin{proof}
We first compute in polynomial time,  by Proposition \ref{lem:prob1-ptime}, 
the set of indices $D = \{ i \in [n] \mid  g^*_i = 1 \}$.
We then eliminate all variables $x_i$ such that $i \in D$
from the max-minPPS, substituting them by the value $1$
and removing their corresponding equations.

So, assume henceforth that $g^* <\vone$.
By Lemma \ref{lem:max-minpps-cases-g-less-1},
 for any deterministic LDF policy $\tau$ for min, and deterministic policy
$\sigma$ for max, we have that $g^*_{\sigma,*} \leq g^* \leq q^*_{*,\tau}$. Furthermore, by Lemma \ref{lem:max-minpps-cases-g-less-1} and
Corollary \ref{cor:opt-maxPPS},
there exist such policies 
which make both of these inequalities tight. Thus,
to  put the problem of approximating $g^*$ in TFNP, it suffices to guess such policies with 
$g^*_{\sigma,*}$ and $q^*_{*,\tau}$ close enough to each other, approximate the two vectors,
and verify that they are close.   

In more detail, the algorithm is as follows.
Guess deterministic policies $\sigma$ and $\tau$ for the max and min players.  Check whether $\tau$ is LDF (in P-time,
by Lemma \ref{ldf-alg}). 
If it is not, we reject this guess.
Otherwise, using our P-time algorithm for computing $g^*$ for minPPSs,  together with the P-time algorithm for computing $q^*$ for maxPPSs
from \cite{esy-icalp12},
we compute approximations $v_\sigma$ and $v_\tau$ to $g^*_{\sigma,*}$ and $q^*_{*,\tau}$ from below, such that 
$\|v_\sigma-g^*_{\sigma,*}\|_\infty \leq \epsilon/2$ and $\|v_\tau-q^*_{*,\tau}\|_\infty \leq \epsilon/2$. 
Check whether $\|v_\sigma-v_\tau\|_\infty \leq \epsilon/2$.
If so, then output $\tilde{g}=v_\sigma$;
otherwise, we reject this guess.

We have to show that the algorithm is sound and complete: 
(1) There is at least one guess for which the algorithm produces an output, and
(2) For every guess $\sigma, \tau$ for which the algorithm produces an
output, the output $\tilde{g}=v_\sigma$ is within $\epsilon$ of $g^*$.

For claim (1), consider a deterministic policy $\sigma$ for the max player and
deterministic LDF policy $\tau$ for the min player such that
$g^*_{\sigma,*} = g^* = q^*_{*,\tau}$.
The algorithm computes values $v_\sigma \leq g^*_{\sigma,*} = g^*$ and
$v_\tau \leq q^*_{*,\tau} =g^*$ such that
$\|v_\sigma-g^*\|_\infty \leq \epsilon/2$ and $\|v_\tau-g^*\|_\infty \leq \epsilon/2$.
Therefore, $\|v_\sigma-v_\tau\|_\infty \leq \epsilon/2$, the algorithm accepts the
guess and outputs $\tilde{g}=v_\sigma$.

For claim (2), suppose that the algorithm accepts a guess $\sigma, \tau$
and outputs $\tilde{g}=v_\sigma$. 
Since $v_\sigma \leq g^*_{\sigma,*} \leq g^* \leq q^*_{*,\tau}$, we have:

\begin{eqnarray*}
	\| g^*  - v_\sigma \|_\infty & \leq &  \| q^*_{*,\tau} - v_{\sigma} \|_\infty \\
	& \leq & \| q^*_{*,\tau} - v_{\tau} \|_\infty  + \| v_{\tau} - v_\sigma \|_\infty\\
	& \leq &  \frac{\epsilon}{2}+\frac{\epsilon}{2}\\
	& = & \epsilon.
\end{eqnarray*}

\end{proof}

\noindent {\bf Acknowledgements.}   Research partially supported
by the Royal Society, and by NSF grant CCF-1320654.

\bibliographystyle{plain}

\begin{thebibliography}{10}

%
%
%
%

%
%
%
%

\bibitem{Bozic13}
Bozic, et. al. 
\newblock  Evolutionary dynamics of cancer in response to targeted combination therapy.  {\em Elife}, volume 2, pages e00747, 2013.

\bibitem{BKL14}
R. Bonnet, S. Kiefer, A. W. Lin:
\newblock Analysis of Probabilistic Basic Parallel Processes. 
In {\em Proc. of FoSSaCS'14}, pages 43-57, 2014.


%
%
%
%
%


\bibitem{BBKO11}
T.~Br{\'a}zdil, V.~Brozek, A.~Kucera, J.~Obdrz{\'a}lek:
\newblock tive reachability in stochastic BPA games. 
{\em Inf. Comput.}, 209(8): 1160-1183, 2011.


\bibitem{BBFK08}
T.~Br{\'a}zdil, V.~Brozek, V.~Forejt, and A.~Kucera.
\newblock Reachability in recursive markov decision processes.
\newblock {\em Inf. Comput.}, 206(5):520--537, 2008.


\bibitem{Chat-Hen14}
K.~Chatterjee and M.~Henzinger.
\newblock Efficient and Dynamic Algorithms for Alternating B{\"{u}}chi
Games and Maximal End-Component Decomposition.
\newblock {\em Journal of the ACM}, 61(3):15:1--15:40, 2014.


\bibitem{CDK12}
T. Chen, K. Dr\"{a}ger, S. Kiefer:
\newblock Model Checking Stochastic Branching Processes. 
\newblock In Proc. of {\em MFCS'12}, Springer LNCS 7464, pages 271-282, 
2012.


\bibitem{CY98}
C.~Courcoubetis and M.~Yannakakis.
\newblock Markov decision processes and regular events.
\newblock {\em IEEE Trans. on Automatic Control}, 43(10):1399--1418, 1998.


%
%
%
%

%
%
%
%


\bibitem{denrot05a}
E.~Denardo and U.~Rothblum.
\newblock Totally expanding multiplicative systems.
\newblock {\em Linear Algebra Appl.}, 406:142--158, 2005.

\bibitem{EGKS08}
J.~Esparza, T.~Gawlitza, S.~Kiefer, and H.~Seidl.
\newblock Approximative methods for monotone systems of min-max-polynomial
  equations.
\newblock In {\em Proc. of 35th ICALP (1)}, pages 698--710, 2008.

%
%
%
%

\bibitem{EKM}
J.~Esparza, A.~Ku\v{c}era, and R.~Mayr.
\newblock Model checking probabilistic pushdown automata.
\newblock {\em Logical Methods in Computer Science}, 2(1):1 -- 31, 2006.

\bibitem{ESY12}
K.~Etessami, A.~Stewart, and M.~Yannakakis.
\newblock Polynomial-time algorithms for multi-type branching processes and
  stochastic context-free grammars.
\newblock In {\em Proc. 44th ACM Symposium on Theory of Computing (STOC)},
  2012.  (All references are to the full preprint Arxiv:1201.2374v2 (*version 2*),
available at: {\tt http://arxiv.org/abs/1201.2374v2}\ .)





\bibitem{esy-icalp12}
K.~Etessami, A.~Stewart, and M.~Yannakakis.
\newblock Polynomial-time algorithms for Branching Markov Decision Processes,
and probabilistic min(max) polynomial Bellman equations.
\newblock In {\em Proc. 
39th Int. Coll. on Automata, Languages and Programming 
(ICALP)},
  2012.
\newblock (All references are to the full preprint 
Arxiv:1202.4789v2 (*version 2*), available at:
{\tt http://arxiv.org/abs/1202.4798v2}\ .)




\bibitem{EWY08}
K.~Etessami, D~Wojtczak, and M.~Yannakakis.
\newblock Recursive stochastic games with positive rewards.
\newblock In {\em Proc. of 35th ICALP (1)}, volume 5125 of {\em LNCS}, pages
  711--723. Springer, 2008.
\newblock see full tech report at {\tt
  http://homepages.inf.ed.ac.uk/kousha/bib\_index.html}.

%
%
%
%
%

\bibitem{rmdp}
K.~Etessami and M.~Yannakakis.
\newblock Recursive {M}arkov decision processes and recursive stochastic games.
\newblock {\em Journal of the ACM}, 62(2):1--69, 2015.
%
%
%
%
%
%
%
%
%

\bibitem{rmc}
K.~Etessami and M.~Yannakakis.
\newblock Recursive {M}arkov chains, stochastic grammars, and monotone systems
  of nonlinear equations.
\newblock {\em Journal of the ACM}, 56(1), 2009.

\bibitem{Harris63}
T.~E. Harris.
\newblock {\em The Theory of Branching Processes}.
\newblock Springer-Verlag, 1963.

\bibitem{HornJohnson85}
R.~A. Horn and C.~R. Johnson.
\newblock {\em Matrix Analysis}.
\newblock Cambridge U. Press, 1985.

%
%
%
%
%
%

\bibitem{KomBol13}
N. L. Komarova and C. R. Boland.
\newblock 
Cancer: Calculated treatment.
\newblock {\em Nature}, volume 499, pages 291-292. (News and Views article.) 

\bibitem{pliska76}
S.~Pliska.
\newblock Optimization of multitype branching processes.
\newblock {\em Management Sci.}, 23(2):117--124, 1976/77.

%
%
%
%

\bibitem{RBCN13}
G.~Reiter, I.~Bozic, K.~Chatterjee, M. A.~Nowak.
\newblock TTP: Tool for tumor progression. 
\newblock  In {\em Proc. of CAV'2013}, pages 101-106, Springer LNCS 8044, 2013.


\bibitem{rotwhit82}
U.~Rothblum and P.~Whittle.
\newblock Growth optimality for branching {M}arkov decision chains.
\newblock {\em Math. Oper. Res.}, 7(4):582--601, 1982.

%
%
%
%

%
%
%
%
%


\end{thebibliography}
\end{document}